\documentclass[10pt]{article}
\usepackage{latexsym}
\usepackage[english]{babel}
\usepackage{graphicx}
\usepackage{xcolor}
\usepackage{epstopdf}
\usepackage{amssymb}
\usepackage{amsmath,amsthm,enumerate,rotate}
\usepackage{algorithm}
\usepackage{algpseudocode}
\usepackage{hyperref}

\newcommand{\R}{\mathbb R}

\def\be#1\ee{\begin{equation}#1\end{equation}}

\newtheorem{proposition}{\bf Proposition}[section]

\newcommand{\bq}{\begin{equation}}
\newcommand{\eq}{\end{equation}}

\def\bqa{\begin{eqnarray}}
\def\eqa{\end{eqnarray}}

\def\e{\epsilon}

\def\bQ{{\bf Q}}
\def\bbf{{\bf f}}

\def\bP{{\bf K}}

\newcommand{\bd}{\begin{displaymath}}
\newcommand{\ed}{\end{displaymath}}
\newcommand{\ba}{\begin{eqnarray}}
\newcommand{\ea}{\end{eqnarray}}

\def\R{\mathbb{R}}

\newenvironment{equations}{\equation\aligned}{\endaligned\endequation}

\begin{document}



\title{Impact of opinion formation phenomena in epidemic dynamics: kinetic modeling on networks}

\author{G. Albi\thanks{Department of Computer Science,
		University of Verona, Strada le Grazie 15, 37134 Verona, Italy. (giacomo.albi@univr.it)}\and E. Calzola \thanks{Department of Computer Science,
		University of Verona, Strada le Grazie 15, 37134 Verona, Italy. (elisa.calzola@univr.it)}  \and
G. Dimarco\thanks{Department of Mathematics and Computer Science \& Center for Modeling, Computing and Statistics (CMCS), University of Ferrara, via Machiavelli 30, 44121 Ferrara, Italy. (giacomo.dimarco@unife.it)} \and M. Zanella\thanks{Department of Mathematics, University of Pavia, via Ferrata 5, 27100 Pavia, ITALY. (mattia.zanella@unipv.it)}}



\maketitle

\begin{abstract}
After the recent COVID-19 outbreaks, it became increasingly evident that individuals' thoughts and beliefs can have a strong impact on disease transmission. It becomes therefore important to understand how information and opinions on protective measures evolve during epidemics. To this end, incorporating the impact of social media is essential to take into account the hierarchical structure of  these platforms. 
In this context, we present a novel approach to take into account the interplay between infectious disease dynamics and socially-structured opinion dynamics. Our work extends a conventional compartmental framework including behavioral attitudes in shaping public opinion and promoting the adoption of protective measures under the influence of different degrees of connectivity. 
The proposed approach is capable to reproduce the emergence of epidemic waves. Specifically, it provides a clear link between the social influence of highly connected individuals and the  epidemic  dynamics.
Through a heterogeneity of numerical tests we show how
this comprehensive framework offers a more nuanced understanding of epidemic dynamics in the context of modern information dissemination and social behavior.
\end{abstract}
\tableofcontents

\section{Introduction}\label{sec:intro}
Several mathematical approaches have provided valuable insights to understand the dynamics of infectious diseases. Traditional compartmental models are commonly based on the subdivision of the population into epidemiologically relevant classes, whose size changes in time through the introduction of suitable transition rates, see e.g. \cite{DH,HWH00,kermack1927} and the references therein. 
Such ideas have been further extended during the recent pandemic to mimic more closely realistic features of the SARS-CoV-2 infection, in this direction we mention  \cite{buonomo2020,Bruno,kantner2020, XLiu, peirlinck2020,lolipiccolomini2020,tang2020}.

However, classical models are typically not derived from first principles and are able to describe the temporal evolution of the epidemic spread only in terms of the average numerical density of individuals in each compartment, thus neglecting other relevant mechanisms related between-hosts dynamics, such the spatial movement \cite{SBKT,barbera2013,bertaglia2021,boscheri2020,colombo2020, Das24,Dutta2020, Sun,viguerie2021} and their heterogeneities \cite{BBT,GFMDC12,Pugliese,Bere}, and in-host dynamics which provides a temporal change of infectivity due to environmental factors \cite{bondesan_etal,DMLT,LGT}. We mention other important factors to take into account, like the agents attitude in relation to protective measures and non-pharmaceutical interventions \cite{KGFPS, fake2,Zanella_m3as}, together with other social features characterising the transmission dynamics, like wealth and number of  social contacts among others, see \cite{Barth, Bon, DPTZ, DPeTZ, DTZ, Nielsen, xiang18}. Uncertainty and data driven approaches also plays a relevant role and it has been discussed and addressed in several works \cite{APZ, APZ2,bertaglia2021b, bertaglia2021a,FMZ24,Rob}.

In this direction, the recent COVID-19 pandemic highlighted the need to consider not only biological and environmental factors but also the impact of the social reaction shaped by the diffusion of the disease  \cite{AM,Block,Bon,Bonde}. Among the most critical social aspects, the formation of consensus about protective measures, such as wearing masks, social distancing, and vaccination plays undoubtedly a pivotal role. How such interventions can mitigate the impact of the disease in a population, also in relation with the spread of fake news \cite{ fake2,fake1} is also a direction that is worth studying. In this work, we will further explore the kinetic approach presented in \cite{bonandin,Zanella_23}, where opinion polarization dynamics connected with the individuals response to threat is coupled with epidemiological transitions. In particular, since opinion formation is heavily influenced by social media platforms, where information and misinformation can spread rapidly and widely, we incorporate opinion dynamics in a population of agents with a connectivity structure based on the the approach presented in \cite{Albi_Calzola}. Indeed, social-media platforms can shape public perceptions and behaviors significantly, as individuals often rely on shared content to form their opinions \cite{cinelli2020}. The role of social media in this context is particularly pronounced when discussing about non-pharmacological interventions, where public adherence to recommended behaviors can drastically alter the course of an epidemic. Among the myriad of voices on social media, certain individuals, commonly referred to as influencers, also play a pivotal role in shaping public opinion. These influencers can drive individual and collective behavior towards specific actions, whether beneficial or detrimental \cite{Albi_Calzola}, and their impact on public health measures can be profound, as they can either encourage widespread adoption of protective behaviors or fuel resistance and skepticism. The resulting interactions will couple behavioral characteristics with epidemiological models to shed light on the role of highly connected individuals in shaping opinion-driven infectious dynamics. 

Starting from the above consideration and inspired by the recent development concerning kinetic models for consensus-type phenomena  \cite{Albi5,Albi4,  DFWZ, DW_15,LRT,PT13,TTZ}, we develop an integrated approach to connect the formation of opinion about preventive measures and the spreading of a disease in a multi-agent system. More in details, we first introduce, through the use of probabilistic and statistical tools, a microscopic dynamics of opinion formation with the presence of leaders and then we upscale such dynamics at the bound of observable quantities, see \cite{Cer, CPS, FPTT}, such as the number of infected. Even if, in the sequel, we concentrate ourselves on a specific compartmental model for the disease transmission, namely the SEIR dynamics, we stress that the ideas here described are not only linked to that model which can be intended as an example of how it is possible to integrate different aspects of human behavior into epidemiological dynamics. Concluding, through this presented approach, we seek to explore the interplay between disease dynamics and social behavior, providing a more comprehensive tool for managing current and future epidemics.

The rest of the work is structured as follows. In Section \ref{sec:epidyn} we present the integrated compartmental model which takes into account the opinion of the agents shared on social media, whose evolution is due two different mechanisms: a one to one interaction between agents and the influence from the social background acting on each single individual. This latter is supposed, in turns, to depend on the trend of the disease within the population. In the same part, we compute a so-called Fokker-Planck asymptotics for the process of opinion formation in presence of social network contacts which permits to recover an equation governing the time evolution of the joint density of opinions and contacts. This allows us to analytically compute, under suitable hypothesis, the steady state of the opinion distribution. Section \ref{sec:macroepi} analyzes the evolution of a surrogate kinetic epidemiological model macroscopic observable quantities. Section \ref{sec:numerics} presents several numerical simulations that show the different disease evolution for various scaling and choices of the model parameters, and we also show the effect of influent agents in shaping the public's opinion and, possibly, modifying the course of the disease. A last Section  
permits to draw some conclusions and discuss some future investigations.

\section{A new model of opinion-driven epidemic dynamics with connectivity structure}\label{sec:epidyn}
The recent COVID-19 outbreaks have shown that personal opinions on the effectiveness of protective masks and social distancing play a pivotal role in shaping the evolution of the disease. To that aim, in this section, we introduce a new compartmental model that incorporates crucial mechanisms capable of modifying the spread of an epidemic, namely individuals' opinions on protective measures and the influence of key figures on large groups. In the following, we will adopt the approach presented in \cite{Zanella_23} to embed consensus formation phenomena into an epidemic model. 

To that aim, let $c >0$ be the number of social media followers (contacts) of an agent and $v\in\Omega =[-1,1]$ the agent's opinion about the importance of using protective measures and/or to keep social distance whenever necessary. Within the above choice, we mean that when $v\equiv -1$ the agent is assumed to be totally against the use of any protection, while the case $v\equiv1$ means that an individual strongly believes on the necessity of protection and acts consequently.
In the following, we also identify the so-called influencers as the agents with a very large number of social contacts $c \gg 1$. We will stick to a compartmentalization  based on the classical Susceptible-Exposed-Infectious-Recovered point of view (SEIR from now on) for what concerns the dynamics of an infectious disease \cite{HWH00}. Such model divides the population into four distinct compartments, each representing a different stage in the progression of the disease: the Susceptible (S), individuals that are healthy and not infected, and that have the possibility to contract the disease if they come into contact with infectious individuals, the Exposed (E) individuals, who have been exposed to the pathogen but are not yet infectious themselves, the Infectious (I) which are both infected and capable of transmitting the disease to susceptible individuals, the Recovered or Removed (R) individuals, who can no longer spread the disease.

Let now $f_J(v,c,t)$, $J \in \mathcal{C}$, $\mathcal C = \{S,E,I,R\}$, denote the distribution of individuals with a given opinion $v$ and a given number of contacts $c$ at time $t>0$ belonging to the $J$th compartment. In this setting, the compartment identifies the individual status with respect to the pathogen agent. We then define $f(v,c,t)$ through
\[
\sum_{J \in \mathcal C}  f_J(v,c,t) = f(v,c,t), \qquad \int_{\Omega \times\mathbb R_+}f(v,c,t)\mathrm{d}v\mathrm{d}c =1. 
\]
to be the joint density of opinions and contacts at time $t \geq 0$ for the population irrespectively of their status with respect to the pathogen agent while  we assume that the total mass of individual is normalized. The mass fraction of the population in each compartment is therefore obtained as
\[
\int_{\Omega \times \mathbb{R}_+} f_J(v,c,t) \mathrm{d}v\,\mathrm{d}c = \rho_J(t),
\]
the mean opinion and the mean number of social contacts at time $t\geq 0$ as
\[
\rho_J m_{J,v}(t) = \int_{\Omega \times \mathbb{R}_+} v f_J(v,c,t) \mathrm{d}v \mathrm{d}c, \quad\mbox{and}\quad
\rho_J m_{J,c}(t) = \int_{\Omega \times \mathbb{R}_+} c f_J(v,c,t) \mathrm{d}v \mathrm{d}c,
\]
respectively. 
In this setting, the time evolution of the functions $f_J(v,c,t)$, $J\in \mathcal D$ is obtained by integrating the compartmentalization epidemiological description with the simultaneous characterization of the formation of a social network and opinions in a society \cite{DPTZ,Zanella_23}. This merging can be expressed by the system
\begin{equation}
	\label{new}
	\partial_t \bbf(v,c,t) = \bP(v,c,\bbf(v,c,t)) + \frac{1}{\tau_1} \bQ^1(\bbf(v,c,t))+\frac{1}{\tau_2} \bQ^2(\bbf(v,c,t)).
\end{equation}
In the above expression, $\bP$ is the $n$-dimensional vector whose components $K_J(v,c,\bbf)$, $J\in \mathcal C$, are the transition rates between compartments (now depending on the social contact variable $c$, on the opinion $v$ and on the vector $\bbf(v,c,t)=(f_J(v,c,t))_{J\in \mathcal{C}} $, while $\bQ^1$ and $\bQ^2$ are $n$-dimensional vectors whose components $Q^1_J(f_J)$, $Q^2_J(f_J)$, $J\in \mathcal C$ are suitable interaction operators describing, for the underlying compartments, the formation of the opinions in the presence of a social network under two different mechanisms which will be detailed later. Last, $\tau_1>0$ and $\tau_2>0$ measure the time-scales at which these interactions take place with respect to the time evolution of the epidemic. 
The dynamics encapsulated in the model incorporate a broader range of influences, particularly the socio-behavioral aspects that can significantly alter the course of an epidemic. By integrating individuals' opinions on protective measures and the role of influential figures in shaping public behavior, this model, as shown later, is able to capture the feedback loop between disease spread and human behavior. This approach allows for the examination of how social attitudes and communication networks can amplify or mitigate the spread of an infection. By assuming, as opposite, that the pathogen transmission does not depend upon the opinion with respect to protective measures which, in turn, depends on the social status, it is easy to see that the proposed description \eqref{new} collapses to a standard compartmentalization approach. 

Let observe that, in this work, our aim is to stress the role played by opinion and social network. For this reason, we do not consider other effects such as the physical contacts in our model and thus when referring to a contact in the sequel, we simply mean a relation which permits to exchange ideas between individuals. 
Given the general modeling ideas described in \eqref{new}, we now consider a specific compartmental model which we will then use in the rest of the paper for describing the time evolution of the outbreak. 


Let now $f_J(v,c,t)$,  with $J \in \mathcal{C}$, following the general approach \eqref{new}. Then, the time evolution for the distributions $f_J(v,c,t)$ is governed by the following set of differential equations
\begin{equation}\label{eq:seiropicont}
	\begin{array}{rcl}
		\partial_t f_S  &=&\displaystyle -K(f_S,f_I)(v,c,t) + \frac{1}{\tau_1}\sum_{H \in \mathcal{C} }Q^1_S(f_S,f_H)(v,c,t)+\frac{1}{\tau_2} Q^2_S(f_S,\rho_I)(v,c,t),\\[12pt]
		\partial_t f_E  &=&\displaystyle K(f_S,f_I)(v,c,t) - \sigma_E f_E + \frac{1}{\tau_1}\sum_{H \in \mathcal{C}}Q^1_E(f_E,f_H)(v,c,t)+ \frac{1}{\tau_2}Q^2_E(f_E,\rho_I)(v,c,t),\\[12pt]
		\partial_t f_I &=&\displaystyle \sigma_E f_E - \gamma f_I +  \frac{1}{\tau_1}\sum_{H \in \mathcal{C}} Q^1_I(f_I,f_H)(v,c,t)+\frac{1}{\tau_2} Q^2_I(f_I,\rho_I)(v,c,t),\\[12pt]
		\partial_t f_R &=&\displaystyle \gamma f_I + \frac{1}{\tau_1}\sum_{H \in \mathcal{C}}Q^1_R(f_R,f_H)(v,c,t) + \frac{1}{\tau_2} Q^2_R(f_R,\rho_I)(v,c,t),
	\end{array}
\end{equation}
where $\sigma_E > 0$  is such that $1/\sigma_E$ is the measure of the latent period of the disease and $\gamma > 0$  the recovery rate. The transmission of the disease is ruled by the local incidence rate given by
\begin{equation}\label{eq:sincidence}
	K(f_S,f_I)(v,c,t) = f_S(v,c,t)\int_{\Omega \times \mathbb{R}_+} \kappa(v,v_*;c,c_*)f_I(v_*,c_*,t)\mathrm{d}c_*\mathrm{d}v_*,
\end{equation}
with $\kappa(v,v_*;c,c_*)$ a nonnegative function that measures the impact of the protective attitudes between the different groups. Referring to \cite{Zanella_23}, an example for the function $\kappa(v,v_*;c,c_*)$ can be obtained assuming 
\begin{equation}\label{eq:sincidence_ex}
	\kappa(v,v_*;c,c_*) = \frac{\beta_1}{4^{\beta_2}}(1-v)^{\beta_2}(1-v_*)^{\beta_2},
\end{equation}
where $\beta_1>0$ is the baseline transmission rate of the epidemics and $\beta_2>0$ is a coefficient that depends on the efficacy of the protective measures. The above choice reflects the fact that the more an individual is against protective measures, meaning $v\approx-1$, the larger is the possibility to get infected, this effect being amplified if both individuals getting in contact share the same negative opinion. The opposite situation verifies when either $v\approx1$ or $v_*\approx1$ meaning that, in this case, the probability of getting infected is lower and tends towards zero. Let us observe that in \eqref{eq:sincidence_ex}, there is not any explicit dependence on the social contacts. This is natural since as previously explained $c$ does not measure the physical contacts but instead the size of the network of each person in a community. Thus, this information does not modify directly the local incidence rate $\kappa(v,v_*;c,c_*)$, however its effect is implicitly present in the process of opinions $(v,v_*)$ formation. In fact, as it will be precised later, the social network and the positioning of the influencers with respect to protective measures and social distancing causes modifications of the distribution of opinions among the population. Before ending this section, let us observe that in the simple case $\beta_2 =1$ one gets the following form for the incidence rate
\begin{equation}\label{eq:K}
	K(f_S,f_I)(v,c,t) = f_S(v,c,t)\frac{\beta_1}{4}(1-v)(1-m_{I,v}(t))\rho_I(t) \geq 0,
\end{equation}
meaning that the probability of infection depends upon the average opinion of the infectious compartment $I$: the larger is this latter, the larger becomes the probability of getting infected.


\subsection{Kinetic models for opinion formation dynamics}\label{sec:opiformation}
We now focus on the opinion formation. We associate the opinion of each agent on protective measures and/or social distance to a variable $v \in \Omega = [-1,1]$, as well as a number of contact $c\in\R_+$ representing her/his connectivity within a social-network. We start by looking at the one to one interactions among individuals. As mentioned before (Section \eqref{sec:intro}), we model two different interaction operators. The operator $Q_J^1(f_J,f_H)$ takes into account the interaction between two agents belonging to two compartments $J, H\in \mathcal{C}$. The post-interaction opinions-contact $(v^\prime,c,v_*^\prime,c) \in (\Omega\times \R_+)\times(\Omega\times \R_+)$ are defined through the binary scheme  
\begin{equation}
	\begin{split}
		\label{eq.trules}
		v^\prime &     =  v + \alpha_J P(v,v_*,c,c_*)(v_*-v) + \xi^{JH} D(v),\\
		c^\prime &= c, \\
		v^\prime_* & =  v_* + \alpha_{H} P(v_*,v,c_*,c)(v-v_*) + \xi^{HJ} D(v_*), \\
		c_*^\prime &= c_*,
	\end{split}
\end{equation}
where contacts  $c,c_*$ are assumed to be stationary.
The introduced scheme encapsulates two different processes:  compromise, weighted by the propensity function $P(\cdot,\cdot,\cdot,\cdot) \in [0,1]$ with $\alpha_J,\alpha_H>0$, and self-thinking, obtained through the centered independent and identically distributed random variables $\xi^{JH}$, $J,H \in \mathcal{C}$, with finite variance and weighted by the function $D(v)\ge0$. For all $J,H \in \mathcal{C}$, the random variables $\xi^{JH}$  are such that $\langle \xi^{JH}  \rangle = 0$ and with finite variance $\langle (\xi^{JH})^2 \rangle = \langle (\xi^{HJ})^2 \rangle =\sigma_{JH}^2$. Let us observe that in \eqref{eq.trules}, the post-interaction {$(v^\prime,v_*^\prime)$} opinions depend upon the number of connections $(c,c_*)$ of both participants to the interaction. In detail, the precise choice of the {interaction} function $P(\cdot)$ will reflect the fact that, the larger is the network of an individual, the larger is the influence that this individual has on other people and, correspondingly, the smaller is the influence that other individuals have on him/her, such as
\begin{equation}\label{eq:kernel_P}
	P(v,v_*,c,c_*) = \chi(|v-v_*|\leq \Delta(c,c_*))W(c,c_*),
\end{equation}
where $\chi(\cdot)$ is the indicator function, where $\Delta(c,c_*)\in[0,2]$, and $W(c,c_*)\in[0,1]$ a weighting function.  This choice consists in a bounded confidence model for $P$, where the threshold and the magnitude of the interactions are weighted by different measures of the relative number of contacts $(c,c_*)$. More specific forms of the function $P(\cdot)$  obeying to the above general principle will be discussed later in the numerical Section \ref{sec:numerics}.

The second {interaction operator $Q_J^2(f_J,\rho_I)$ takes into account the opinion changes due to the interaction with the epidemic state. } 
This is modeled by a {linear interaction scheme of the type proposed in} \cite{Bonde}. In particular, we consider 
\begin{equation}\label{eq:opinion_infected}
	v'' =  v + \tilde \alpha_J \tilde P(\rho_I,v,c)\left(G(\rho_I) - v\right) + \zeta^J \tilde D(v),\\
\end{equation}
where $\tilde \alpha_J$ is a positive real number, which in general depends on the compartment, $\tilde P:(0,1) \times\Omega \times (0,\infty) \to [0,1]$ is a compromise function that, in general, depends  on the opinion and on the number of connections of each agent.  We assume $\tilde P$ to be a decreasing function with respect to the absolute value of the first variable: the more extremal is the agent's opinion, the smaller is the compromise effect. The function $\tilde P$ is also assumed to be decreasing with respect to the second variable, namely influencers have strong opinions and they are less inclined to change idea through interaction with the rest of the population.
We can exemplify this function as follows
\begin{align}\label{eq:kernel_Ptilde}
	\tilde P(\rho_I,v,c) &= \psi(\rho_I)(1-v^2)\tilde W(c,\bar c),
\end{align}
where $\psi(\cdot)$ is a non decreasing function of $\rho_I$, e.g. $\psi(\rho_I)=\rho_I^r,\, r\geq 0$, measuring the impact of the background, $\tilde W(\cdot)$ is a weighting function and $\bar c$ is a fixed number of contacts.

The function $G(\cdot)$ represents the average population opinion, namely the one conveyed by the media, information channels and the government. It is defined in such a way that the higher the number of infected is, the stronger is the convincement of individuals towards protective measures. Conversely, as $\rho_I(t)$ decreases, $G(\cdot)$ moves closer to the value of $-1$: the population becomes less inclined towards social distances and the use of protective masks. For example we can consider
\begin{align}\label{eq:Gfunction}
	G(\rho_I) &= 2\chi(\rho_I\geq\overline \rho)-1,
\end{align}
with  $\overline \rho\in[0,1]$ an alert parameter based on a certain percentage of infected.  More specific forms following to the above general principles will be discussed later in the numerical Section \ref{sec:numerics}.
We finally consider, as for the case \eqref{eq.trules}, a certain amount of randomness to be involved in this interaction, that it is modeled by the diffusion function $\tilde D$ multiplied by a random variable $\zeta_J$ with mean $\langle \zeta_J \rangle = 0$ and variance $\langle \zeta_J^2 \rangle = \theta_J^2$. 

Since $v$ belongs to the bounded domain $[-1,1]$, it is important to only consider interactions that remain in this domain. A sufficient condition to preserve the bounds is {provided in the following result.}
\begin{proposition}
	{Let us consider, for any $c \in \mathbb R_+$  and $D(\pm 1) = 0$, a pre-interaction opinion pair such that $(v,v_*)\in\Omega \times\Omega$, then, the binary interaction scheme \eqref{eq.trules} is such that the post-interaction opinions preserve the bounds, i.e. $(v',v'_*)\in\Omega \times \Omega$, provided}
	\begin{equation}\label{bounds}
		0 < P(v,v_*,c,c_*) \leq 1, \quad 0<\alpha_J \leq 1/2, \quad |\xi^{JH}| \leq(1-\lambda^*)d
	\end{equation} 
	where
	\begin{equation}\label{bounds2}
		\lambda^* = \min_{\substack{J\in \mathcal{C}, \\v,v_* \in [-1,1], \\c,c_*>0}}  \alpha_J P(v,v_*,c,c_*), \quad d = \min_{\substack{v\in[-1,1], \\ c>0}} \left\{ \frac{1-|v|}{D(v)}, D(v) \neq 0\right\}.
	\end{equation} 
	Similarly, if $v\in[-1,1]$ and $\tilde D(\pm 1) = 0$, the binary interaction \eqref{eq:opinion_infected} preserves the bounds, i.e. $v''\in[-1,1]$, if 
	\begin{equation}\label{bounds3}
		0 < \tilde P(\rho_I,v,c) \leq 1, \quad 0< \tilde \alpha_J <1/2, \quad |\zeta^J| \leq(1-\tilde \lambda^*)\tilde d 
	\end{equation} 
	where
	\begin{equation}\label{bounds4}
		\tilde \lambda^* = \min_{\substack{J\in \mathcal{C}, \\\rho_I\in [0,1 ]\\c>0}}  \tilde \alpha_J \tilde P(\rho_I,v,c),\quad \tilde d = \min_{v\in[-1,1]} \left\{ \frac{1-|v|}{\tilde D(v)}, \tilde D(v) \neq 0\right\}.
	\end{equation} 
\end{proposition}
\begin{proof}
	We refer to \cite{Albi_Calzola} for the proof of the first statement, i.e. that \eqref{eq.trules} preserves the bounds under \eqref{bounds} and \eqref{bounds2}. The proof of the rest of the Proposition follows the same path: let $\tilde \lambda = \tilde \alpha_J\tilde P(\rho_I,v,c)$ when $\tilde D(v) = 0$ we have 
	\[
	\left| v'' \right| = \left| v + \tilde \lambda\left(G(\rho_I(t)) - v\right)  \right| \leq \left| v \right|(1-\tilde \lambda) +\tilde \lambda\left|G(\rho_I(t))\right| \leq 1
	\]
	since $\left| v \right|, \left|G(\rho_I(t))\right| \leq 1$ and $\tilde \lambda\in (0,1)$. When $\tilde D(v) \neq 0$
	\[
	\begin{array}{rcl}
		\left|v'' \right| &=& \left| v + \tilde \lambda\left(G(\rho_I(t)) - v\right) + \zeta^J \tilde D(v)\right| \leq \left| v \right|(1-\tilde \lambda) + \tilde \lambda\left|G(\rho_I(t)\right| + \left|\zeta^J \right|\tilde D(v) \\
		&\leq& \left| v \right|(1-\tilde \lambda) + \tilde \lambda + \left|\zeta^J \right|\tilde D(v)
	\end{array}
	\]
	leading to
	\[
	\left|\zeta^J \right| \leq \frac{(1-\tilde \lambda^*)(1-\left|v \right|)}{\tilde D(v)}, \quad \mbox{with}\quad\tilde \lambda^* = \min_{\substack{J\in \mathcal{C}, \\\rho_I\in [0,1 ]\\c>0}}  \tilde \alpha_J \tilde P(\rho_I,v,c)
	\]
	in order to guarantee that $\left|v'' \right|\leq 1$. 
\end{proof}

Furthermore, we observe that the interaction \eqref{eq.trules} is a contraction for $\alpha_J,\alpha_{J_*} \in(0,1/2]$. Specifically, using the shorthand notation $P=P(v,v_*,c,c_*)$ and $P_*=P(v_*,v,c_*,c)$, we compute 
\begin{equation}\label{contraction}
	\begin{split}
		|\langle v'-v'_*\rangle| &=|\langle v(1-(\alpha_J P+\alpha_{J_* }P_*)) -v_*(1-(\alpha_J P+\alpha_{J_*} P_*)) - \xi^{JH} D(v)+\xi^{HJ} D(v_*)\rangle|\\
		&=|v(1-(\alpha_J P+\alpha_{J_* }P_*)) -v_*(1-(\alpha_J P+\alpha_{J_*} P_*))|\leq | v-v_*|,
	\end{split}
\end{equation}
we have that the post-interaction opinions are on average closer than the pre-interaction ones. The same holds true for the $v''$ in \eqref{eq:opinion_infected}, namely we have
\begin{equation}\label{contraction2}
	\begin{split}
		&	|\langle v''-G(\rho_I(t))\rangle| =|\langle v(1-\tilde \alpha_J \tilde P(\rho_I(t),v,c)) -G(\rho_I(t))(1-\tilde \alpha_J \tilde P(\rho_I(t),v,c)) +\zeta^J \tilde D(v)\rangle| \\
		&\qquad \qquad=|v(1-\tilde \alpha_J \tilde P(\rho_I(t),v,c)) -G(\rho_I(t))(1-\tilde \alpha_J \tilde P(\rho_I(t),v,c)) |\leq | v-G(\rho_I(t))|.
	\end{split}
\end{equation}

Now, by introducing the probability density $f_J(v,c,t)$, the time evolution of agents with underlying binary interactions given by \eqref{eq.trules} and mean field interaction \eqref{eq:opinion_infected} is determined by the Boltzmann-like equation
\begin{equation}\label{eq:boltz_J}
	\partial_t f_J(v,c,t) = \frac{1}{\tau_1}\sum_{H \in \mathcal{C}} Q^1_J(f_J,f_H)(v,c,t)+\frac{1}{\tau_2} Q^2_J(f_J,\rho_I)(v,c,t),
\end{equation} 
with $\tau_1, \tau_2 > 0$ two distinct frequencies of interaction. The Boltzmann operator $Q^1_J$, in weak form, reads as
\begin{equation} \begin{aligned}
		\label{kine-ww}
		& \int_{\Omega \times \mathbb{R}_+} \varphi(v,c)Q^1_J(f_J,f_H)(v,c,t)\,\mathrm{d}v\,\mathrm{d}c  =\cr&\qquad\qquad\qquad\qquad \displaystyle
		\Big \langle\int_{\Omega^2\times \mathbb{R}_+^2}\bigl(\varphi(v',c)
		-\varphi(v,c) \bigr) f_H(v_*,c_*,t)f_J(v,c,t)
		\,\mathrm{d}v\,\mathrm{d}v_*\,\mathrm{d}c\,\mathrm{d}c_* \Big \rangle,
	\end{aligned}
\end{equation}
where $(v',v_*')$ are defined in \eqref{eq.trules}. The other operator $Q^2_J$, instead, can be written in weak form as
\begin{equation} 
	\label{kine-ww2}
	\int_{\Omega \times \mathbb{R}_+} \varphi(v,c)Q^2_J(f_J,\rho_I)(v,c,t)\,\mathrm{d}v\,\mathrm{d}c  =\displaystyle\Big \langle\int_{\Omega \times \mathbb{R}_+}\bigl(\varphi(v'',c) -\varphi(v,c) \bigr) f_J(v,c,t) \,\mathrm{d}v\,\mathrm{d}c \Big \rangle,
\end{equation}
with $v''$ defined in \eqref{eq:opinion_infected}. 

We stress that the interaction operators $Q^1_J$ and $Q^2_J$ do not contain any direct variation with respect to the network shape since the distribution of the social connections is stationary.
Here we will assume that the contact are distributed according to a fixed density that we will denoted by $h(c)$. In particular, we can express the joint density as  $f_J(v,c,t)  = g_J(v|c,t)h(c)$ where $g_J(v|c,t)$ denotes the conditional density of opinions for a given level of contact within the compartment $J$. Among various choice of the density $h(c)$ we are interested in capturing the structure of network in social media, 
such as the following log normal distributions  Among the various choices of density $h(c)$ we will consider a lognormal distribution of the following form
\begin{equation}\label{eq:logn}
	h(c) = \frac{1}{\sqrt{2 \pi \Gamma} c}\text{exp} \left\{ -\frac{(\text{ln} c - \Lambda)^2}{2 \Gamma}\right\}, \qquad c\in(0,+\infty), \Gamma, \Lambda >0 .
\end{equation} 
Remarkably, the shape of the distribution \eqref{eq:logn}  fits real network structures particularly well, as detailed in \cite{Albi_Calzola}. Indeed, \eqref{eq:logn} has been obtained as the stationary solution of a Fokker-Planck equation, modeling the evolution of social connections on online platforms on a social network. In particular the study done in \cite{Albi_Calzola} refers to networks on Twitter (now $\mathbb{X}$). 
Thus, the choice we have done in this work of using this model for social connections is motivated by the fact that, among various platforms, Twitter is primarily focused on information dissemination and this makes it particularly suitable for studying opinion formation, especially in relation to the spread of epidemics and the relative role of influential individuals in shaping ideas in a community.

\subsection{Derivation of Fokker-Planck-type asymptotics}\label{sect:fp}
The dynamics described in \eqref{kine-ww}-\eqref{kine-ww2} is able to characterize the formation of opinion in a society. However, in real situations, it happens that a single interaction determines only a small change of the opinion $v$. Moreover, from the mathematical point of view studying the qualitative behaviors of \eqref{kine-ww}-\eqref{kine-ww2} and consequently getting some insights on the this model under the knowledge of the network \eqref{eq:logn} is very difficult. For both reasons, in order to investigate in detail the behavior of the system, one can rely on a simplified approach obtained by rescaling both the interaction and diffusion parameters in \eqref{kine-ww}-\eqref{kine-ww2} through a quasi invariant limit as done for instance in the classical Boltzmann equation for deriving the well-known Fokker-Planck equation \cite{Cer}.

Thus, in order to model the fact that, the formation of the opinions is due to a large number of interactions, each one producing a small change in the point of view of individuals up to the moment in which the final opinion is reached, we rely on the quasi-invariant scaling
\begin{equation}\label{eq:scaling}
	\alpha_J\to\e\alpha_J,\qquad \sigma_{JH}^2  \to \e \sigma_{JH}
	^2,\quad\mbox{and}\quad
	\tilde \alpha_J\to \e\tilde \alpha_J,\qquad \theta_J^2  \to \e \theta_J
	^2,
\end{equation} 
for the small parameter $\e\ll 1$.

We assume that the scaled random variables $\xi_{\e}^J$, $\xi_{\e*}^{J_*}$ and $\zeta_{\e}^J$ are independent, 
and that the following relations hold true
\begin{equation}
	\langle\xi_{\e}^{JH}\rangle =\langle\xi_{\e}^{HJ}\rangle= 0,\quad \langle(\xi_{\e}^{JH})^2\rangle =\langle(\xi_{\e}^{HJ})^2\rangle = \e\sigma_{JH}^2,\quad  \langle(\xi_{\e}^{JH})^3\rangle=\langle(\xi_{\e}^{HJ})^3\rangle < +\infty,
\end{equation} 
and that 
\begin{equation}
	\langle\zeta_{\e}^J\rangle = 0,\quad \langle(\zeta_{\e}^J)^2\rangle = \e\theta_J^2, \quad  \langle(\theta_{\e}^J)^{3}\rangle<+\infty.
\end{equation}

Hence, we consider the time parametrization $t_\e = \epsilon t$ for the weak equations \eqref{eq:boltz_J} as follows
\begin{equation}\label{eq:boltz_J_scaled}
	\begin{aligned}
		\frac{d}{dt_\e}\int_{\Omega\times\mathbb{R}_+} f_{J,\e}(v,c,t)\varphi(v,c)\,\mathrm{d}v\,\mathrm{d}c &= \frac{1}{\e\tau_1}\sum_{H \in \mathcal{C}} \int_{\Omega \times \mathbb{R}_+} Q^1_J(f_{J,\e}f_{H,\e})(v,c,t)\varphi(v,c)\,\mathrm{d}v\,\mathrm{d}c\cr
		&\qquad\qquad+  \frac{1}{\e\tau_2}	\int_{\Omega \times \mathbb{R}_+}Q^2_J(f_{J,\e},\rho_I)(v,c,t)\varphi(v,c)\,\mathrm{d}v\,\mathrm{d}c,
	\end{aligned}
\end{equation} 
and to retrieve an asymptotic Fokker-Planck model we consider the scaling \eqref{eq:scaling} and the resulting limit for $\e \to 0$.
In what follows we describe the asymptotics related to the operators $Q_J^1$ and $Q_J^2$, and  eventually the resulting Fokker-Planck model.
\paragraph{Grazing limit for $Q^1_J$}
Using the properties of the random quantities $\xi^{JH}, \xi^{HJ}$ and \eqref{eq.trules}, we have that
by expanding the smooth function $\varphi$ in Taylor series up to order two we get
\begin{equation}\label{eq:Tayex}
	\begin{aligned}
		&\langle \varphi(v',c)-\varphi(v,c) \rangle =
		\e\left( \alpha_J P(v,v_*,c,c_*)(v_*-v)\partial_v \varphi  + \frac 12 \sigma_{JH}^2 D(v)^2 \partial^2_v \varphi \right)
		\cr
		&+ \frac {\e^2}2 \alpha_J^2P(v,v_*,c,c_*)^2(v_*-v)^2 \partial^2_v \varphi+ \e^2R_\e(v,v_*,c,c_*),
	\end{aligned}
\end{equation} 
where the reminder of the Taylor expansion $R_\e(v,v_*,c,c_*)$ is expressed as follows
\begin{equation}\label{eq:Tayrem}
	\begin{aligned}
		R_\e(v,v_*,c,c_*) = &\displaystyle  \frac{1}{6} \partial^3_v \varphi  (\hat v,c) \left( \e\alpha_J^3P(v,v_*,c,c_*)^3(v_*-v)^3 +  \e^{-1/2}\varrho D^3(v)\right.\cr
		& \left.+ 3 \alpha_J\sigma_{JH}^2 P(v,v_*,c,c_*)(v_*-v)D(v)^2\right)
	\end{aligned}
\end{equation}
for $\hat v = \vartheta_v v'+(1-\vartheta_v)v$ with $\vartheta_v\in[0,1]$.  The time scaled interaction operator $Q^1_J$ in weak form writes
\begin{equation}\label{eq:weakfp}
	\begin{aligned}
		&\displaystyle\frac{1}{\e}\int_{\Omega\times\R_+} Q^1_J(f_{J,\e},f_{H,\e})\varphi(v,c)\,\mathrm{d}v\,\mathrm{d}c  =\frac{1}{\e }\int_{\Omega\times\R_+}\int_{\Omega\times\R_+}\langle \varphi(v',c)-\varphi(v,c) \rangle\,\mathrm{d}v\,\mathrm{d}c\mathrm{d}v_*\,\mathrm{d}c_*  \\
		&\quad\displaystyle
		=	\int_{\Omega\times\R_+} \left( \mathcal{P}_J[f_{H,\epsilon}](v,c,t_\e)\partial_v \varphi+ \frac 12 \sigma_{JH}^2 D^2(v) \partial^2_v \varphi \right)f_{J,\epsilon}(v,c,t_\e)\,\mathrm{d}v\,\mathrm{d}c + \mathcal {R}^1_\e(\varphi),
	\end{aligned}
\end{equation} 
where we introduced the following notation for the non-local operator
\[
\mathcal{P}_J[f_{H,\epsilon}](v,c,t_\e) =\alpha_J \int_{\Omega\times\R_+} P(v,v_*,c,c_*)(v_*-v) f_{H,\epsilon}(v_*,c_*,t_\e)\,\mathrm{d}v_*\,\mathrm{d}c_*,
\]
and where the scaled reminder $\mathcal{R}^1_\e(\varphi)$ vanishes in the limit $\e\to 0$. Hence, reverting to strong form we can define the Fokker-Planck operator as follows
\begin{equation}\label{eq:ope1_scaled}
	\overline Q^1_J(f_J,f_H)(v,c,t) := -\partial_v\left(   \mathcal{P}_J[f_H](v,c,t) f_J \right) + \frac 12 \sigma_{JH}^2 \partial^2_v (D^2(v)f_J),
\end{equation}
where we indicate with a slight abuse of notation the density $f_J$ in the limit to zero for $\e$.

\paragraph{Grazing limit for $Q_J^2$} Let us now focus on the grazing limit for the operator $Q_J^2$. 
Similarly to the previous paragraph, we compute the Taylor expansion up to order three of $\varphi(v,c)$ around $(v,c)$ with respect to $v$ 
and we get
\begin{equation}\label{eq:Tayex2}
	\begin{aligned}
		\langle \varphi(v'',c)-\varphi(v,c) \rangle = \displaystyle \e\left( \tilde \alpha_J\tilde P(\rho_I,v,c)\left(G(\rho_I(t)) - v\right)\partial_v \varphi + \frac 12 \theta^2_J \tilde D^2(v) \partial^2_v \varphi \right)  \\
		\qquad\displaystyle + \frac {\e^2}2 \left( \tilde \alpha_J^2\tilde P^2(\rho_I,v,c)\left(G(\rho_I(t)) - v\right)^2 \partial^2_v \varphi \right) +  \e^2\tilde R_{ \e}(v,c),
	\end{aligned}
\end{equation} 
with remainder of the expansion expressed as
\[
\begin{aligned}
	\tilde R_{\e}(\hat v,c) =\displaystyle\frac{1}{6} \partial^3_v \varphi(\hat v,c) \left(\e\tilde \alpha_J^3\tilde P^3(\rho_I, v,c)\left(G(\rho_I(t)) - v\right)^3 + 3\tilde \alpha_J\tilde P(\rho_I, v,c) \left(G(\rho_I(t)) -  v\right)\theta^2_J \tilde D^2( v)\right),
\end{aligned}
\] 
for $\hat v = \tilde\vartheta_v v''+(1-\tilde\vartheta_v)v$, with $\tilde\vartheta_v \in [0,1]$. The time scaled operator $Q^2_J$ in weak form writes
\begin{equation}\label{eq:weakfp2}
	\begin{aligned}
		&\displaystyle\frac{1}{\e}\int_{\Omega\times\R_+} Q^2_J(f_{J,\e},\rho_I)\varphi(v,c)\,\mathrm{d}v\,\mathrm{d}c  =\frac{1}{\e }\int_{\Omega\times\R_+}\langle \varphi(v'',c)-\varphi(v,c) \rangle\,\mathrm{d}v\,\mathrm{d}c  \\
		&\qquad=	\int_{\Omega\times\R_+} \left( \tilde \alpha_J\tilde P(\rho_I,v,c)\left(G(\rho_I(t)) - v\right)\partial_v \varphi + \frac 12 \theta^2_J \tilde D^2(v) \partial^2_v \varphi \right)f_{J, \epsilon}(v,c,t_{\e})\,\mathrm{d}v\,\mathrm{d}c + \mathcal {\tilde R}_{\e}(\varphi),
	\end{aligned}
\end{equation}
where the scaled integral reminder 
$\mathcal {\tilde R}_{\e}(\varphi)$ vanishes for $\epsilon \to 0$, and reverting to strong form we can write the Fokker-Planck operator as follows
\begin{equation}\label{eq:ope2_scaled}
	\overline Q^2_J(f_J,\rho_I)(v,c,t) := -\partial_v \left(\tilde \alpha_J\tilde P(\rho_I,v,c)\left(G(\rho_I(t)) - v\right) f_J \right) +\frac{\theta^2_J}{2} \partial^2_v \left(\tilde D^2(v) f_J \right).
\end{equation}
where we indicate with a slight abuse of notation the density $f_J$ in the limit to zero for $\e$.
\paragraph{Fokker-Planck model}  
Following the previous discussions, considering the scaling \eqref{eq:scaling} the limiting Fokker-Planck model for \eqref{eq:boltz_J} reads as follows
\begin{equation}\label{eq:fp_full}
	\partial_t f_J(v,c,t) = \frac{1}{\tau_1}\sum_{H \in \mathcal{C}} \overline Q^1_J(f_J,f_H)(v,c,t)+\frac{1}{\tau_2} \overline Q^2_J(f_J,\rho_I)(v,c,t),
\end{equation}
with  $ \overline Q^1_J$ and $ \overline Q^2_J$  defined respectively in \eqref{eq:ope1_scaled} and \eqref{eq:ope2_scaled}. Furthermore, we assume that model \eqref{eq:fp_full} is complemented the following set of  zero-flux boundary conditions for $c\in\mathbb{R}_+$
\begin{align}\label{eq:bc}
	\displaystyle&\left(\frac{1}{\tau_1}\sum_{H \in \mathcal{C}}\mathcal{P}_J[f_H](v,c,t)  +\frac{1}{\tau_2}(G(\rho_I)-v)\right)f_J
	- \partial_v\left[ \left( \frac{1}{2\tau_1} \sigma_J^2 D^2(v) + \frac{1}{2\tau_2}\theta_J^2 \tilde D^2(v) \right)f_J\right]\bigg|_{v = \pm 1} = 0, \\
	\displaystyle&\qquad\qquad \left(\sigma_J^2D^2(v)+ \theta_J^2\tilde D^2(v)\right)f_J\bigg|_{v = \pm 1} = 0,
\end{align}
for $c\in\mathbb{R}_+$
for all $t \geq 0$, $J\in \mathcal{C}$ and where we introduce the parameter $\sigma_J^2:=\sum_{H\in\mathcal{C}} \sigma_{JH}^2$.
\subsubsection{Steady state of the opinion distribution}\label{sec:steadystate}
We now consider the two above detailed mechanisms of opinion formation in the grazing collision limit together. In this general framework, the steady solution of \eqref{eq:fp_full} is {very hard to compute analytically}. Under the following simplifying assumptions
\begin{align*}
	&P(v,c)=\tilde P(\rho_I,v,c) =1, \qquad D(v)=\tilde D(v)= \sqrt{1-v^2},
\end{align*}
we can assume separation of variables for $f_J(v,c,t) = h(c)g_J(v,t)$, and obtain explicit equilibrium for the opinion marginal $g_J(v,t)$. Indeed,  
exploiting the separation of variables we rewrite \eqref{eq:fp_full} as follows
\begin{equation}\label{eq:fp_opinioni2part}
	\begin{split}
		&\partial_t g_J =-\partial_v\left(\left(\frac{1}{\tau_1}\alpha_J\left(M_v(t) - v\right)+\frac{1}{\tau_2}\tilde\alpha_J(G(\rho_I)-v)\right)g_J\right)
		+\left(\frac{\tau_2\sigma_J^2+\tau_1\theta_J^2}{2\tau_1\tau_2}\right)\partial^2_v\left( \left(1-v^2\right)g_J\right),
	\end{split}
\end{equation}
where we introduced the momentum $M_v(t)$ as the sum of over the momentum of each population as follows
\begin{align}\label{eq:momentum_opi}
	M_v(t) = \sum_{H\in\mathcal{C}}\rho_Hm_{H,v}(t) = \sum_{H\in\mathcal{C}}\int_{-1}^1v g_H(v,t)\ dv.
\end{align}
Hence, integrating against $v$ equation \eqref{eq:fp_opinioni2part} and considering boundary conditions \eqref{eq:bc} we retrieve conservation of mass $\rho_J$. The evolution of the momentum $\rho_Jm_{J,v}(t)$ is obtained by multiplying \eqref{eq:fp_opinioni2part} by $v$ and integrating against $v$ as follows
\begin{equation}\label{eq:fp_momentum}
	\begin{split}
		&\partial_t (\rho_Jm_{J,v}(t)) =\frac{1}{\tau_1\tau_2}\left(\tau_2\alpha_J\left(\sum_{H\in\mathcal{C}}\rho_H m_{H,v}(t) - m_{J,v}(t)\right)+\tau_1\tilde\alpha_J\left(G(\rho_I)-m_{J,v}(t)\right)\right)\rho_J,
	\end{split}
\end{equation}
which is not a conserved quantity. Nevertheless, if we further assume that $\alpha_J = \alpha$, $\tilde \alpha_J = \tilde \alpha$, we can sum  \eqref{eq:fp_momentum} over the compartments to obtain the total momentum as
\begin{equation}\label{eq:fp_momentum_tot}
	\begin{split}
		&\frac{d}{dt} M_v(t) =\frac{\tilde\alpha}{\tau_2}(G(\rho_I)-M_v(t)).
	\end{split}
\end{equation}
%
%
In order to find the stationary solution $g_{J,\infty}$ of \eqref{eq:fp_opinioni2part} we solve
\[
\frac{1}{2}\left(\frac{\sigma_J^2\tau_2+\theta_J^2\tau_1}{\tau_1\tau_2}\right)\frac{\mathrm{d}\left(\left(1-v^2\right)g_{J,\infty}\right)}{\mathrm{d}v} = \left(\frac{\alpha\tau_2 M^\infty_{v} + \tilde \alpha\tau_1 G(\rho_I)}{\tau_1\tau_2} - v\left(\frac{\alpha\tau_2+\tilde \alpha\tau_1}{\tau_1\tau_2}\right)\right)g_{J,\infty},
\]
where $M^\infty_v$ corresponds to the equilibrium of \eqref{eq:fp_momentum_tot} for $t\to\infty$ corresponding to $M_v^\infty =G(\rho_I)$. 
Thus, we can rewrite the previous relation as
\begin{equation}\label{eq:ode_beta}
	\frac{\mathrm{d}\left(\left(1-v^2\right)g_{J,\infty}\right)}{\mathrm{d}v} = \frac{2\left(\alpha\tau_2+\tilde \alpha\tau_1\right)}{\sigma_J^2\tau_2+\theta_J^2\tau_1}\left(G(\rho_I) - v\right)g_{J,\infty}.
\end{equation}
The solution to equation \eqref{eq:ode_beta} is a Beta-type distribution that can be written as 
\begin{equation}\label{eq:sol_beta}
	g_{J,\infty}(v) = \rho_{J}\frac{4^{(L_J-1)/L_J}}{\mbox{B}\left(\frac{1+G(\rho_I)}{L_J},\frac{1-G(\rho_I)}{L_J}\right)} \left(\frac{1}{1+v}\right)^{1-\frac{1+G(\rho_I)}{L_J}}\left(\frac{1}{1-v}\right)^{1-\frac{1-G(\rho_I)}{L_J}}
\end{equation}
where we introduced the parameter
\begin{equation}\label{eq:LM}
	L_J= \frac{\sigma_J^2\tau_2+\theta_J^2\tau_1}{\left(\alpha\tau_2+\tilde \alpha\tau_1\right)}.
\end{equation}
We highlight that the first two moments of this Beta-type distribution are given by
\begin{equation}\label{eq:Beta-moments}
	\begin{aligned}
		\rho_Jm_{v,J}^\infty &= \int_{-1}^1 v g_{J,\infty}(v)\mathrm{d}v = G(\rho_I)\rho_{J}, \\
		\rho_Jm^{(2),\infty}_{v,J}& =\quad\int_{-1}^1 v^2 g_{J,\infty}(v) \mathrm{d}v = \left(\frac{L_J}{2+L_J} + \frac{2G(\rho_I)^2}{2+L_J}\right)\rho_{J}. 
	\end{aligned}
\end{equation}	

\section{Macroscopic contacts and opinion-based epidemic dynamics}\label{sec:macroepi}

As discussed in Section \ref{sect:fp}, in the quasi-invariant limit we have the possibility to study the equilibrium distribution of the operators $\overline{Q}^1_J (\cdot, \cdot)$, $\overline{Q}^2_J (\cdot, \cdot)$ defined in \eqref{eq:ope1_scaled} and \eqref{eq:ope2_scaled}. Hence, we can study the macroscopic behavior of a surrogate kinetic epidemic model by computing the evolution of observable macroscopic equations with opinion-based incidence rate. Indeed, thanks to the derivation of the Fokker-Planck-type operators , we can rewrite system \eqref{eq:seiropicont} expressing the time evolution of a disease under the combined effect of the opinion and social network obtaining
\begin{equation}\label{eq:seiropicont_fp}
	\begin{aligned}
		\partial_t f_S  &=\displaystyle -K(f_S,f_I)(v,c,t) + \frac{1}{\tau_1}\sum_{J \in \mathcal{C}}\overline Q^1_S(f_S,f_J)(v,c,t)+\frac{1}{\tau_2}\overline Q^2_S(f_S,\rho_I)(v,c,t),\\
		\partial_t f_E  &=\displaystyle K(f_S,f_I)(v,c,t) - \sigma_E f_E + \frac{1}{\tau_1}\sum_{J \in \mathcal{C}}\overline Q^1_E(f_E,f_J)(v,c,t)+ \frac{1}{\tau_2}\overline Q^2_E(f_E,\rho_I)(v,c,t),\\
		\partial_t f_I &=\displaystyle \sigma_E f_E - \gamma f_I +  \frac{1}{\tau_1}\sum_{J \in \mathcal{C}} \overline Q^1_I(f_I,f_J)(v,c,t)+\frac{1}{\tau_2}\overline Q^2_I(f_I,\rho_I)(v,c,t),\\
		\partial_t f_R &=\displaystyle \gamma f_I + \frac{1}{\tau_1}\sum_{J \in \mathcal{C}}\overline Q^1_R(f_R,f_J)(v,c,t) + \frac{1}{\tau_2}\overline Q^2_R(f_R,\rho_I)(v,c,t).
	\end{aligned}
\end{equation}
It becomes particularly interesting now to observe the temporal evolution of the disease independently of the temporal dynamics of both opinion and formation of a social network. By integrating the above system \eqref{eq:seiropicont_fp} both with respect to the opinion $v$ and the contact variable $c$, we can study the disease progression by taking these factors into account implicitly. By performing such operation, we can gain a clearer understanding of the course of the outbreak and underlying mechanisms. This can also help in identifying key stages in the disease progression, critical points for intervention, and potential outcomes under various scenarios.

Thus by using as operator $K$ the one defined in \eqref{eq:K}, we get the evolution of the mass functions $\rho_J(t)$, namely the relative number of susceptible, exposed, infected and recovered, which reads 
\begin{equation}\label{eq:seiropicont_rho}
	\begin{aligned}
		\frac{\mathrm{d}}{\mathrm{d}t} \rho_S(t) &=-\frac{\beta_1}{4}\left( 1 - m_{S,v}(t)\right)\left(1 - m_{I,v}(t)\right)\rho_S(t)\rho_I(t), \\
		\frac{\mathrm{d}}{\mathrm{d}t} \rho_E(t) &=\frac{\beta_1}{4}\left( 1 - m_{S,v}(t)\right)\left(1 - m_{I,v}(t)\right)\rho_S(t)\rho_I(t)  - \sigma_E\rho_E(t), \\
		\frac{\mathrm{d}}{\mathrm{d}t} \rho_I(t) &=\sigma_E\rho_E(t) - \gamma \rho_I(t), \\
		\frac{\mathrm{d}}{\mathrm{d}t} \rho_R(t) &= \gamma \rho_I(t).
	\end{aligned}
\end{equation}
However, one can notice that the above system is not closed, since the evolution of $\rho_J$ depends on the evolution of $m_{J,v}$, i.e. the first order moments with respect to the opinion variable. The closure of system \eqref{eq:seiropicont_fp} can be formally obtained thanks to the following considerations. The first is about the time scale at which individuals shape their ideas with respect to the time scale at which the epidemic evolves. It has been observed for instance during the COVID-19 outbreak that people used to react differently with respect to the use of protective masks and social distance during the course of the year. This observation permits to  take $\tau_1,\tau_2 \ll 1$ in \eqref{eq:seiropicont} and suppose that each compartment $J\in \mathcal{C}$ reaches its local equilibrium with a given mass fraction $\rho_J$ and with a local mean opinion $m_{J,v}$.
Under the above hypothesis, we can get the time evolution of the mean values, following the procedure in \cite{Zanella_23}. Thus we multiply \eqref{eq:seiropicont_fp} by the opinion $v$ and we integrate over $\Omega$ and $\mathbb{R}_+$ for the number of contacts, to obtain a system of equations for $\rho_Jm_{J,v}(t)$, which again depends on the higher order moment $\rho_Sm^{(2)}_{S,v}(t)$ and for $J=S$ reads as
\begin{equation*}
	\begin{aligned}
		&\displaystyle\frac{\mathrm{d}}{\mathrm{d}t}\left( \rho_S(t) m_{S,v}(t) \right) =\displaystyle-\frac{\beta_1}{4}\left( 1  - m_{I,v}(t) \right)\rho_I(t)\rho_S(t)\left(m_{S,v}(t)-\int_{\mathbb R^+}\int_{-1}^1v^2f_S(v,c,t)\mathrm{d}v\,\mathrm{d}c \right)\cr
		&\qquad\quad\displaystyle+\frac{1}{\tau_1}\sum_{J\in\mathcal{C}}\rho_J(t)\rho_S(t)\left(m_{J,v}(t) - m_{S,v}(t)\right) 
		+ \frac{1}{\tau_2}\rho_S(t)\left(G(\rho_I) - m_{S,v}(t)\right).
	\end{aligned}
\end{equation*}
Hence, in order to close the system we consider an equilibrium closure approach based on the Beta distribution \eqref{eq:sol_beta} whose first two moments are defined in \eqref{eq:Beta-moments}. In detail, we get
\[
\int_{\mathbb R^+}\int_{[-1,1]}v^2 f_S(v,c,t)dv\,dc \approx \int_{\mathbb R^+} \int_{[-1,1]} v^2 g_{S,\infty}(v)h(c)\mathrm{d}v \,\mathrm{d}c.
\] 
Combining this with the mass system \eqref{eq:seiropicont_rho} we can write the system of equation for $m_J(t)$ as follows
	\begin{equation}\label{eq:seiropicont_mw}
		\begin{aligned}
			&\displaystyle\frac{\mathrm{d}}{\mathrm{d}t}m_{S,v}(t) = \displaystyle-\frac{\beta_1}{4}\left( 1  - m_{I,v}(t)\right)\rho_I(t) \left(m_{S,v}^2-m_{S,v}^{(2)} \right)
			\cr
			&\qquad\quad \displaystyle+\frac{1}{\tau_1}\sum_{J\in\mathcal{C}}\rho_J(t)\left(m_{J,v}(t) - m_{S,v}(t)\right)+ \frac{1}{\tau_2}\left(G(\rho_I) - m_{S,v}(t)\right), \\
			&\displaystyle\frac{\mathrm{d}}{\mathrm{d}t} m_{E,v}(t)=\displaystyle\frac{\beta_1}{4}\frac{\rho_S(t) \rho_I(t)}{\rho_E(t)}\left( 1  - m_{I,v}(t) \right) \left[m_{S,v}(t) - m_{S,v}^{(2)} - m_{E,v}(t)(1-m_{S,v}(t) )\right] \cr
			&\qquad\quad \displaystyle +\frac{1}{\tau_1}\sum_{J\in\mathcal{C} }\rho_J(t)\left(m_{J,v}(t) - m_{E,v}(t)\right)+ \frac{1}{\tau_2}\left(G(\rho_I) - m_{E,v}(t)\right), \\
			&\displaystyle\frac{\mathrm{d}}{\mathrm{d}t}m_{I,v}(t) =\displaystyle\sigma_E \frac{\rho_E(t)}{\rho_I(t)}\left(m_{E,v}(t) - m_{I,v}(t) \right) \cr
			&\qquad\quad \displaystyle + \frac{1}{\tau_1}\sum_{J\in\mathcal{C} }\rho_J(t)\left(m_{J,v}(t) - m_{I,v}(t)\right)  + \frac{1}{\tau_2}\left(G(\rho_I) - m_{I,v}(t)\right), \\
			&\displaystyle\frac{\mathrm{d}}{\mathrm{d}t}m_{R,v}(t)= \displaystyle\gamma \frac{\rho_I(t)}{\rho_R(t)} \left(m_{I,v}(t)-m_{R,v}(t) \right)\cr
			&\qquad\quad + \frac{1}{\tau_1}\sum_{J\in\mathcal{C}}\rho_J(t)\left(m_{J,v}(t) - m_{R,v}(t)\right)+ \frac{1}{\tau_2}\left(G(\rho_I) - m_{R,v}(t)\right).
		\end{aligned}
	\end{equation}
	The coupled system \eqref{eq:seiropicont_rho}-\eqref{eq:seiropicont_mw} deserves some remarks. Let us observe in fact that the first part of system \eqref{eq:seiropicont_rho} is very similar to that of a standard SEIR model, the main difference being related to the infection rate which in our case depends upon the average opinion of susceptible and infected. If these individual on average are prone to use protective masks and keep social distance then the probability of getting infected decreases toward zero. Let also observe that even if the number of contacts does not appear explicitly into the system, influencers have an important role in shaping the opinion of individuals since their position with respect to the use of protective measures impacts on the average opinion $m_{J,v}$ more than the opinion of other individuals.
	The second set of equations \eqref{eq:seiropicont_mw} for the average opinions retains the coupling with the epidemic model, as well as the terms for the alignment of opinion (multiplied by $\tau_1^{-1}$), and the interaction with the background opinion $G(\rho_I)$ (multiplied by $\tau_2^{-1}$).

\section{{Numerical results}}\label{sec:numerics}
In order to {approximate numerically the solution to} system \eqref{eq:seiropicont_fp}, we choose a time discretization step $\epsilon=\Delta t > 0$ of the time interval $[0,T]$ and we are interested in the approximation $f_J^n(v,c)$ of $f_J(v,c,t^n)$, where $t^n = n\Delta t$. Following \cite{Zanella_23}, we introduce a splitting between the opinion evolution step $f^*_J = \mathcal{I}_{\Delta t}(f_J^*,f^*)$ 
\begin{equation}\label{eq:step_1}
	\begin{cases}
		\displaystyle\partial_t f^*_J = \frac{1}{\tau_1}\sum_{H \in \mathcal{C}} Q^1_J(f^*_J,f^*_H) + \frac{1}{\tau_2}Q_J^2(f^*_J,\rho_I), \\
		f_J^*(v,c,0) = f_J^n(v,c),
	\end{cases}
\end{equation}
for all $J\in \mathcal{C}$, and the epidemiological step $f^{**}_J = \mathcal{E}_{\Delta t}(f_J^{**})$
\begin{equation}\label{eq:step_2}
	\begin{cases}
		\partial_t f^{**}_S = -K(f^{**}_S,f^{**}_I)(v,c,t), \\
		\partial_t f^{**}_E = K(f^{**}_S,f^{**}_I)(v,c,t)-\sigma_E f^{**}_E,\\
		\partial_t f^{**}_I = \sigma_E f^{**}_E - \gamma f_I^{**}, \\
		\partial_t f^{**}_R  = \gamma f_I^{**},\\
		f_J^{**}(v,c,0) = f_J^*(v,c,\Delta t),
	\end{cases}
\end{equation}
with incidence rate $K$ given by \eqref{eq:K}. The solution at time $t^{n+1}$ is given by the combination of \eqref{eq:step_1} and \eqref{eq:step_2}, i.e.
\[
f_J^{n+1}(v,c) = \mathcal{E}_{\Delta t}\left( \mathcal{I}_{\Delta t} \left(f_J^n(v,c)\right)\right),
\]
for all $J\in\mathcal{C}$. Both \eqref{eq:step_1} and \eqref{eq:step_2} are solved using a Monte Carlo approach, as explained, respectively, in Algorithms \ref{alg:opi} and \ref{alg:seir}. {For $\xi^{JH}$ and $\zeta^{J}$, with $J,H\in \mathcal{C}$, the choice fell on uniformly distributed random variables.}

\begin{algorithm}
	\caption{Asymptotic particle-based algorithm (Nanbu-like algorithm) for approximating \eqref{eq:step_1}}\label{alg:opi}
	\begin{algorithmic}[h!]
		\State Fix $\tilde\epsilon = \tau_1\tau_2/(\tau_1 + \tau_2)$ and $N_s$.
		\State Sample $N_s$ particles $\left\{v_i^n,c_i\right\}_{i=1}^{N_s}$ from the initial distribution $f_J^n(v,c)$.
		\State Randomly choose $N_1 = \tilde\epsilon/\tau_1 < N_s$ particles $\left\{v_i^n,c_i\right\}_{i=1}^{N_1}$ from the initial distribution $f_J^n(v,c)$.
		\For{$k = 0:\Delta t$ with step $\tilde \epsilon$}	
		\State Set $N_c = round(N_1/2)$.
		\State Select $N_c$ random pairs $(i, i_*)$ uniformly, without repetition among all possible pairs.
		\For {$i = 1:N_c$}
		\State  Sample $\xi^n_i,\xi^n_{i_*}$ from a uniform distribution $\mathcal{U}\left[-1,1\right]$  and compute
		\begin{equation}\label{eq:trules_scaled}
			\begin{aligned}
				v_i^{n+1}         &=         v_i^{{n}}          + \tilde\epsilon \alpha_J P(v^n_i,v^n_{i_*},c_i,c_{i_*})(v_{i_*}^n - v_{i}^n) + \sqrt{3\tilde\epsilon}\sigma_J D(v^n_i,c_i) \xi^n_i,\cr
				v_{i_*}^{n+1}  &=          v_{i_*}^{{n}}  + \tilde\epsilon \alpha_J P(v^n_{i_*},v^n_i,c_{i_*},c_i)(v_{i}^n - v_{i_*}^n) + \sqrt{3\tilde\epsilon}\sigma_J D(v^n_{i_*},c_{i_*}) \xi^n_{i_*}.
			\end{aligned}
		\end{equation}
		\EndFor
		\EndFor
		\State Choose the remaining $N_2 = \tilde\epsilon/\tau_2 < N_s$ particles $\left\{v_i^n,c_i\right\}_{i=N_1+1}^{N_s}$ from $f_J^n(v,c)$.
		\For{$k = 0:\Delta t$ with step $\tilde \epsilon$}	
		\For {$i = N_1+1:N_s$}
		\State  Sample $\zeta^n_i$ from a uniform distribution $\mathcal{U}\left[-1,1\right]$ and compute
		\begin{equation}\label{eq:opinion_infected_scaled}
			v_i^{n+1} =  v_i^n + \tilde \epsilon\tilde \alpha_J\tilde P(\rho_I^n,v_i,c_i)\left(G(\rho_I^n,v_i^n) - v_i^n\right) + \sqrt{3\tilde \epsilon}\theta_J\zeta^n_i \tilde D(v_i^n).
		\end{equation}
		\EndFor
		\EndFor
	\end{algorithmic}
\end{algorithm}

\begin{algorithm}
	\caption{Asymptotic particle-based algorithm for approximating step \eqref{eq:step_2}}\label{alg:seir}
	\begin{algorithmic}[h!]
		\State Fix $N_s$.
		\State Sample $\left\{v_i,c_i,J_i\right\}_{i=1}^{N_s}$ from the initial distribution $f_0^{**}(v,c)$.
		\State Let $N_J \subseteq \{1,\dots,N_s\}$ be the set of indices of the particles with disease status $J$, $J\in \{S, E, I \}$.
		\For {$i \in N_S$}
		\State Compute $p = \beta_1(1-v_i)\rho_I(1-m_I)/4$.
		\State Sample from a Bernoulli distribution $\mathcal{B}(p)$.
		\If{the sample is 1}
		\State Let $\{v_i,c_i,J_i\} = \{v_i^{n},c_i,E\}$.
		\EndIf
		\EndFor
		\For {$i \in N_E$}
		\State Sample from a Bernoulli distribution $\mathcal{B}(\sigma_E)$.
		\If{the sample is 1}
		\State Let $\{v_i^{n+1},c_i,J_i\} = \{v_i^{n},c_i,I\}$.
		\EndIf
		\EndFor
		\For {$i \in N_I$}
		\State Sample from a Bernoulli distribution $\mathcal{B}(\gamma)$.
		\If{the sample is 1}
		\State Let $\{v_i^{n+1},c_i,J_i\} = \{v_i^{n},c_i,R\}$.
		\EndIf
		\EndFor
	\end{algorithmic}
\end{algorithm}

\subsection{Test 1: Evolution of opinion dynamics}
In this first test, we focus on the opinion dynamics, simulating  \eqref{eq:boltz_J}  via \eqref{eq:step_1} neglecting the epidemiological compartmentalization.
Hence, we solve numerically the following
\[
\partial_t f(v,c,t) = \frac{1}{\tau_1} Q^1(f,f)(v,c,t) + \frac{1}{\tau_2} Q^2(f,\rho)(v,c,t),
\]
with Algorithm \eqref{alg:opi}, 
where operator $Q^1$ and $Q^2$ are defined as in \eqref{kine-ww} and \eqref{kine-ww2}, and $\rho$ is a given function with values in $[-1,1]$, i.e. $\rho: \mathbb{R}_+ \to [-1,1]$. We consider as initial distribution
$f(v,c,0) = g_0(v,c;\Theta)h(c)$, where we introduce the following parametric function
\begin{align}\label{eq:initdata_g}
	g_\ell(v,c;\Theta) =  C_\ell\left[ \chi(v\in\Omega_\ell) \chi(c\leq 10 \bar c)+ \left(\Theta\mathcal{N}(v;-m,s^2)+(1-\Theta)\mathcal{N}(v;m,s^2)\right)\chi(c> 10 \bar c)\right],
\end{align}
where for the current case, i.e. $\ell=0$, the opinion domain is such that $\Omega_0 = [-1,1]$, and where the second part of the expression is a convex combination of two symmetric gaussians $\mathcal{N}(v;\pm m,s^2)$ with  $\Theta \in[0,1]$. We fix these parameters in the following way mean $|m|=0.7$ and standard deviation $s=0.05$, and  $\Theta =1/4$. Finally,
$C_\ell>0$ is a constant such that  $\int_{\mathbb{R}_+}\int_{-1}^1 g_0(v,c,0)h(c) \, \mathrm{d}v \,\mathrm{d}c = 1$.  
%
%
Hence, the initial data represents a population of agents with contact density at equilibrium $h(c)$. The distribution of opinions is divided as follows: for those with contact levels less or equal then $10\bar c$, the opinion density is uniformly spread over the interval $[-1,1]$; among agents with more than $10\bar c$ contacts, $3/4$ of the population has opinions normally distributed around the value $v=0.7$,  while the remaining $1/4$ of the population has opinions normally distributed around the value $v=-0.7$. In this setting, a majority of the agents with more contacts reflects a positive opinion rather than a negative one.

The interaction functions $P$ and $D$ in \eqref{eq:trules_scaled}, for  $v,v_* \in \Omega$ and $c,c_* \in \mathbb{R}_+$, follow the structure proposed in \eqref{eq:kernel_P} where
\begin{align}\label{eq:kern_T1a}
	\Delta(c,c_*) =  \delta\frac{c_*}{\min\{c,c_*\}} ,\quad W(c,c_*)= \frac{c_*}{c+c_*}, \quad D(v) = \sqrt{1-v^2} 
\end{align}
with parameter $\delta= 0.1$, and we set $\alpha=1$. 
The choice of the interaction kernel $P$ in \eqref{eq:kern_T1a} models the fact that agents tend to interact only with people in a certain proximity of their own opinion, but this ``proximity'' is larger if the interaction takes place with an agent whose number of connection is higher. Moreover, also the intensity of the interaction is greater if the second agent has a larger number of followers, meaning that influent people (the grade of influence in our setting is given by the popularity, i.e. the number of connections) tend to have a more significant impact on the opinion of the less influent ones.

For the dynamics in \eqref{eq:opinion_infected_scaled}, we set $\tilde \alpha = 0.25$ and the following functions $\tilde W$, $G$, and $\tilde D$ as
\begin{align}\label{eq:kern_T1b}
	\tilde W(c,\overline c) = \frac{\bar c}{\max\{c,\bar c\}},\quad G(\rho(t)) = \rho(t) - 0.5, \quad \tilde D(v) = \sqrt{1-v^2},
\end{align}
where $\rho(t) = \exp(-t)$. Thus, $\tilde P$ accounts only for the number of contacts and the agent opinion, while the value of function $G$, according with $\rho(t)$, varies with time, being $0.5$ at the initial time and tending monotonically to the value $-0.5$ as the time increases. The parameters associated to the random variables in \eqref{eq:trules_scaled} and \eqref{eq:opinion_infected_scaled} are fixed to $\sigma^2 =\theta^2 = 10^{-3}$.

We simulated three different scenarios with Algorithm \ref{alg:opi} considering the scaling parameter $\tilde \epsilon = \tau_1 \tau_2/(\tau_1+\tau_2)$ and different choices of $\tau_1$ and $\tau_2$. The first column of Figure \ref{fig:test1}, refers to the case $\tau_1 = 1$ and $\tau_2 = 10^{-2}$, so the interaction with the background is faster than the binary interaction between agents. In this case, influent agents are not capable of dragging the opinion of the less popular ones towards a positive value, instead we witness a slow consensus of the majority of the agents (especially the less influent ones) towards $v=-0.5$. In the second case we set $\tau_1=\tau_2=1$, in this scenario the interaction with the background, represented by $\rho$, and the binary interaction between agents have the same speed, and we can observe in the central column of Figure \ref{fig:test1}, how the consensus to $v=-0.5$ is reached much faster and also the influent agents converge towards that value. In the last simulation of this section, we set $\tau_1 = 10^{-2}$ and $\tau_1=1$. In this case, as shown in the last column of Figure \ref{fig:test1}, the speed of the binary interaction between agents is faster than the interaction with the background, and the influent agents are able to drag the opinion of the rest of the population towards their initial one.

\begin{figure}[h!]
	\centering
	\includegraphics[width=0.26\textwidth]{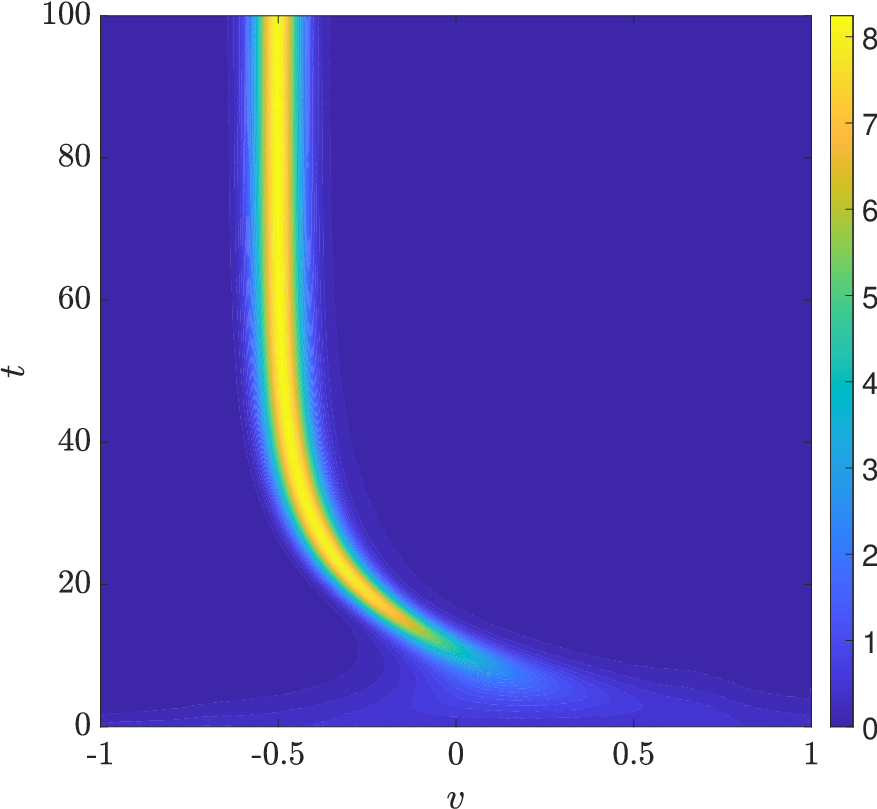}\quad \includegraphics[width=0.265\textwidth]{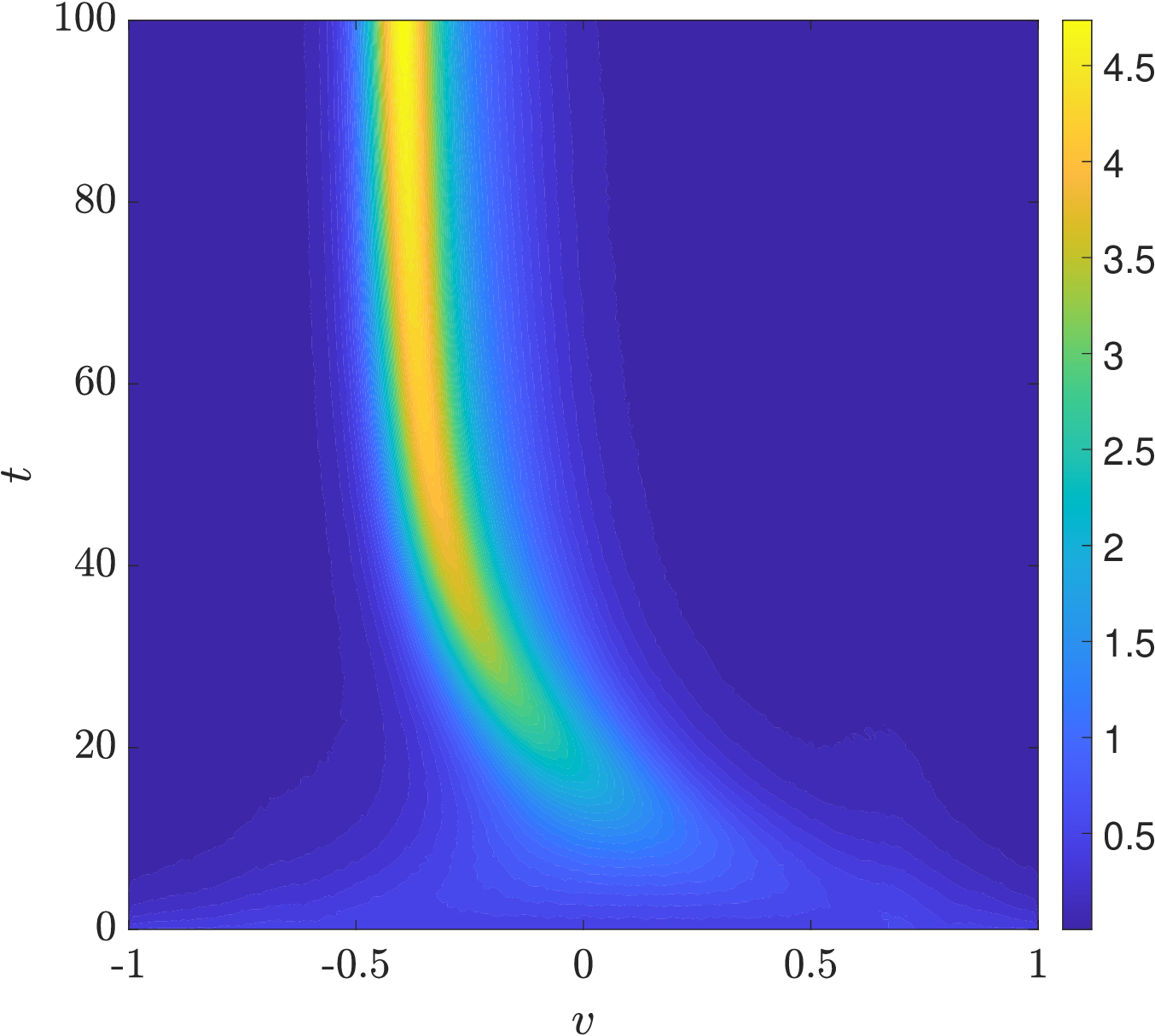}\quad\includegraphics[width=0.26\textwidth]{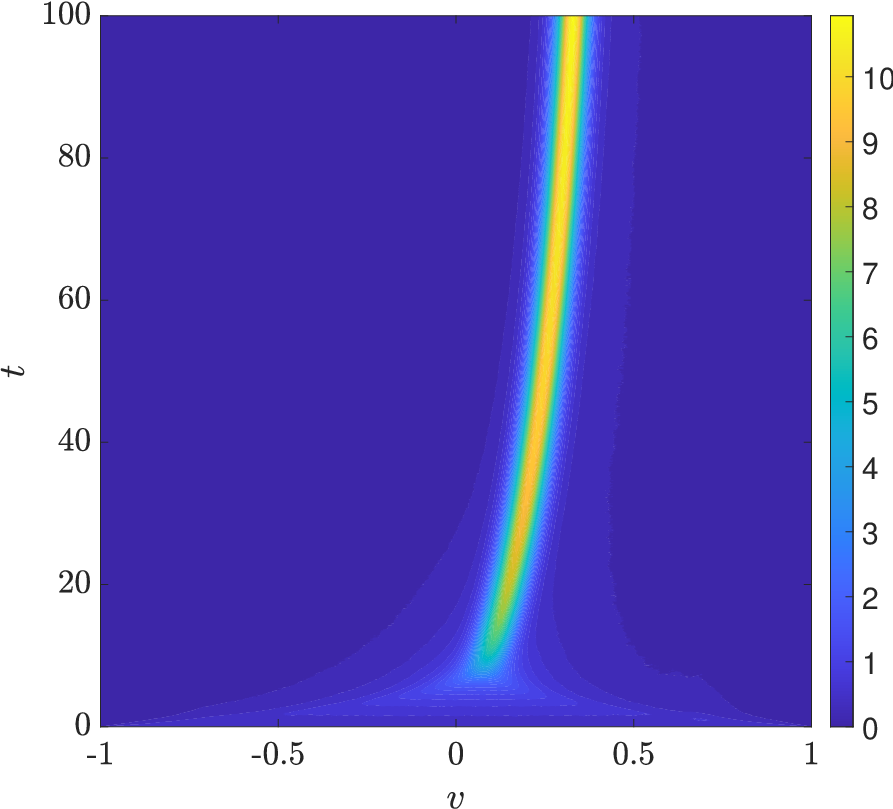}
	\centering
	\includegraphics[width=0.235\textwidth]{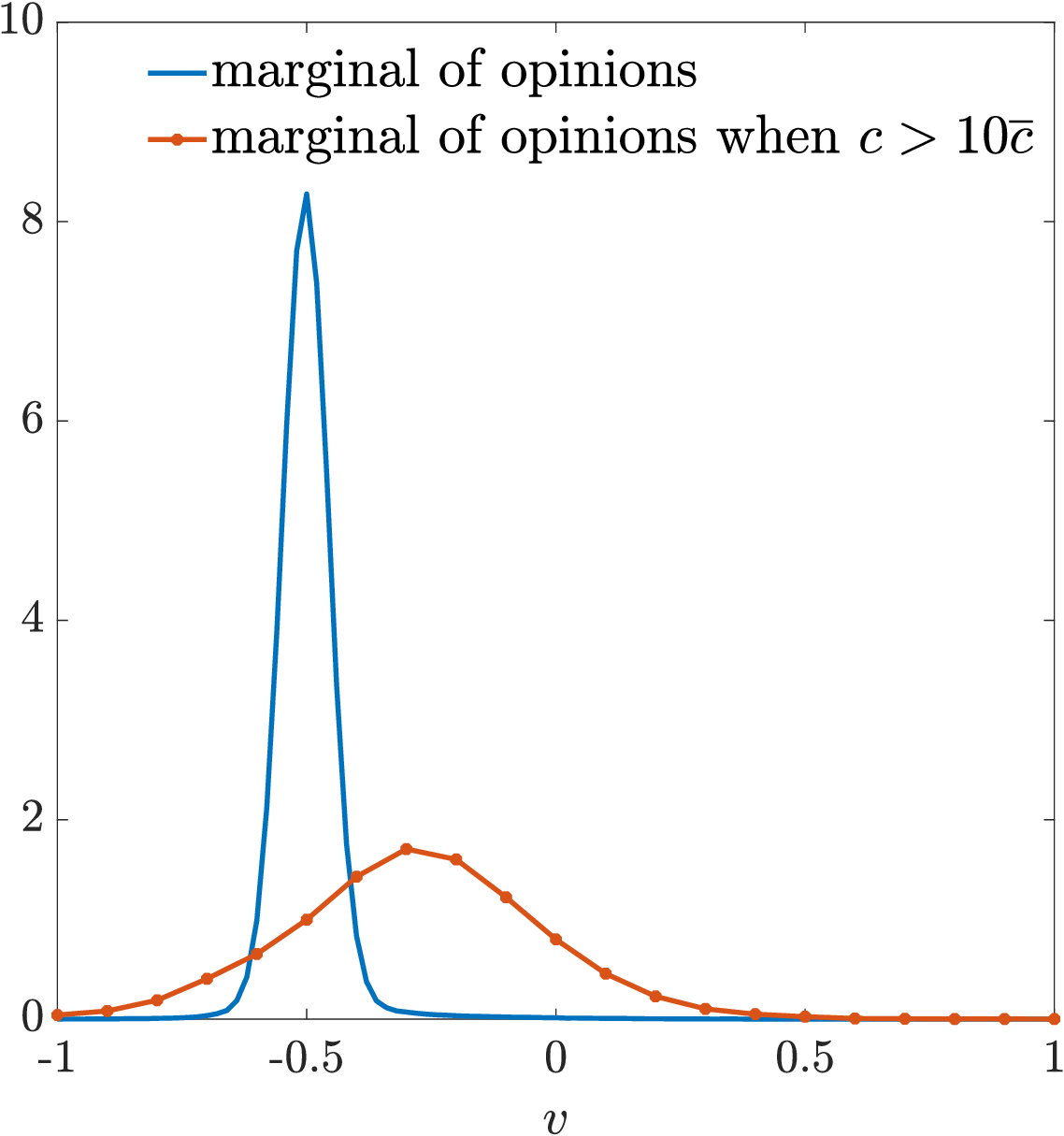}\quad\quad \includegraphics[width=0.23\textwidth]{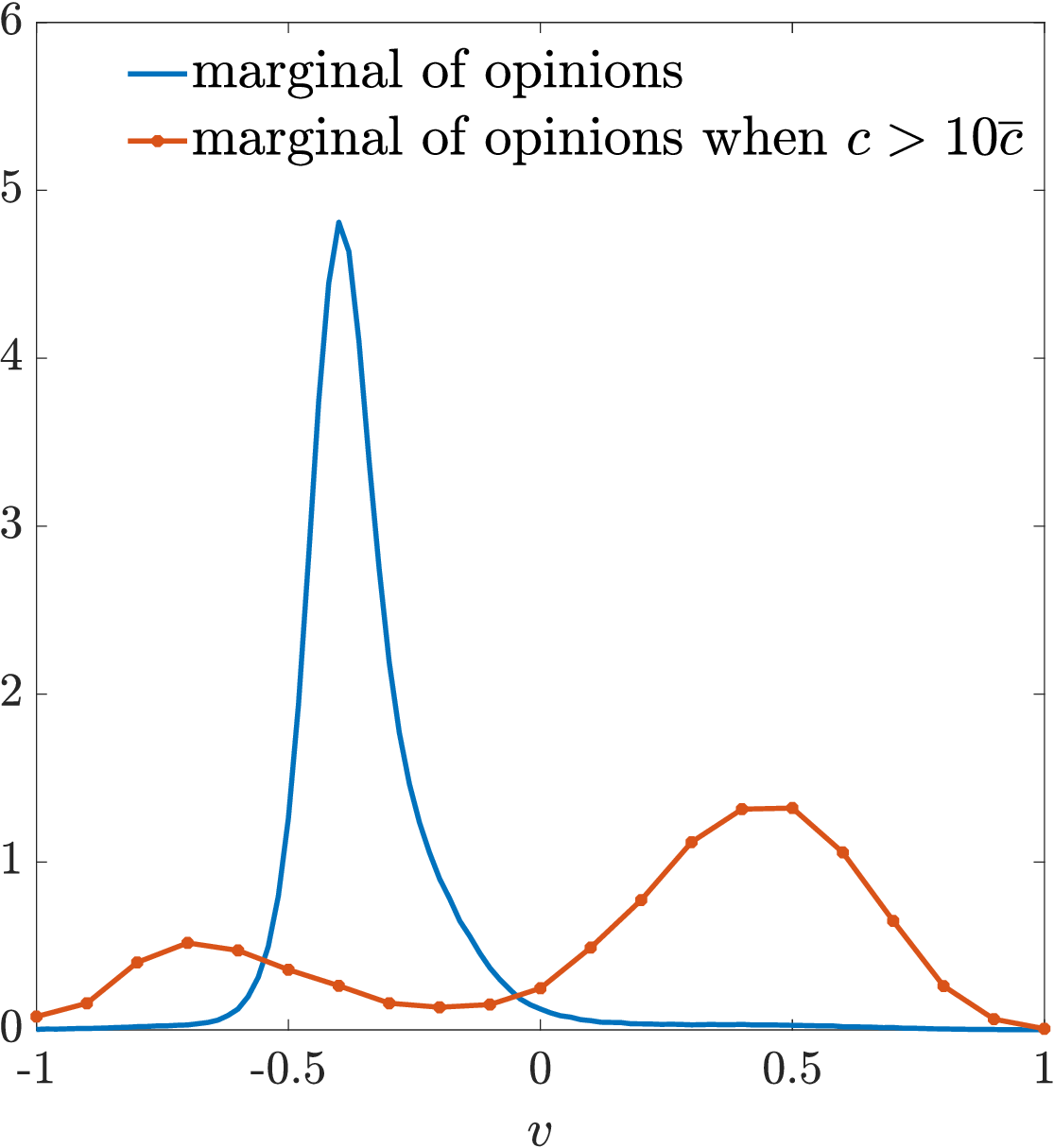}\quad\quad\includegraphics[width=0.235\textwidth]{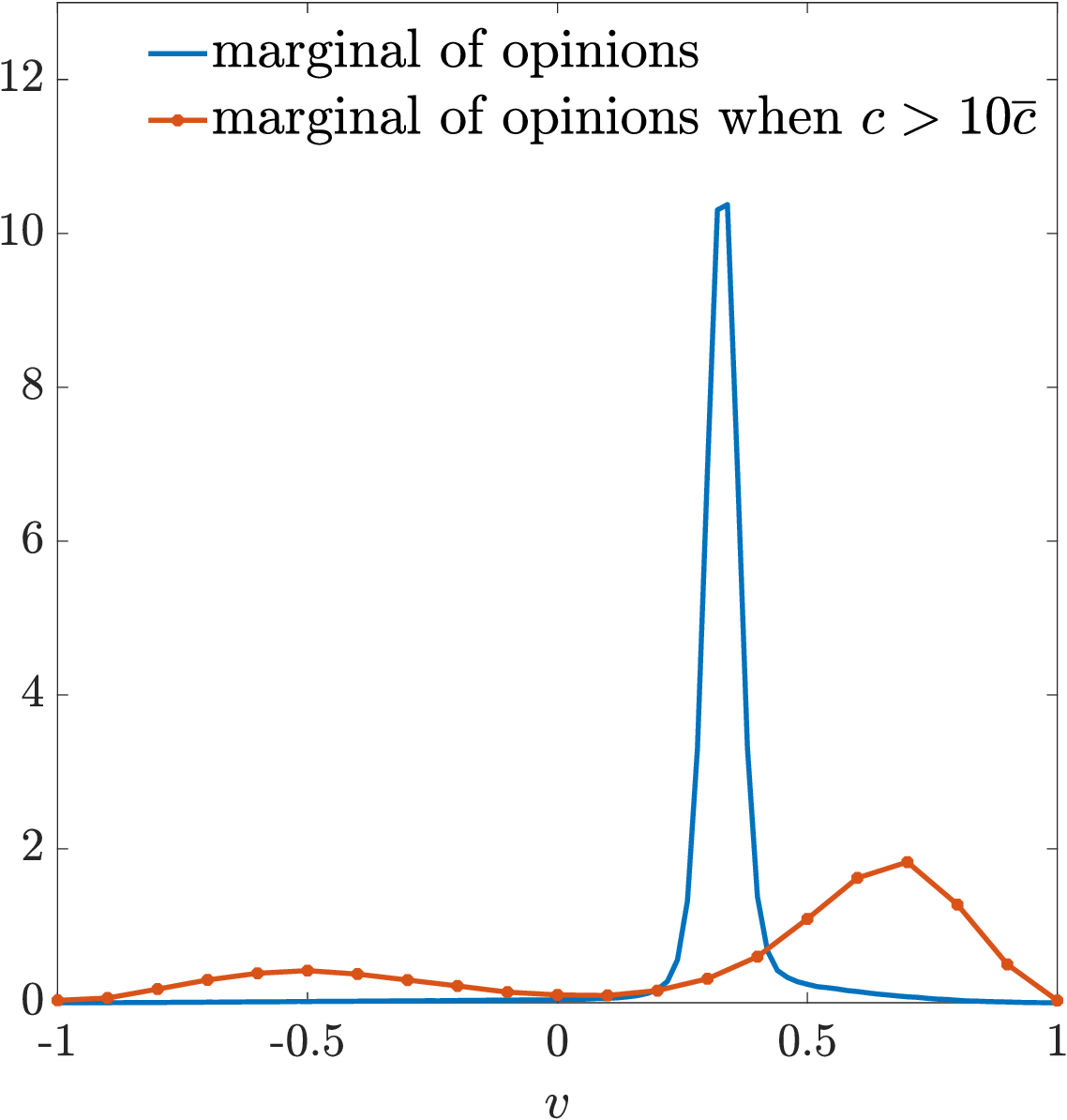}
	\caption{{\em Test 1.} Evolution in time of the marginal distribution of the opinions (upper row) and the normalized marginal distribution of opinions both for the entire population and for $c>10 \overline c$ at $t=100$  (bottom row), for the case $\tau_1=1$ and $\tau_2=10^{-2}$ (left), $\tau_1=1$ and $\tau_2=1$ (center), and  $\tau_1=10^{-2}$ and $\tau_2=1$ (right).}
	\label{fig:test1}
\end{figure}

\subsection{Test 2: Consistency of the macroscopic model}
We discuss the consistency of the macroscopic model \eqref{eq:seiropicont_rho}-\eqref{eq:seiropicont_mw}, with the macroscopic quantities reconstructed from the kinetic model \eqref{eq:seiropicont} in the Fokker-Planck regime \eqref{eq:seiropicont_fp}.

In order to address correctly this problem, we remark that the equilibrium closure previously introduced is based on the ansatz that the opinion dynamics evolve faster than the epidemic one. To exploit this behavior, we introduce the positive parameter $\lambda>0$, and we scale in model  \eqref{eq:seiropicont_fp} the parameter $\tau_1$ and $\tau_2$ as follows
\begin{equation}\label{eq:scaling_lambda}
	\tau_1\to \lambda\tau_1,\qquad \tau_2\to\lambda\tau_2.
\end{equation}
We consider
zero-flux boundary conditions \eqref{eq:bc} scaled by $\lambda$, and initial data
\begin{equations}
	f_S(v,c,0) &= \rho_S(0)\chi(v\in[-1,0])h(c), \quad &f_E(v,c,0) &= \rho_E(0)\chi(v\in[-1,0])h(c), \\
	f_I(v,c,0) &= \rho_I(0)\chi(v\in[0,1])h(c),\quad& f_R(v,c,0) &= \rho_R(0)\chi(v\in[0,1])h(c),
\end{equations}
where we set
\begin{equation}\label{eq:rho_init}
	\rho_E (0) = \rho_I (0) = \rho_R(0) =  10^{-2}, \quad \rho_S(0) = 1- \left(\rho_E(0)+\rho_I(0)+\rho_R(0)\right),
\end{equation}
and parameters $\tau_1 = 1$, $\tau_2 = 10$, $\beta_1 = 0.4$, $\sigma_E = 0.5$, $\gamma = 1/12$, and $\sigma_J^2= \theta_J^2 = 10^{-3}$ for all $J \in \mathcal{C}$. 
Thus, for a fixed $\lambda$, we introduce the reconstructed macroscopic quantities as
\begin{equation}\label{eq:rho_w_lambda}
	\rho_J^\lambda(t) = \int_{\Omega\times\R_+} f_J(v,c,t) dv dc,\qquad 
	m_J^\lambda(t) = \frac{1}{	\rho_J^\lambda(t)}\int_{\Omega\times\R_+} v f_J(v,c,t) dv dc,
\end{equation}
where  the solution of the macroscopic system \eqref{eq:seiropicont_rho}-- \eqref{eq:seiropicont_mw} is recovered for $\lambda \to 0^+$.

Finally, we compare the macroscopic quantities \eqref{eq:rho_w_lambda}, reconstructed from the numerical solution of \eqref{eq:seiropicont} under scaling \eqref{eq:scaling_lambda}, with the solution of the macroscopic system  \eqref{eq:seiropicont_rho}--\eqref{eq:seiropicont_mw}, solved with initial data \eqref{eq:rho_init} for the masses, $m_{S,v}(0) = -0.5$, $m_{E,v}(0) = -0.5$, $m_{I,v}(0) = 0.5$, $m_{R,v}(0) = 0.5$, for the first moments, and the same parameters chosen for the kinetic model. 
In Figure \ref{fig:test_macro} we depict the solution $(\rho_J^\lambda(t),m_{J,v}^\lambda(t))$ showing that  for decreasing values of $\lambda$, i.e. $\lambda=1, 10^{-3}$,  the solution of the macroscopic system $(\rho_J(t),m_{J,v}(t))$ is approached.
%

\begin{figure}[h!]
	\centering
	\includegraphics[width=0.22\textwidth]{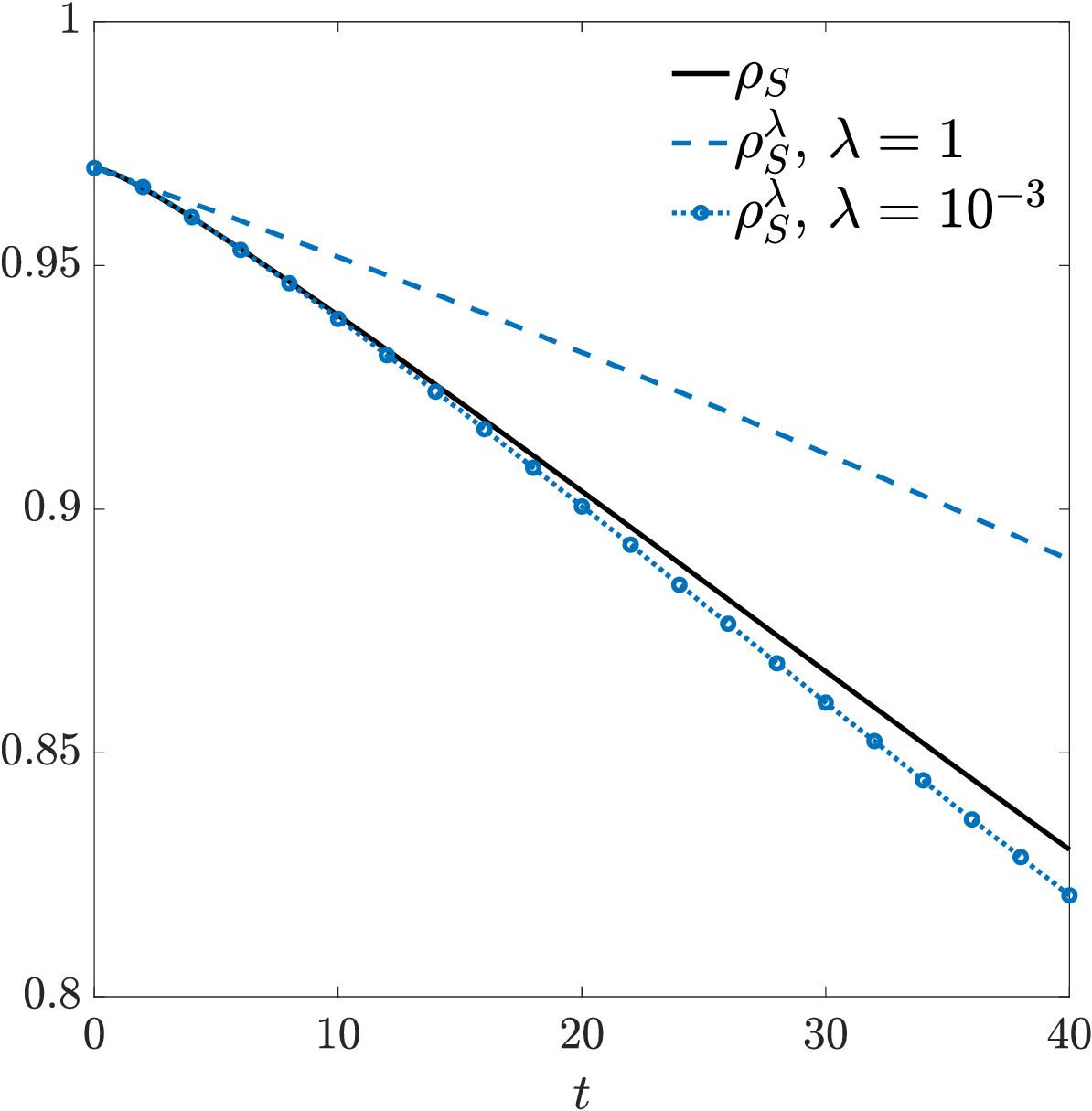}\quad\includegraphics[width=0.22\textwidth]{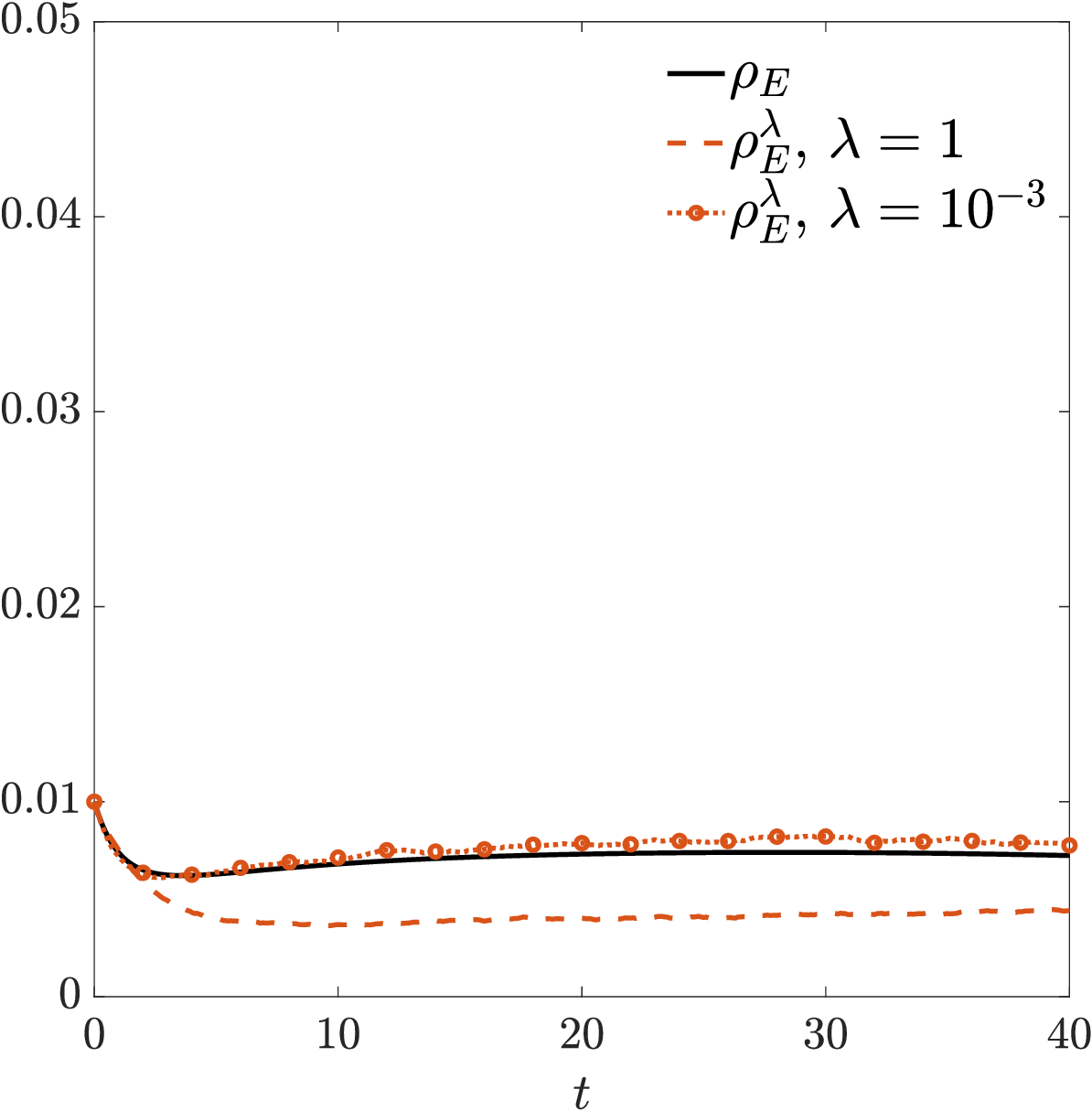}\quad\includegraphics[width=0.22\textwidth]{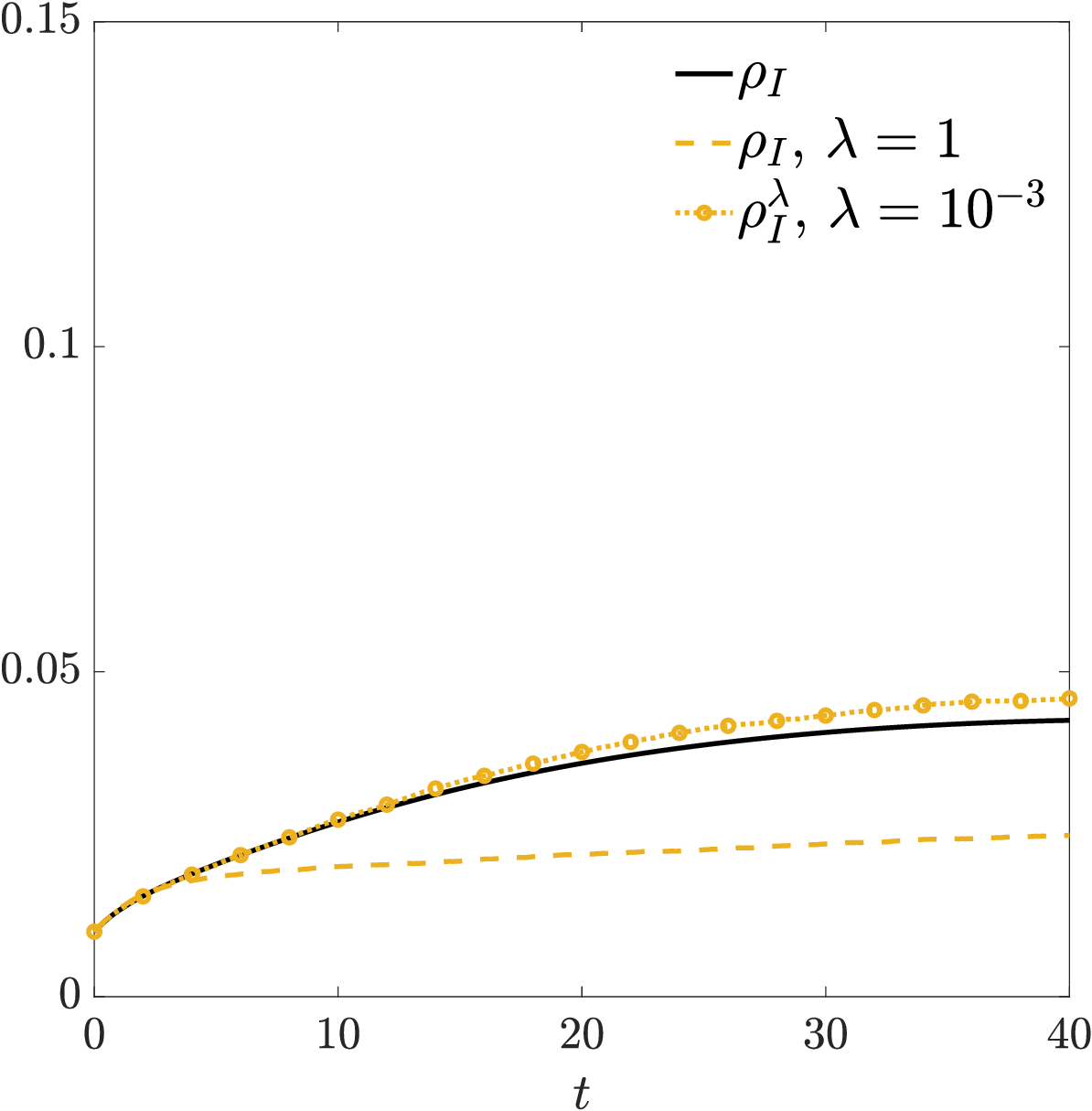}\quad\includegraphics[width=0.22\textwidth]{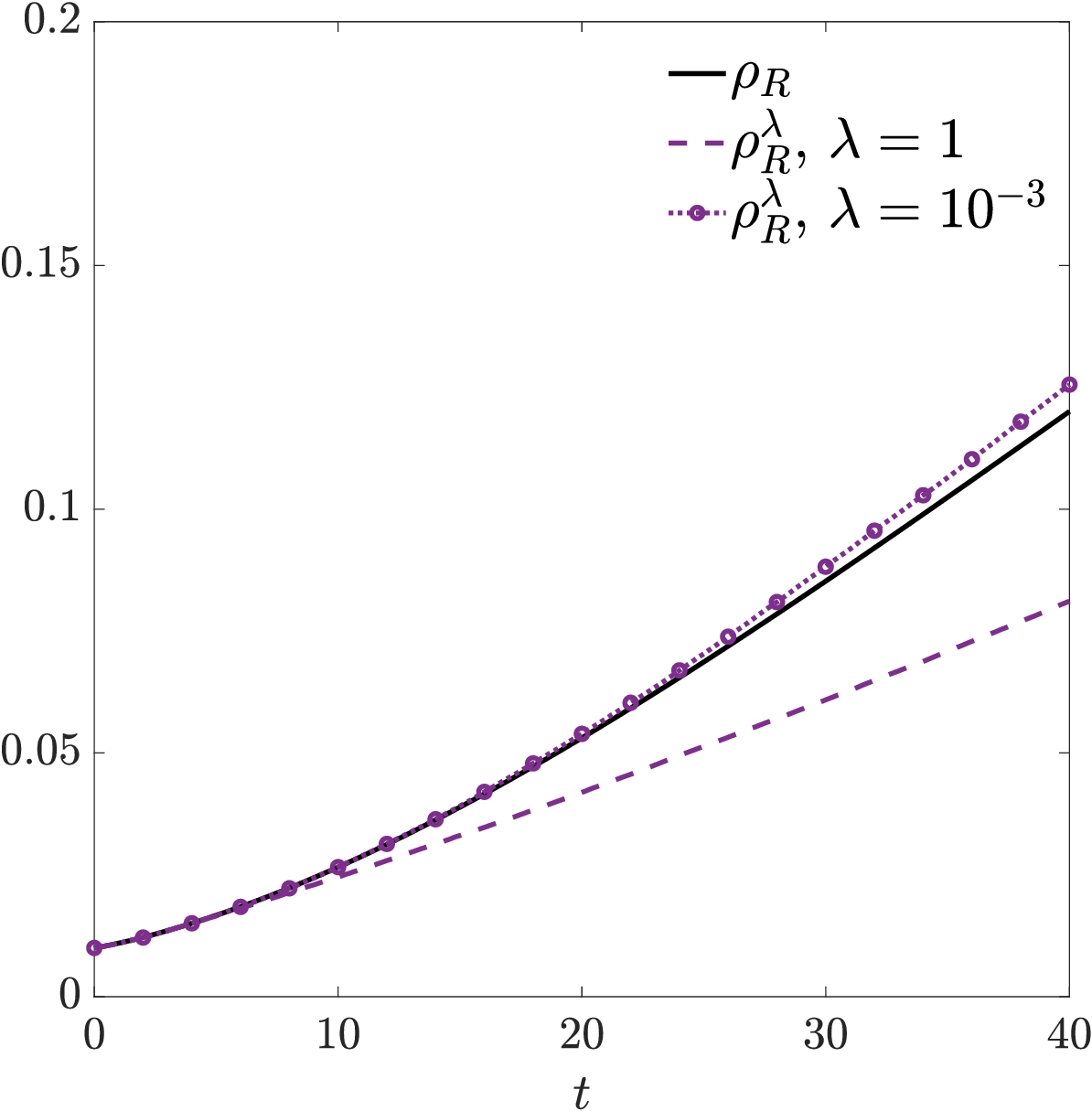}
	\centering
	\includegraphics[width=0.22\textwidth]{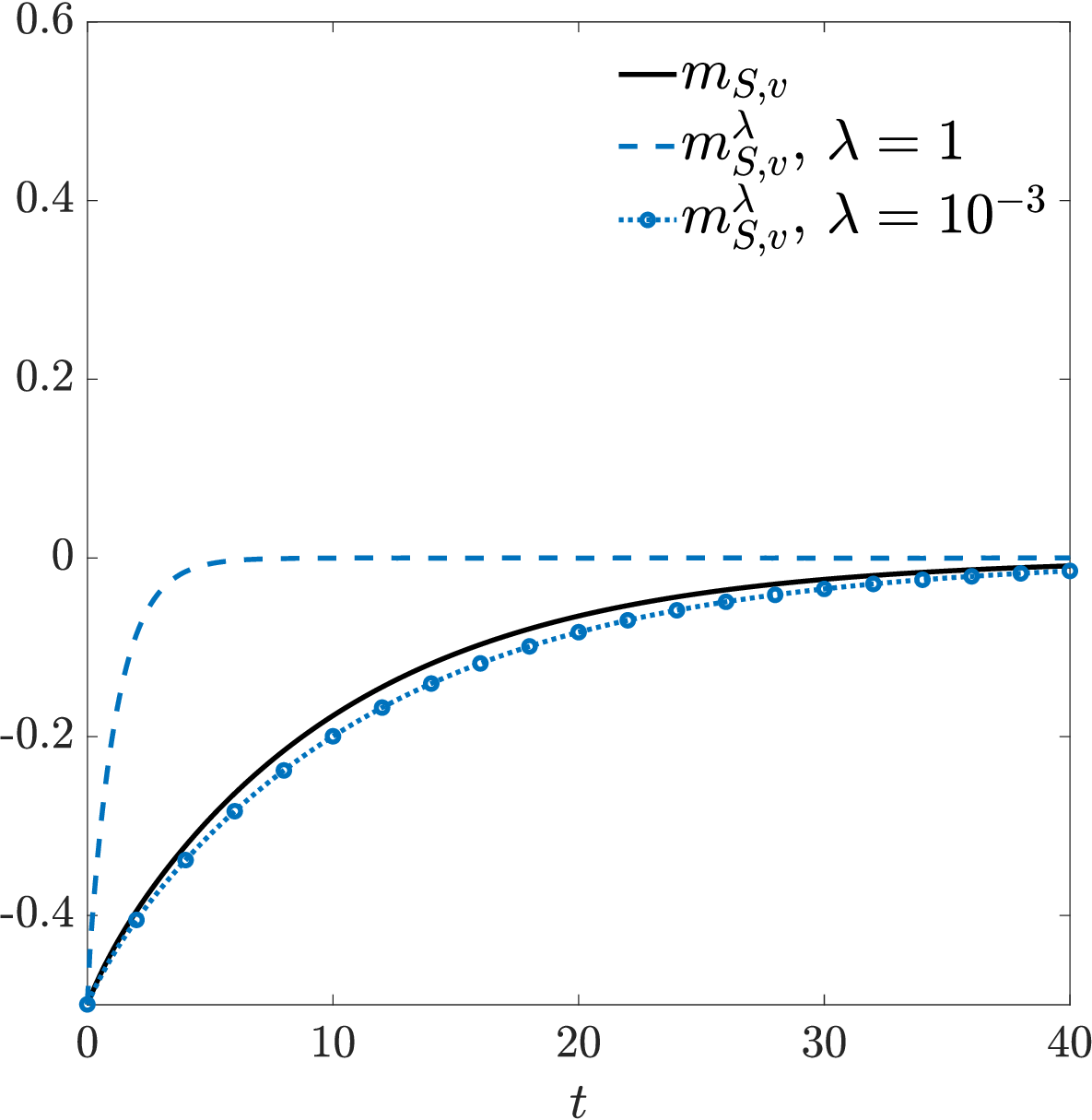}\quad
	\includegraphics[width=0.22\textwidth]{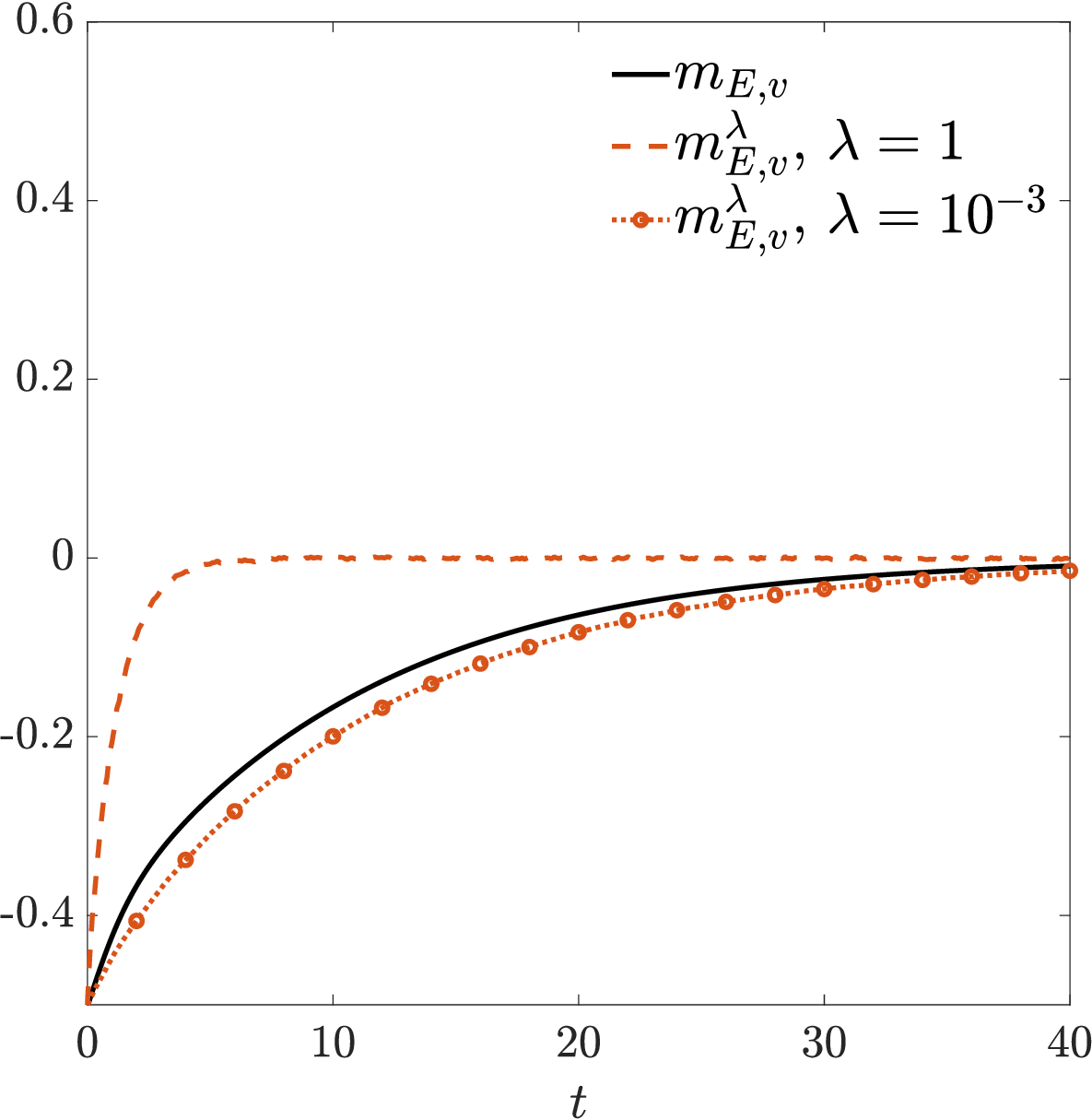}
	\quad \includegraphics[width=0.22\textwidth]{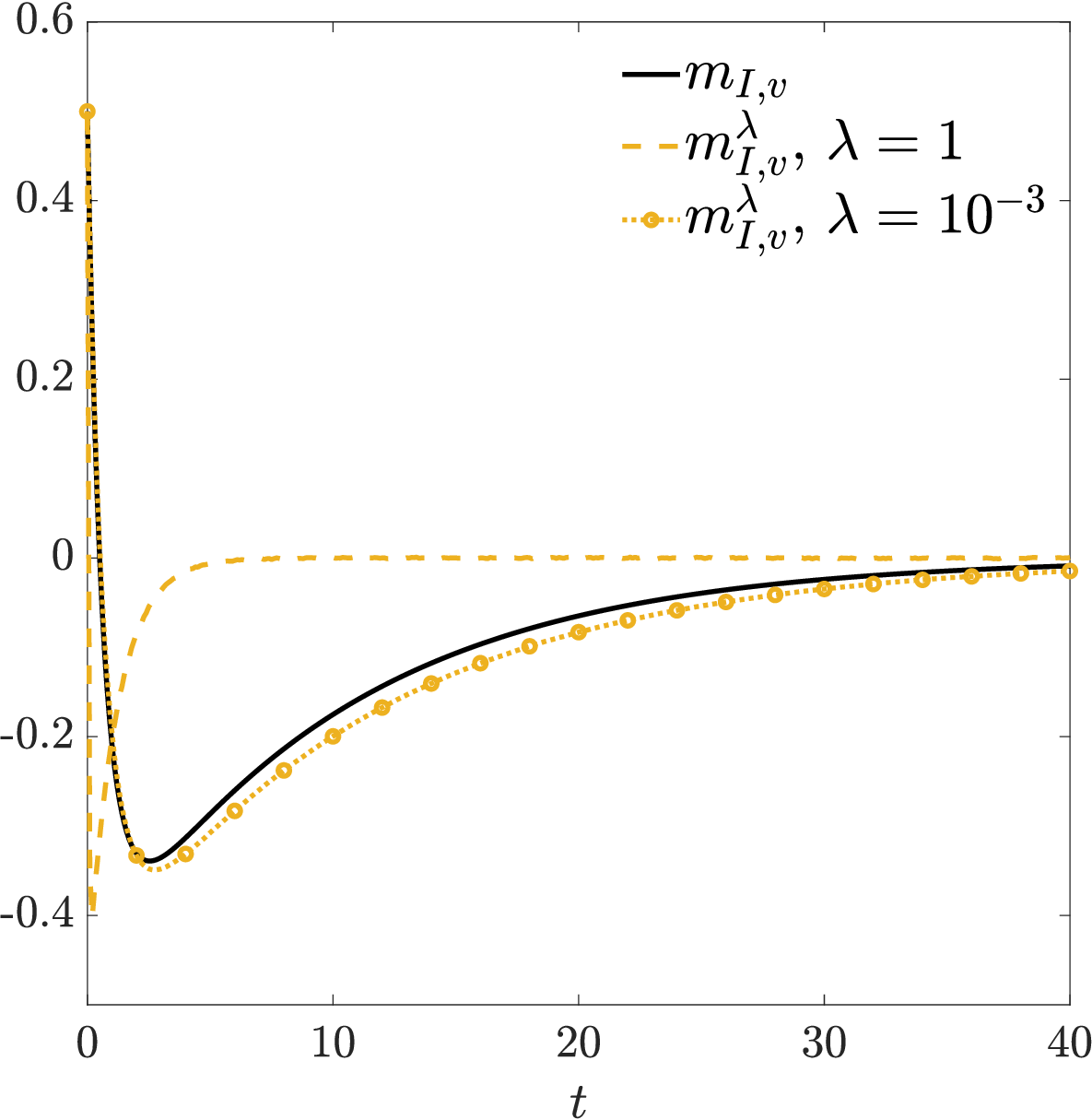}
	\quad \includegraphics[width=0.22\textwidth]{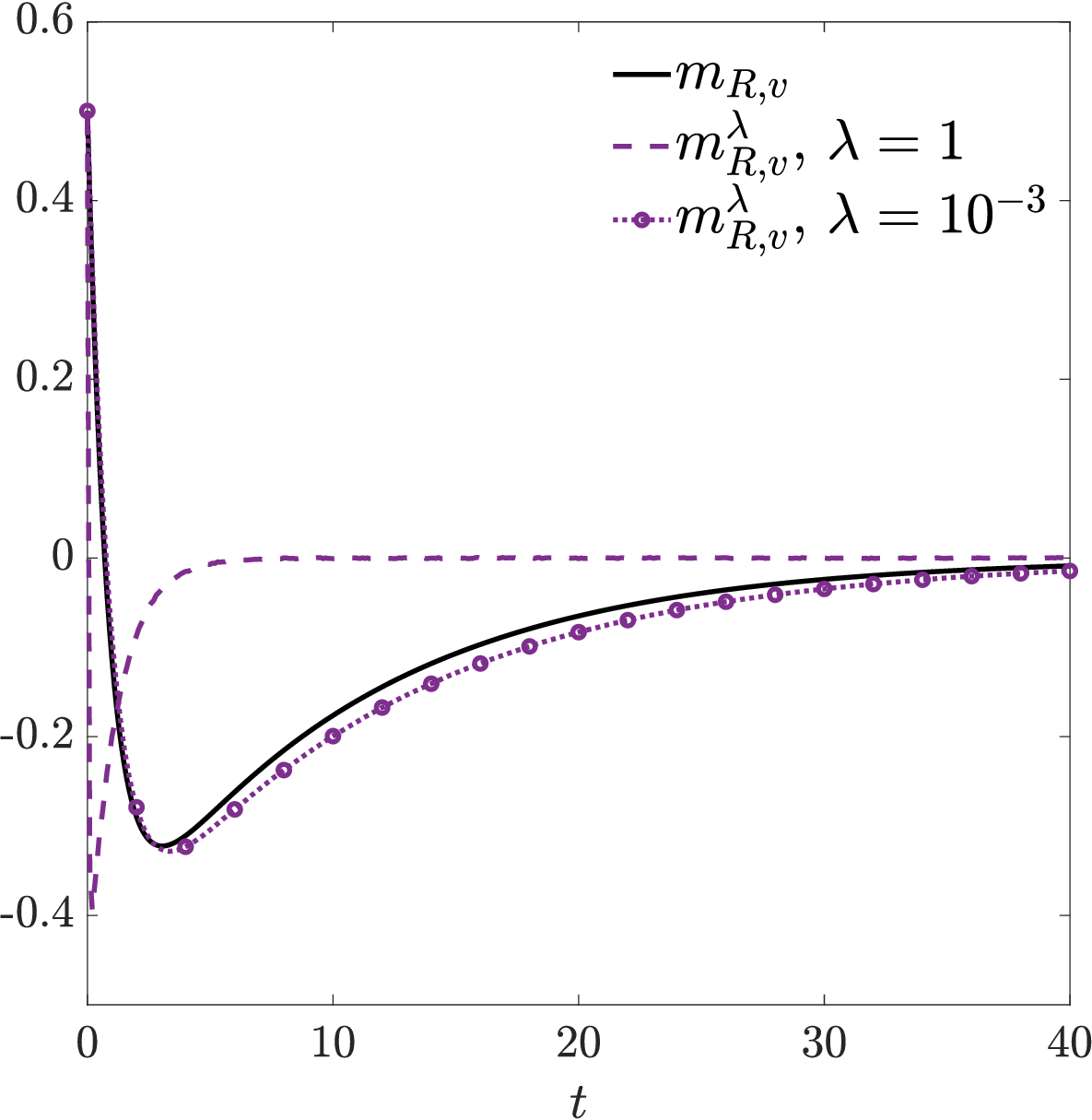}
	\caption{\textit{Test 2.} Evolution in time of the mass (upper row) and of the first moment with respect to the opinions (bottom row) of each compartment solving the macroscopic system \eqref{eq:seiropicont_rho}-\eqref{eq:seiropicont_mw} with $\tau_1 = 1$ and $\tau_2 = 10$, compared to the ones reconstructed from the solution of the kinetic system \eqref{eq:seiropicont} with scaling \eqref{eq:scaling_lambda} for the same values on $\tau_1, \tau_2$ and  $\lambda = 1,10^{-3}$.}
	\label{fig:test_macro}
\end{figure}

\subsection{Test 3: Impact of influent agents on the spread of the disease}\label{sec:test2}
In this second test, we consider the coupled process of epidemics and opinion dynamics. In particular we are interested in studying the presence of agents with higher influence and differing opinions toward protective measures and treatment. Our aim is to show how influent agents can reduce (or increase) the spread of the epidemics and to investigate how the speed of opinions diffusion can alter the course of the contagion.

For the opinion dynamics \eqref{eq:trules_scaled} and \eqref{eq:opinion_infected_scaled} governed by the operators $Q_J^1$ and $Q_J^2$ in \eqref{eq:step_1}, with $J\in \mathcal{C}$,  we retain the same $\alpha$ and functions $P,D$ and $\tilde D$ chosen in Test 1 and described in \eqref{eq:kern_T1a} and  \eqref{eq:kern_T1b}, with the same parameter $\Delta = 0.1$ and standard deviations $\sigma_J^2= \theta_J^2 = 10^{-3}$ for all $J \in \mathcal{C}$.
%
The main difference consists in the choice of the function $\tilde P$, defined as \eqref{eq:kernel_Ptilde} with 
\begin{equation}\label{eq:kernel2}
	\psi(\rho_I) = \rho_I,\quad	\tilde W(c,\overline c) =  \frac{\bar c}{\max\{c, \bar c\}},
\end{equation} 
and $G$, chosen as in \eqref{eq:Gfunction} with $\overline \rho$ set equal to $0.05$, meaning that the target opinion $G(\rho_I)$ is equal to $-1$ if the percentage of infected is under $5\%$, otherwise it is equal to $1$. We set $\tilde \alpha_J = 0.5$ for all $J\in\mathcal{C}$. 
With the choices made in \eqref{eq:kernel2} the interaction kernel $\tilde P$ models the fact that the higher is the percentage of population infected, the greater is the tendency to move towards opinion $G$ expressed in \eqref{eq:Gfunction}, but the greater is the number of contacts, the less the agent is prone to be impacted by the advancement of the disease (assuming that the so-called \emph{influencers} tend to retain a form of credibility and ``brand name" that is tied to their opinions, and therefore cannot be drastically affected by the trend of contagions). Also, $\tilde P$ is small for agents with opinions close to the extremal values $-1$ and $1$, modeling the fact that extremists usually tend to only slightly modify their opinions, no matter what happens in their environment. 
We performed  simulations comparing two different setting for different initial data. 

{\em Setting 1}. (Majority of the {``influencers''} against of protective measures). In this first case we consider the initial densities of the various compartments defined as
\begin{equation}
	\begin{aligned}\label{eq:init_1}
		f_S(v,c,0) &= \rho_S(0)g_1(v,c;\Theta)h(c), \quad f_E(v,c,0) = \rho_E(0)g_1(v,c;\Theta)h(c), \cr
		f_I(v,c,0) &= \rho_I(0)g_2(v,c;\Theta)h(c), \quad f_R(v,c,0) = \rho_R(0)g_2(v,c;\Theta)h(c),
	\end{aligned}
\end{equation}
where $\rho_E(0) = \rho_I(0) = \rho_R(0) = 10^{-2}$, while $\rho_S(0) = 1- (\rho_E(0) + \rho_I(0) + \rho_R(0))$, $h(c)$ is the contact distribution \eqref{eq:logn}, and $g_1,g_2$ defined according to
\eqref{eq:initdata_g} with $\Theta = 3/4$, $\Omega_1 = [-1,0]$, and $\Omega_2 = [0,1]$, and the remaining parameters the same as Test 1.
Thus, the contacts are stationary, and the agents' opinions depend on their popularity: for the ones with less than $10\bar c$ contacts, if they are in compartments $S$ or $E$ then they have opinions on the protective measures uniformly distributed in $[-1,0]$, while the ones in $I$ and $R$ have opinions uniformly distributed in $[0,1]$. A fraction of $3/4$ of the agents with more than $10\bar c$ contacts is initially very strongly against the use of protective measures, while the remaining $1/4$ of them is strongly in favor.

{\em Setting 2}. (Majority of the {``influencers''} in favor of protective measures).  In this first case we consider the initial densities of the various compartments defined as \eqref{eq:init_1}
where $\rho_E(0) = \rho_I(0) = \rho_R(0) = 10^{-2}$, while $\rho_S(0) = 1- (\rho_E(0) + \rho_I(0) + \rho_R(0))$,  $h(c)$ is the contact distribution \eqref{eq:logn}, and $g_1,g_2$ defined according to
\eqref{eq:initdata_g} with $\Theta = 1/4$, $\Omega_1 = [-1,0]$, and $\Omega_2 = [0,1]$,  and the remaining parameters the same as Test 1.

	Again the contacts are distributed according to \eqref{eq:logn}, and for the agents with less than $10\bar c$ contacts the situation is the same as the one described in \eqref{eq:init_1}, but in this case $3/4$ of the ones with more than $10\bar c$ contacts are initially very strongly in favor of the use of protective measures and the remaining $1/4$ is strongly against.
	
	Furthermore, we consider various time scales for the spread of the disease and for the opinions diffusion, i.e. for different choices of $\tau_1$ and $\tau_2$. 
	In particular, we simulated the case $\tau_1=\tau_2 = 1$ (both processes of opinion formation have the same speed, comparable to the velocity of the spread of the disease in object) and $\tau_1 = 10^{-2}$, $\tau_2=10^{-3}$ (opinions spread faster than the disease, especially the ones influenced by the background). 
	The parameters for the epidemiological model in \eqref{eq:step_2} are set as $\beta_1 = 0.4$, $\sigma_E = 0.5$ and $\gamma = 1/12$, where $K$ is defined as in \eqref{eq:K}.
	
	Figure \ref{fig:test2_tau1} shows the results when $\tau_1=\tau_2=1$. The upper row refers to the simulation with initial condition given by \eqref{eq:init_1} with $\Theta = 3/4$ ({\em Setting 1}), while the bottom row refers to the ones starting from \eqref{eq:init_1} with $\Theta = 1/4$ ({\em Setting 2}). On the left, we depict the time evolution of the relative mass of the compartments $\rho_J$; in the center, the time evolution of the marginal density of opinions; on the right, we present the normalized marginal of opinions at final time $T$, both for the entire population and for the subset where $c>10 \overline c$. In this scenario, opinions evolve at the same time scale as the disease, leading to the emergence of three distinct epidemic waves, starting from \eqref{eq:init_1}. Notably, for {\em Setting 1} the opinions of the influential agents ultimately steer the rest of the population towards a negative stance on the use of protective measures.  
	Conversely, for {\em Setting 2}, it is evident the impact of these influential agents in guiding the population toward a positive opinion on the use of protective measures.
	
	Figure \ref{fig:test2_tausmall} reports the results when $\tau_1= 10^{-2}$ and $\tau_2 = 10^{-3}$. Similarly to Figure  \ref{fig:test2_tau1}, we report again on the upper row the simulation with initial condition given by {\em Setting 1}, while the bottom row refers to the ones starting from {\em Setting 2}. 
	In this case, we observe that the impact of influent agents is weaker. Indeed, for {\em Setting 1}  the evolution of the marginal opinions clearly shows that opinions are gradually steered towards the negative stance of these agents by the final time. Conversely, for {\em Setting 2} opinions tend to shift toward a neutral value.
	
	Finally, Figure \ref{fig:test2_confronto} shows a comparison between the evolution of $\rho_I$ for {\em Setting 1} and {\em Setting 2}, both with $\tau_1 = \tau_2 = 1$ (left) and  with $\tau_1= 10^{-2}$ and $\tau_2 = 10^{-3}$ (right). The difference in severity of the two epidemic waves is evident compared to the alert threshold  of $5\%$, i.e. $\overline \rho = 0.05$.

	\begin{figure}[h!]
		\centering
		\includegraphics[width=0.3\textwidth]{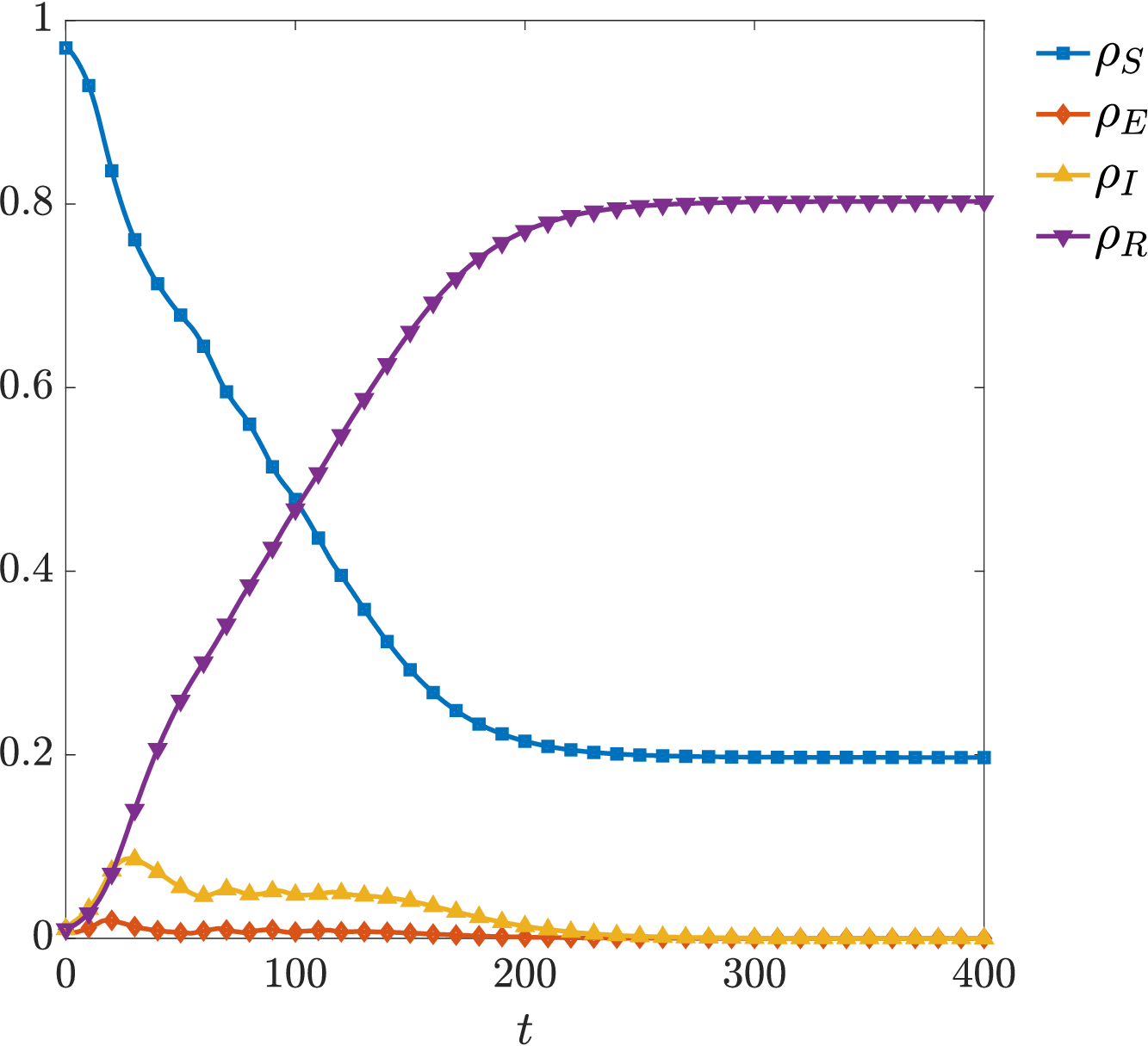}\quad\includegraphics[width=0.305\textwidth]{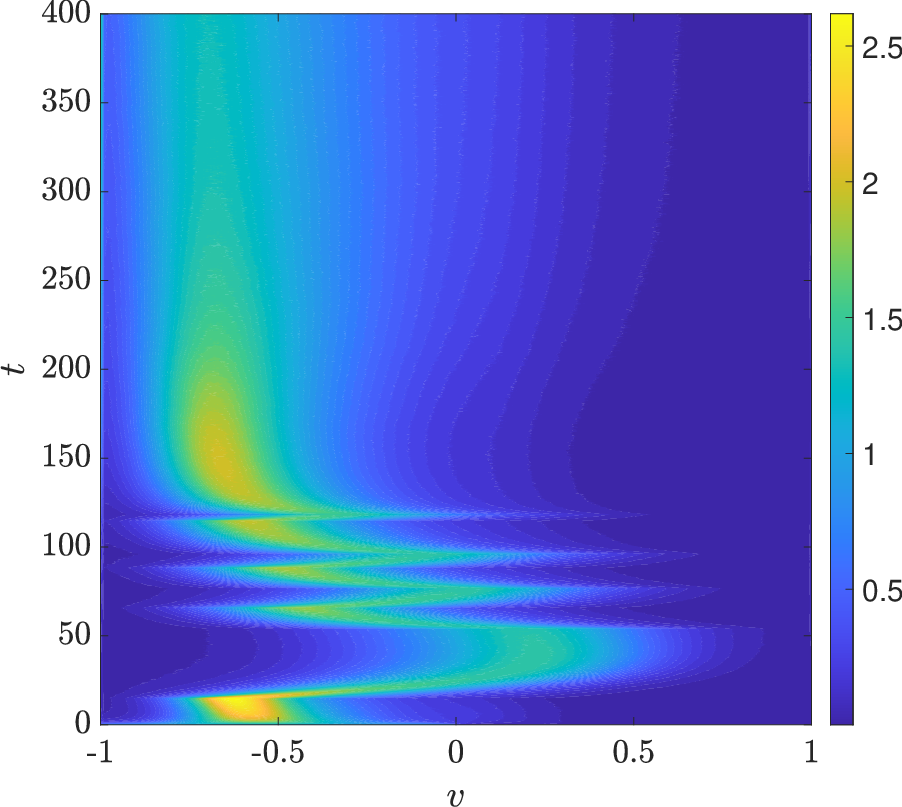}\quad\includegraphics[width=0.26\textwidth]{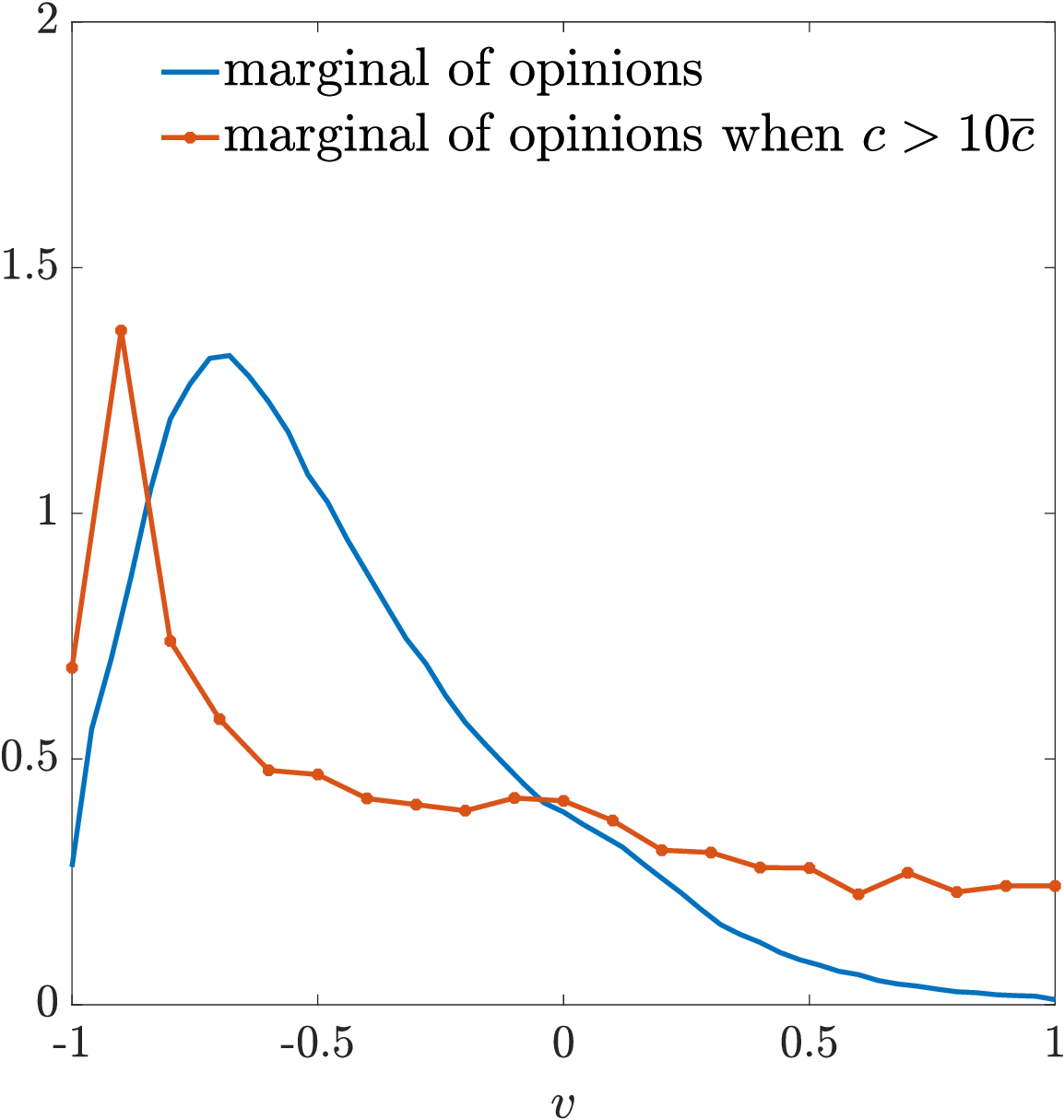}
		\includegraphics[width=0.295\textwidth]{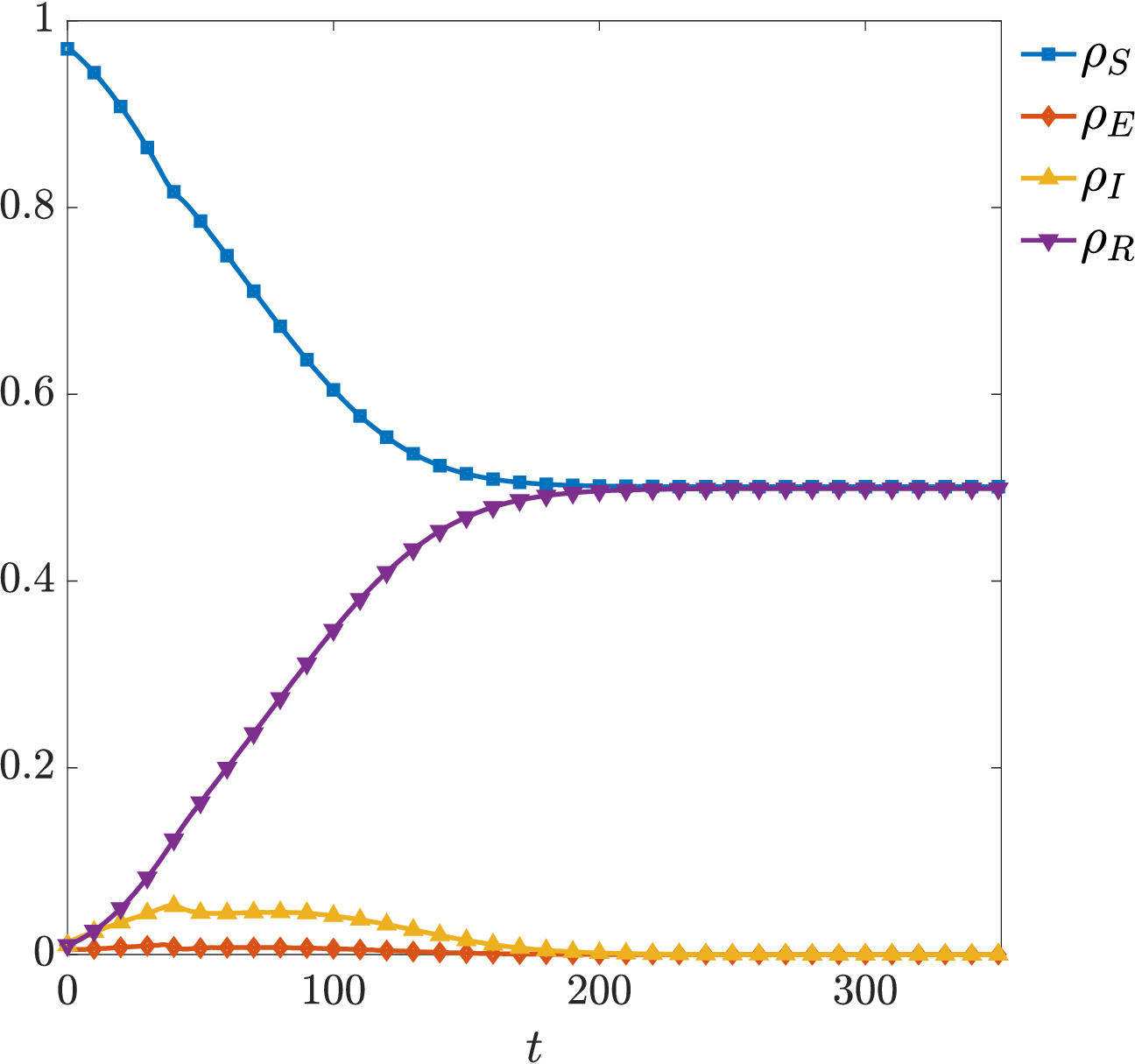}\quad\includegraphics[width=0.31\textwidth]{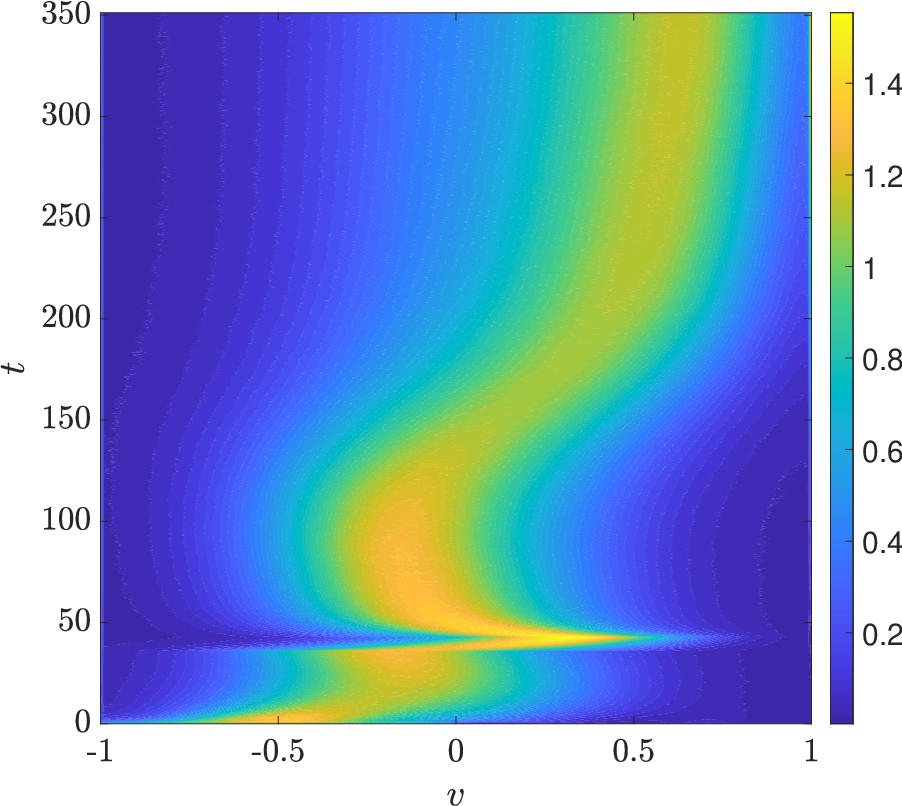}\quad\includegraphics[width=0.255\textwidth]{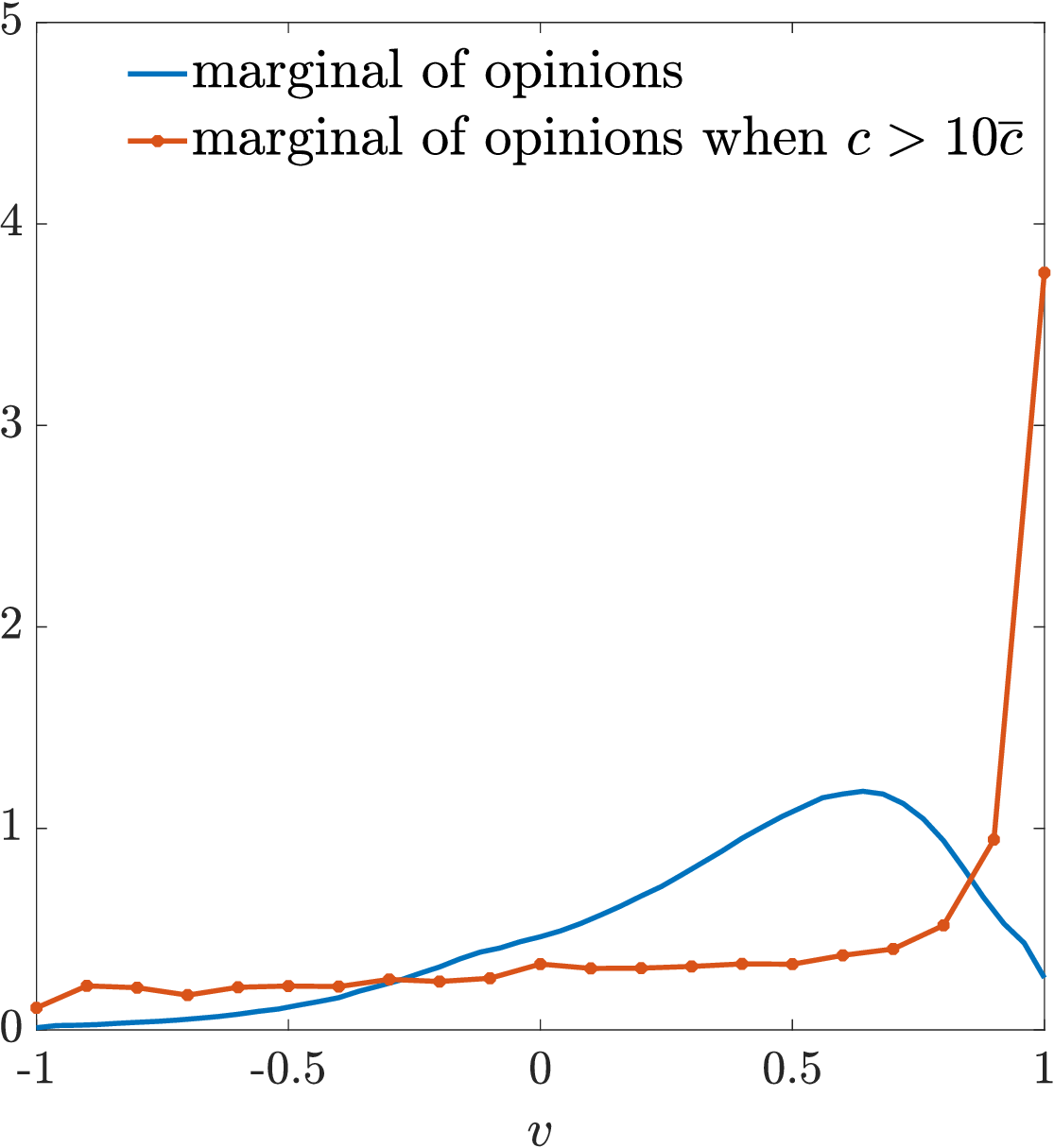}
		\caption{{\em Test 3.} ( $\tau_1 = \tau_2 = 1$). Time evolution of the marginal density of opinions (left) and comparison between the normalized terminal marginal of opinions both for the entire population and for $c>10 \overline c$ (right), for {\em Setting 1} (upper row) and  {\em Setting 2} (bottom row).}
		\label{fig:test2_tau1}
	\end{figure}

	\begin{figure}[h!]
		\centering
		\includegraphics[width=0.27\textwidth]{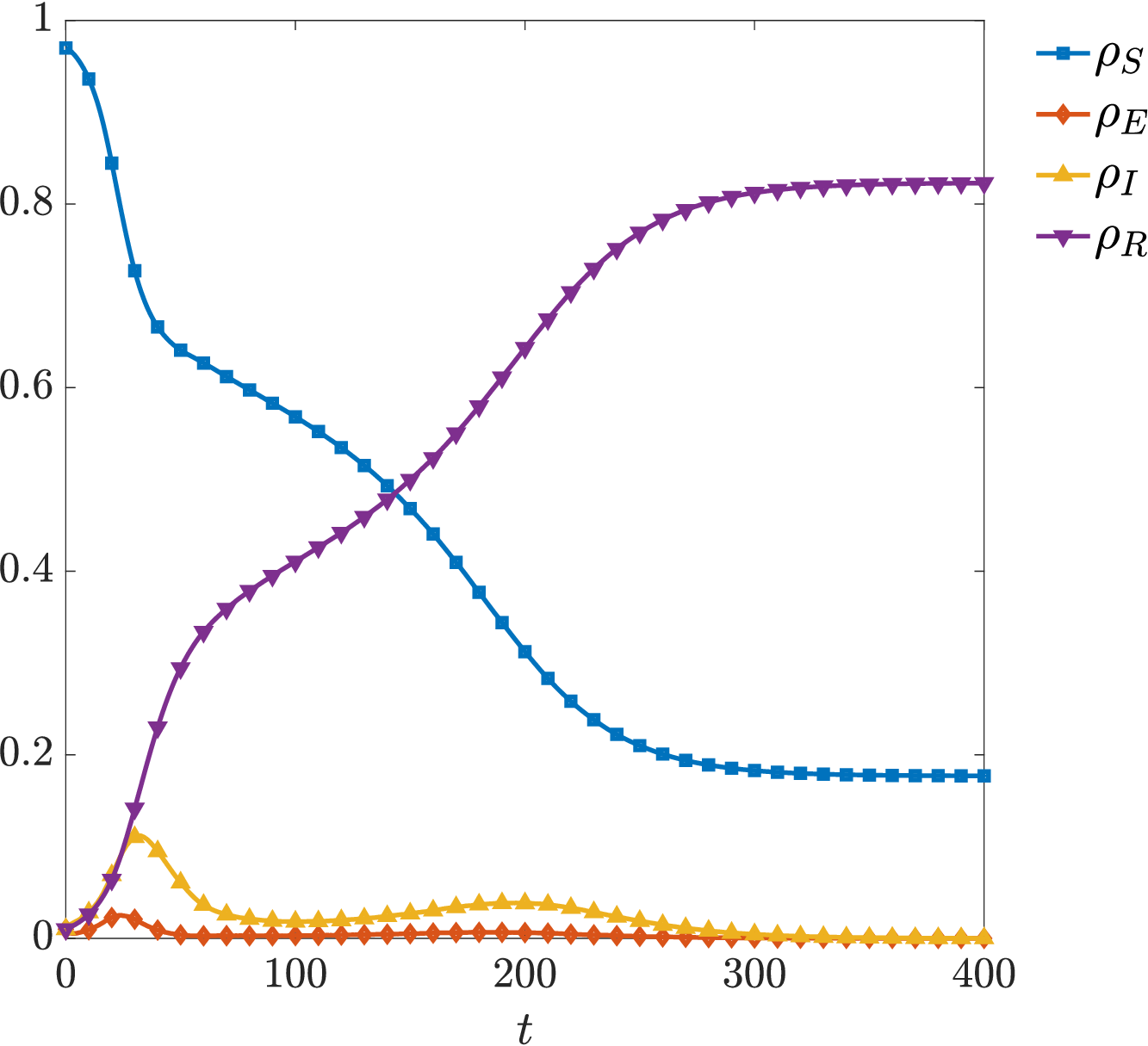}\quad\includegraphics[width=0.275\textwidth]{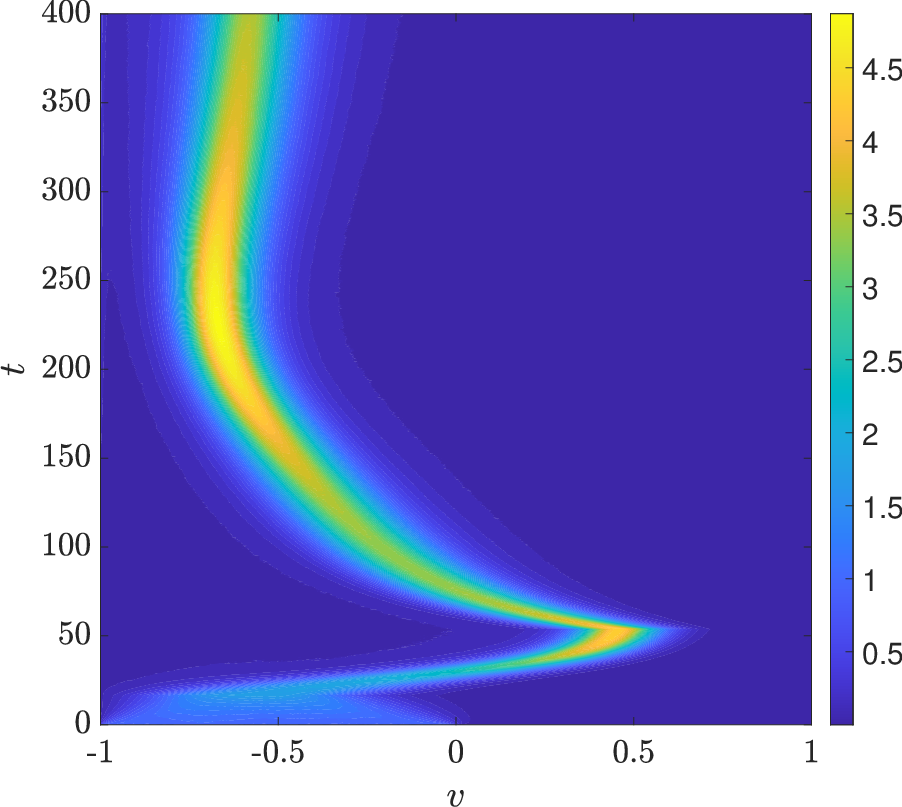}\quad\includegraphics[width=0.225\textwidth]{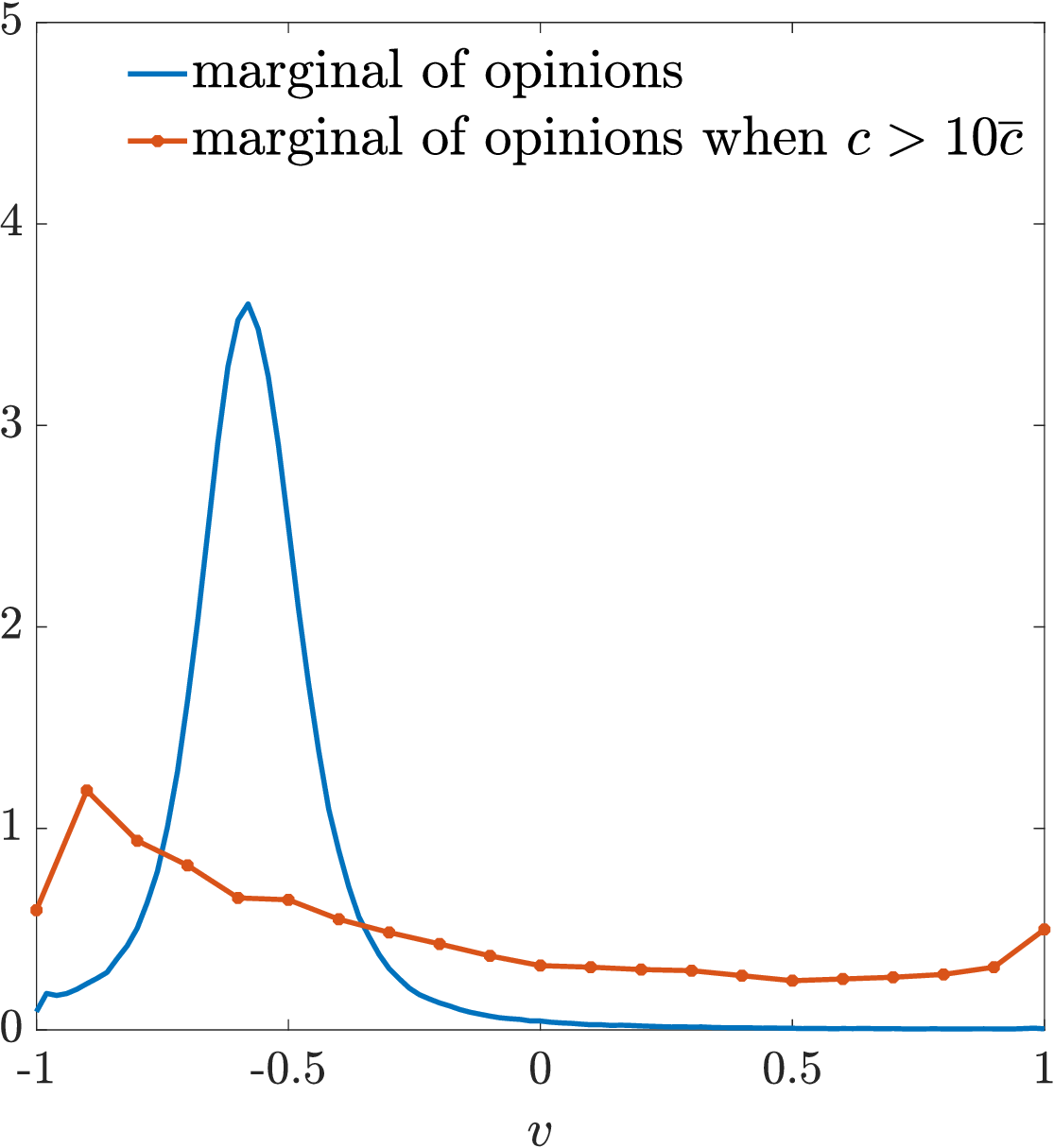}
		\includegraphics[width=0.27\textwidth]{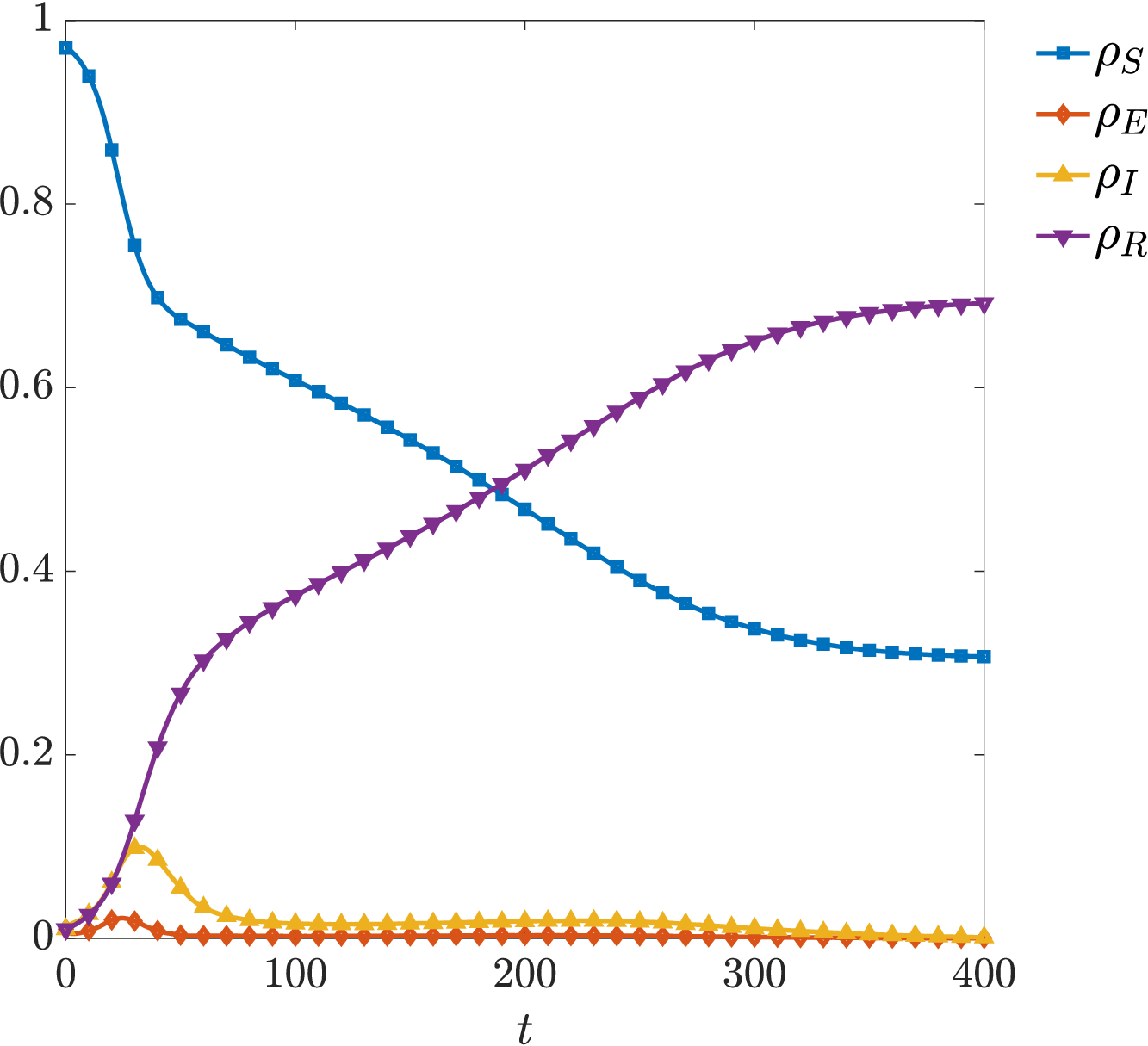}\quad\includegraphics[width=0.275\textwidth]{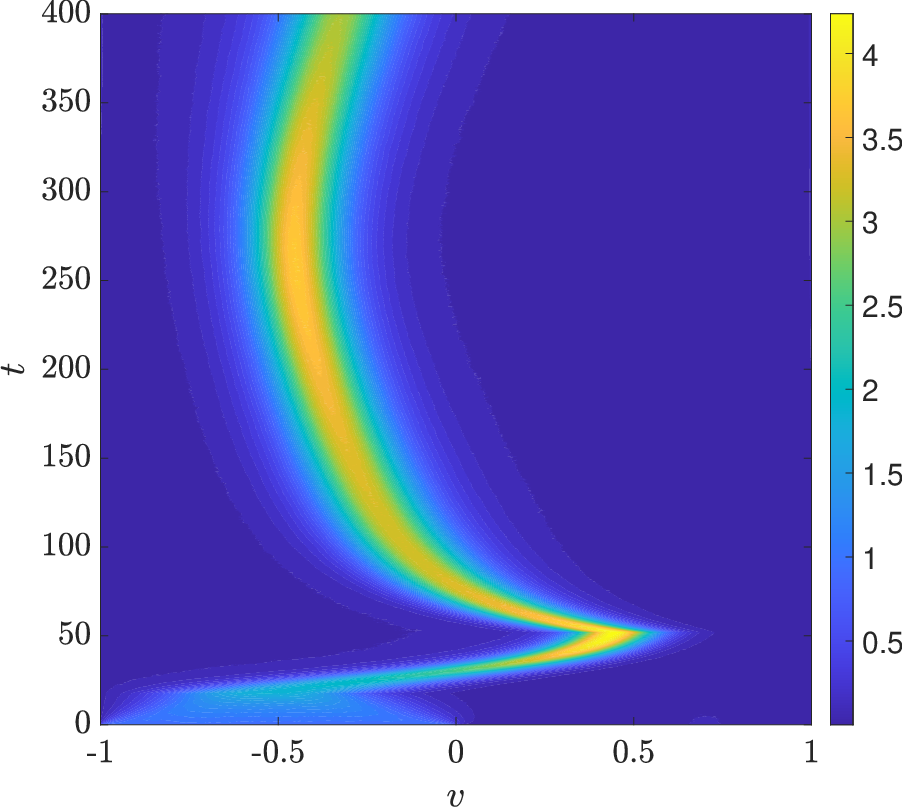}\quad\includegraphics[width=0.235\textwidth]{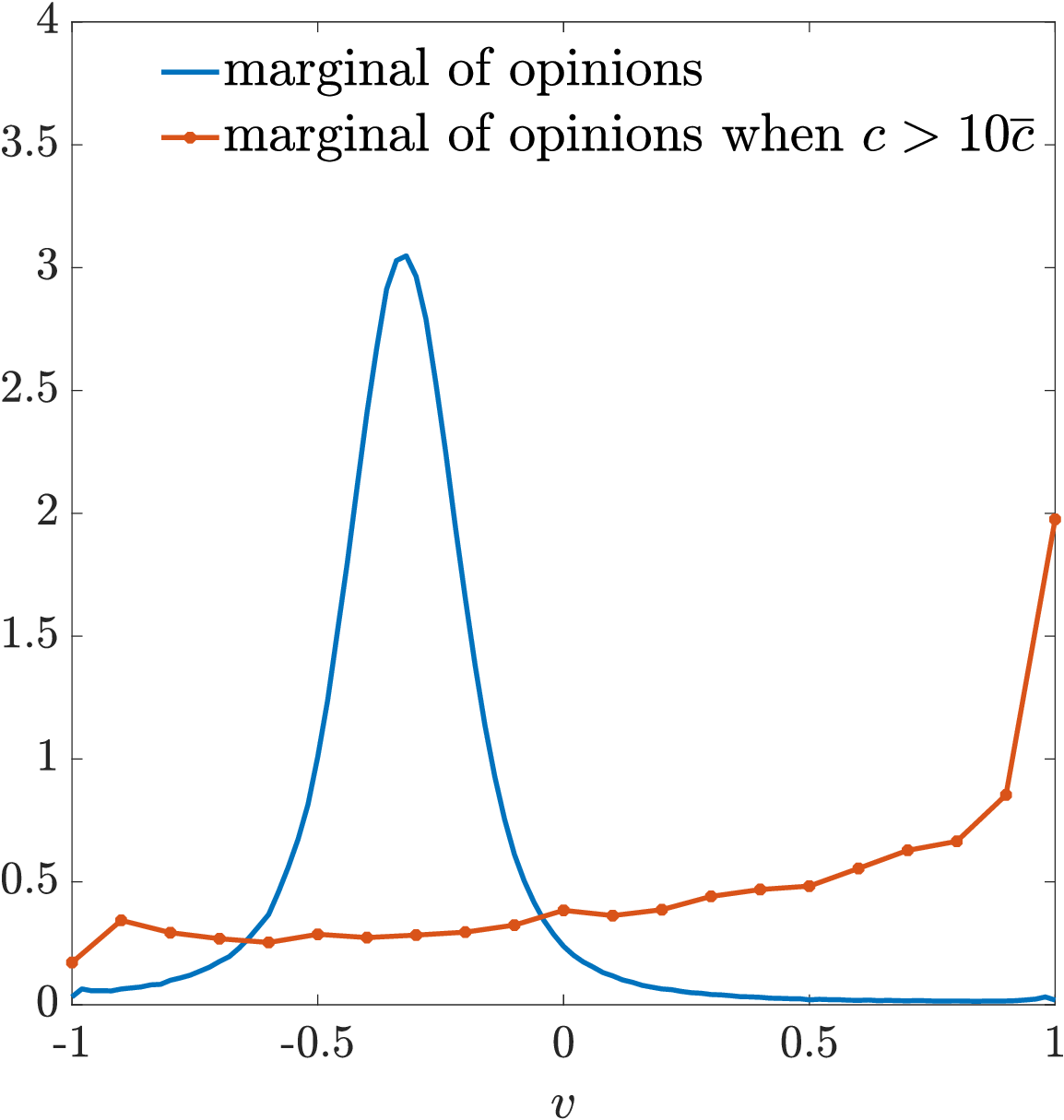}
		\caption{{\em Test 3.}  ($\tau_1 = 10^{-2}$ and $\tau_2 = 10^{-3}$).  Time evolution of the marginal density of opinions (left) and comparison between the normalized terminal marginal of opinions both for the entire population and for $c>10 \overline c$ (right). Upper row refers to {\em Setting 1}, bottom row refers to {\em Setting 2}.}
		\label{fig:test2_tausmall}
	\end{figure}
	
	\begin{figure}[h!]
		\centering
		\includegraphics[width=0.4\textwidth]{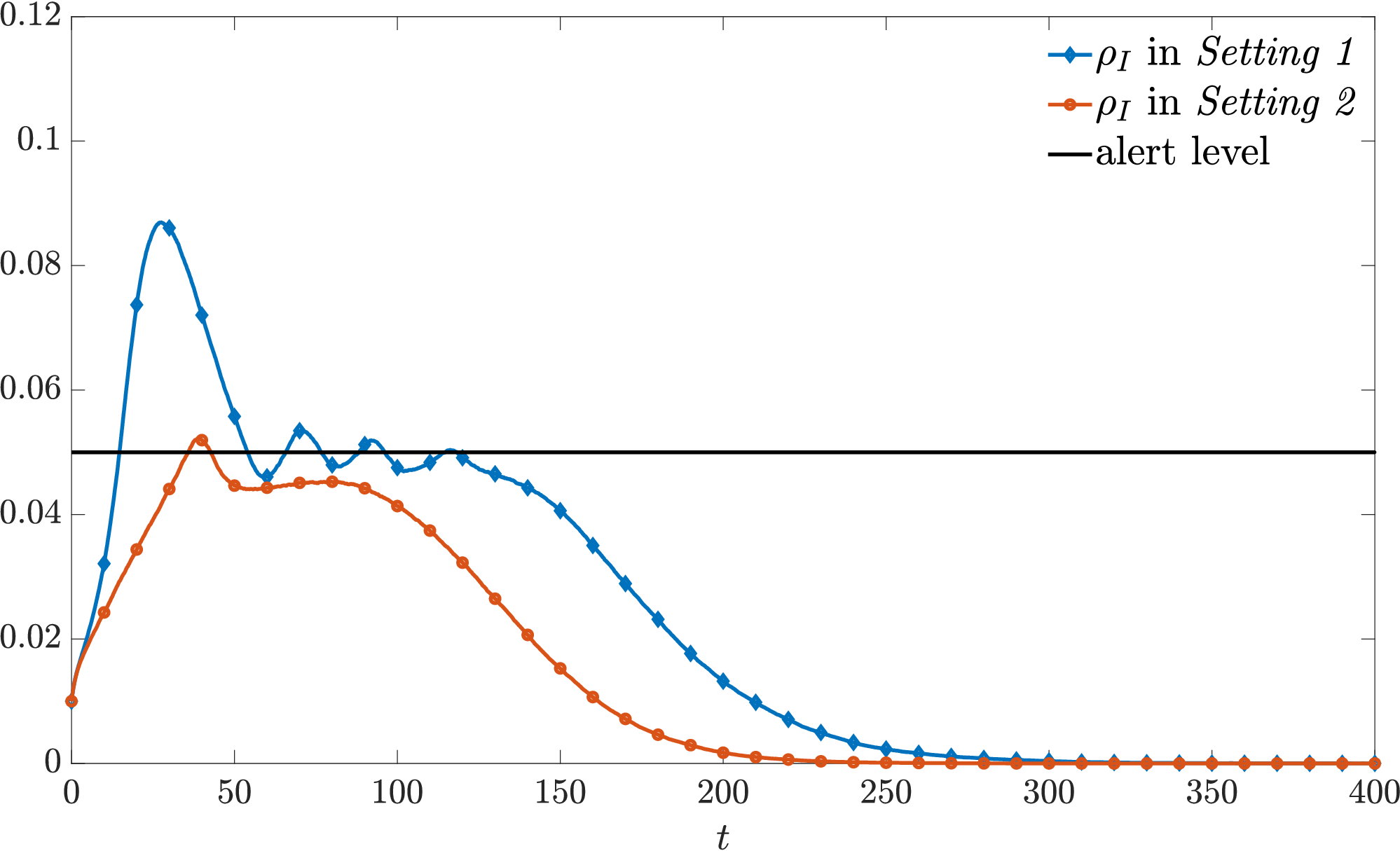}\quad\includegraphics[width=0.4\textwidth]{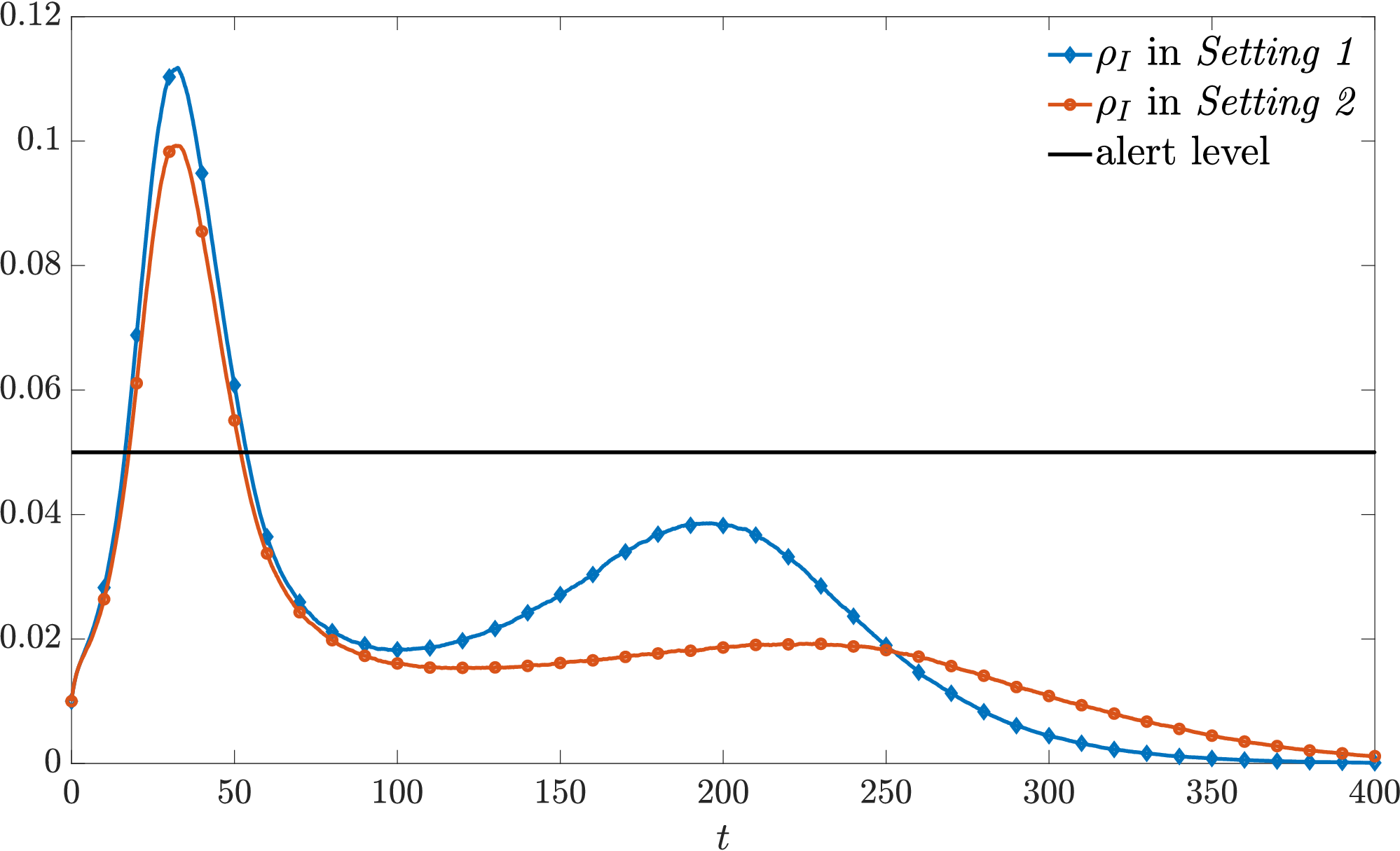}
		\caption{{\em Test 3.} Time evolution of the infected population  $\rho_I(t)$ considering {\em Setting 1}, and {\em Setting 2},  respectively when influencers are against and in favor of protective measures.The black solid line indicates the alert level. On the left, the case with $\tau_1=\tau_2=1$ and on the right, $\tau_1 = 10^{-2}$ and $\tau_2 = 10^{-3}$.}
		\label{fig:test2_confronto}
	\end{figure}
	
	\subsection{Test 4: opinions depending on the trend of the disease}\label{sec:test3}
	In this last set of simulations we want to explore possible effects of the disease status of the individuals on the evolution of the opinions' distribution and on the alert level. As in the previous section, we fix $\rho_E(0) = \rho_I(0) = \rho_R(0) = 10^{-2}$, while $\rho_S(0) = 1- (\rho_E(0) + \rho_I(0) + \rho_R(0))$. 
	
	\subsubsection{Test 4a. The case with $G_J$ dependent on the compartment}
	We assume that the interaction with the background in \eqref{eq:opinion_infected_scaled} depends on the disease status of the agents. In particular we choose $\tilde \alpha_J = 0.5$ for all $J\in \mathcal{C}$ and $\tilde P$ as in \eqref{eq:kernel2}, while the function $G$ is now defined as
	\begin{equation}\label{eq:gtest3a}
		G_J(\rho_I, \rho_R) =\begin{cases} 2\chi\left(\rho_I - \bar \rho> 0\right) - 1, &\mbox{if } J=S,E, \\
			0.7, &\mbox{if } J=I, \\
			-0.7\left(1 - e^{-10\int_0^t \rho_R(\tau) \mathrm{d}\tau}\right), &\mbox{if } J=R.\end{cases}
	\end{equation}
	with alert level $\bar \rho = 0.05$. Here, the choice of $G_J$ reflects the following assumptions: susceptible and exposed individuals behave according to the disease status in terms of number of infected $\rho_I$; infected individuals are assumed to have in general a positive opinion towards protective measures; removed individuals are assumed to have increasingly negative opinions towards protective measures as the disease evolves in time.
	
	As in Section \ref{sec:test2}, we consider {\em Setting 1}, and {\em Setting 2}, and  we simulate  these scenario with parameters $\tau_1 = 10^{-2}$ and $\tau_2 = 10^{-3}$.
	
	Figure \ref{fig:test3a_1} shows the time evolution of the masses of the compartments (left), marginal of the opinions (center) and the normalized terminal marginal of opinions both for the entire population and for $c>10 \overline c$ (right) both for \eqref{eq:init_1} $\Theta = 3/4$ (top row) and  $\Theta = 1/4$ (bottom row). Figure \ref{fig:test3a_confronto} (left) shows the comparison between the evolution of $\rho_I$ for both the initial conditions in {\em Setting 1} and {\em Setting 2}.
	\begin{figure}[h!]
		\centering
		\includegraphics[width=0.27\textwidth]{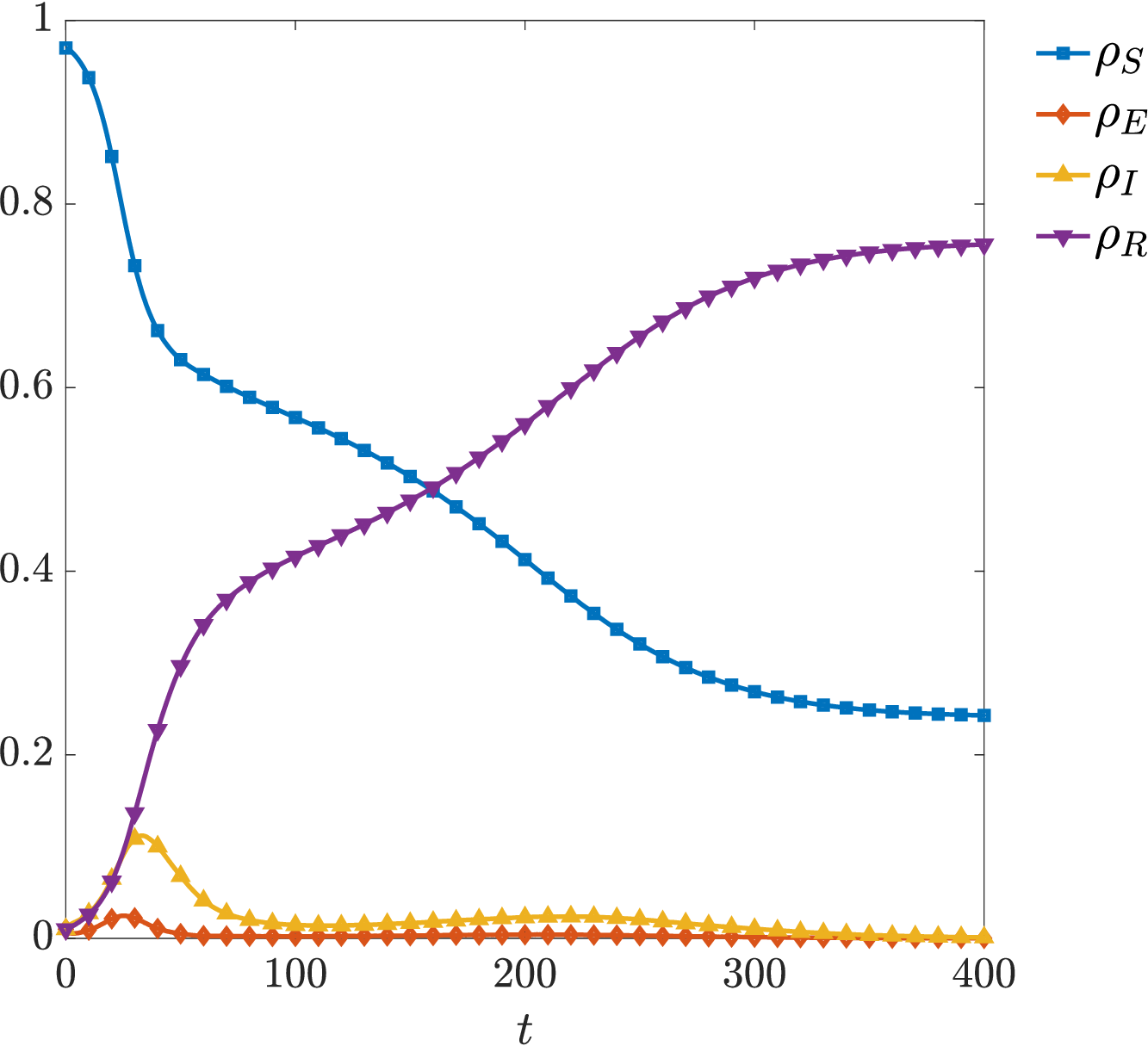}\quad\includegraphics[width=0.275\textwidth]{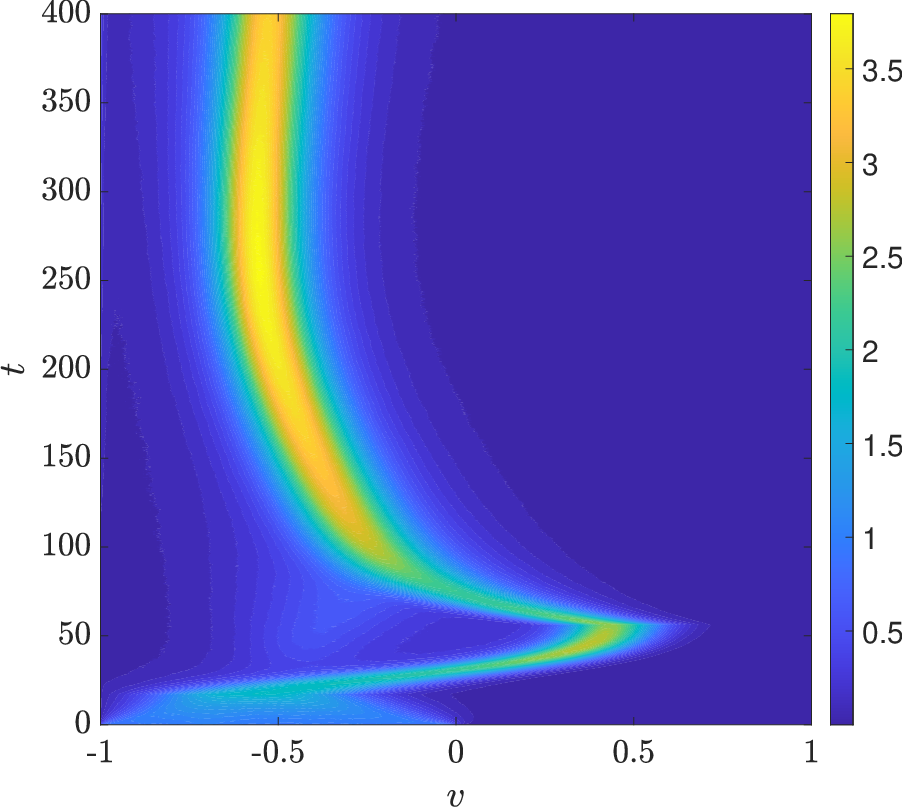}\quad\includegraphics[width=0.235\textwidth]{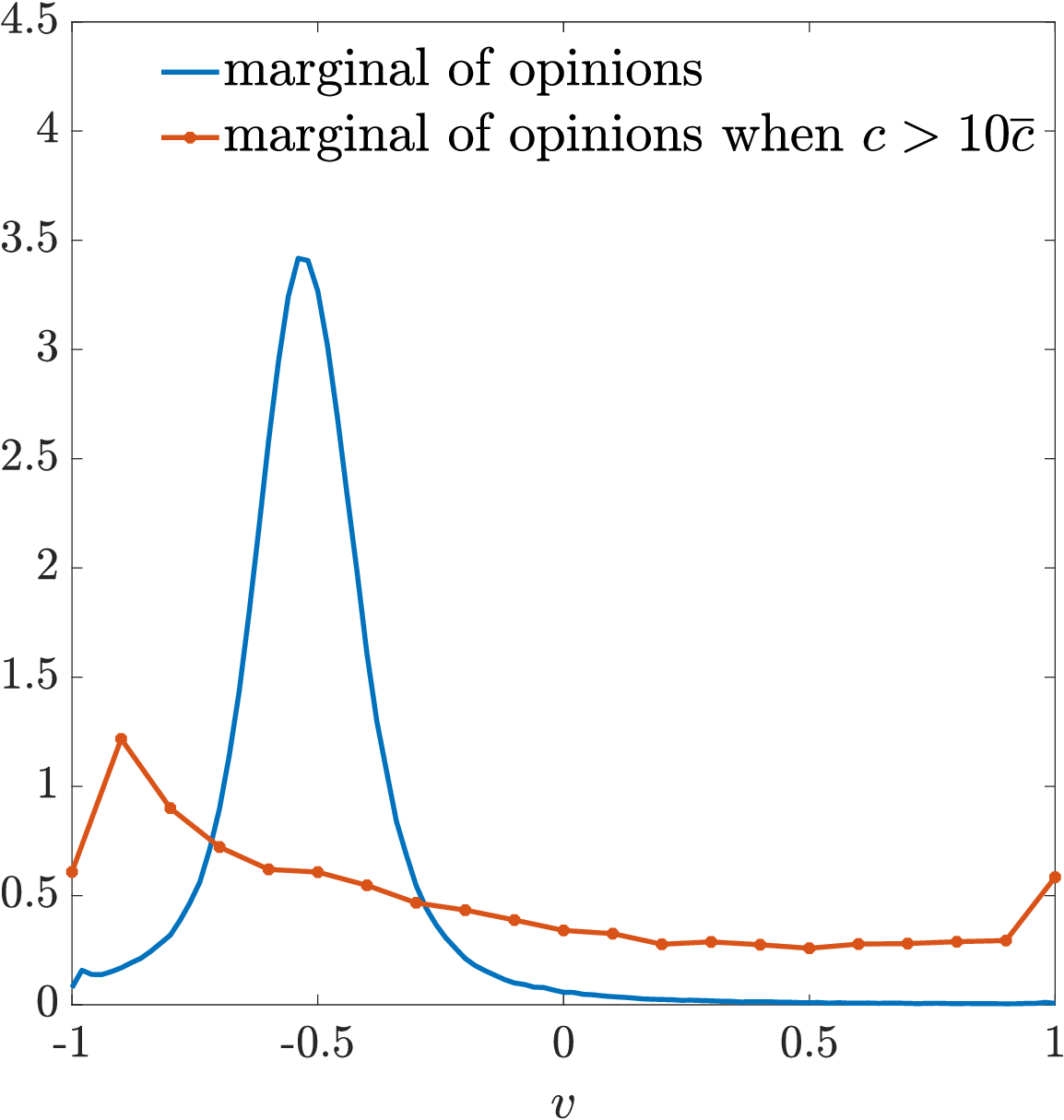}
		\includegraphics[width=0.27\textwidth]{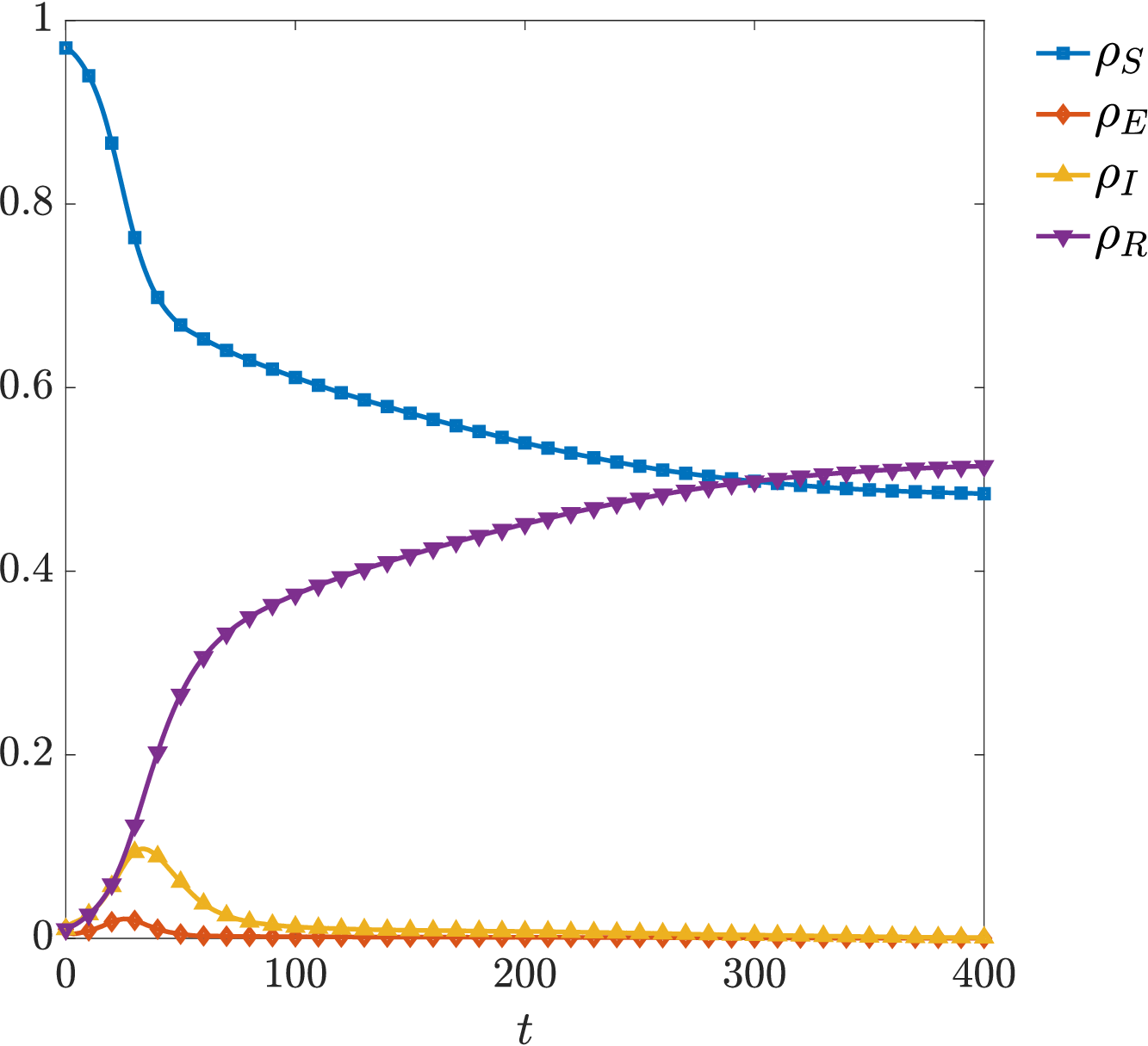}\quad\includegraphics[width=0.275\textwidth]{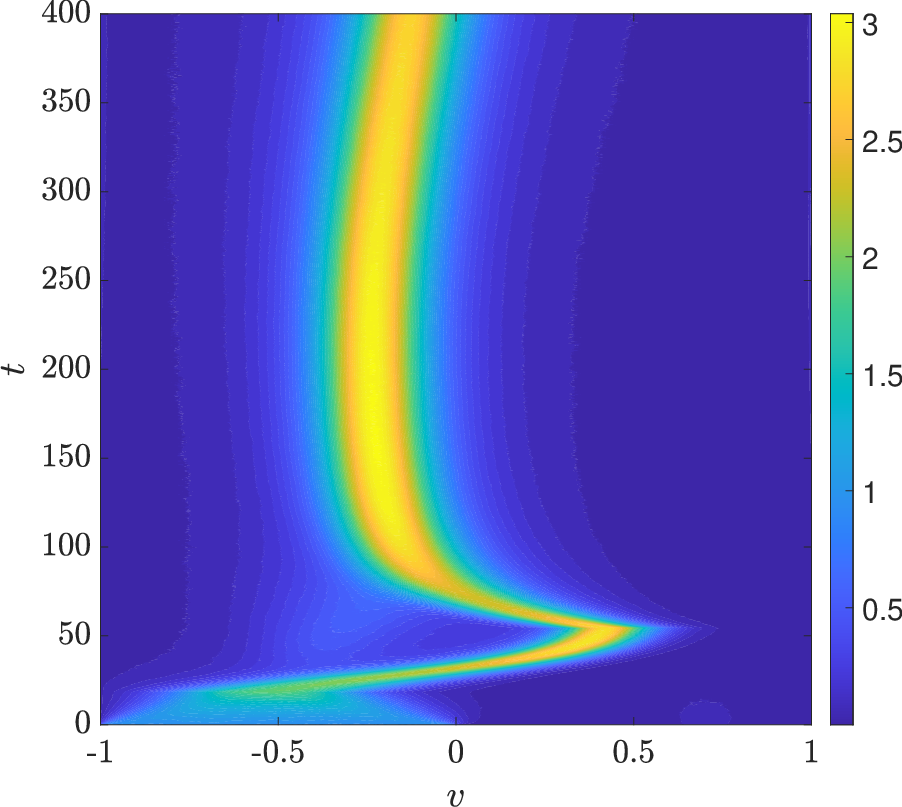}\quad\includegraphics[width=0.235\textwidth]{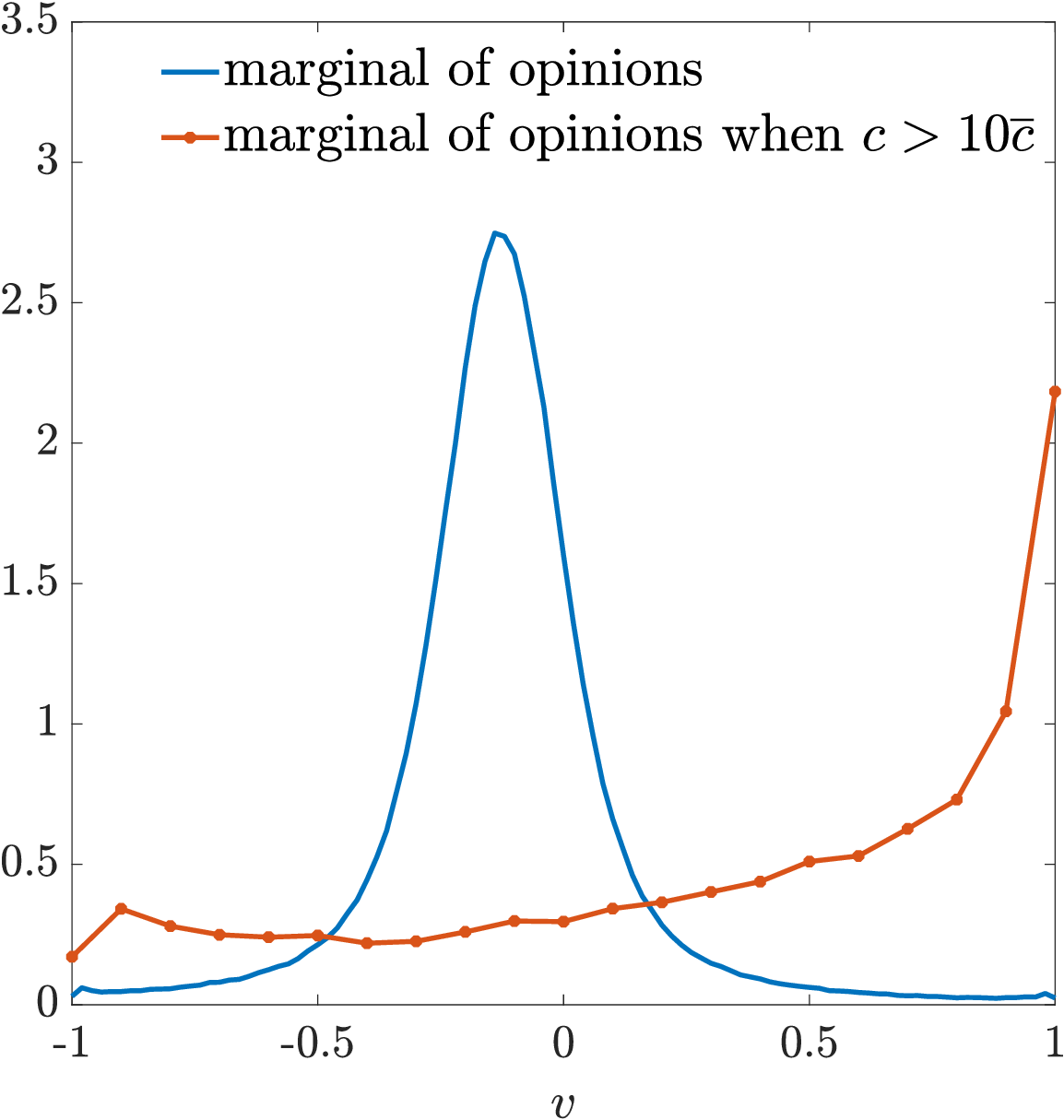}
		\caption{{\em Test 4a.} Time evolution of the masses $\rho_J$ (left), marginal density of opinions (center) and comparison between the normalized terminal marginal of opinions both for the entire population and for $c>10 \overline c$ (right). {\em Setting 1} (upper row), and {\em Setting 2} (bottom row). Here we set $\tau_1 = 10^{-2}$ and $\tau_2 = 10^{-3}$.}
		\label{fig:test3a_1}
	\end{figure}

	\subsubsection{Test 4b. Alert level $\bar \rho$ disease-responsive}
	We now consider an interaction with the background $G$ similar to the one presented in Section \ref{sec:test2}, additionally introducing the time
	\[
	\bar t = \min\{t\in(0,T] \mbox{ s.t. } \rho_I(t) > \tilde \rho\},
	\]
	setting $\tilde \rho = 0.05$ and defining the time dependent alert level $\bar \rho(t)$ 
	\begin{equation}\label{eq:barrhot}
		\overline \rho(t) = \begin{cases} \tilde \rho &\mbox{if } t < \bar t,\\
			\displaystyle \frac{\tilde \rho}{2} &\mbox{if } t \geq \bar t.
		\end{cases}
	\end{equation}
	Hence, the alert level is lowered once the initial threshold is exceeded, in order to mimic a different perception of risk after the disease has spread significantly. The function $G$ is now defined as in \eqref{eq:Gfunction} but with $\overline \rho (t)$ as in \eqref{eq:barrhot} instead of a constant in time $\overline \rho$, 
	and we choose $\tilde \alpha_J = 0.5$ for all $J\in \mathcal{C}$ and $\tilde P$ as in \eqref{eq:kernel2}. 
	As in Section \ref{sec:test2}, we consider {\em Setting 1}, and {\em Setting 2}, and  we simulate  these scenario with parameters $\tau_1 = 10^{-2}$ and $\tau_2 = 10^{-3}$.
	
	Figure \ref{fig:test3b_1} shows the time evolution of the masses of the compartments (left), marginal of the opinions (center) and the normalized terminal marginal of opinions both for the entire population and for $c>10 \overline c$ (right) for {\em Setting 1} (top row) and {\em Setting 2} (bottom row).  Figure \ref{fig:test3a_confronto} (right) shows the comparison between the evolution of $\rho_I$ for both the initial conditions {\em Setting 1} and {\em Setting 2}. At the final time $T=800$, for initial condition \eqref{eq:init_1}, we can see a resurgence of the epidemics.
	\begin{figure}[h!]
		\centering
		\includegraphics[width=0.27\textwidth]{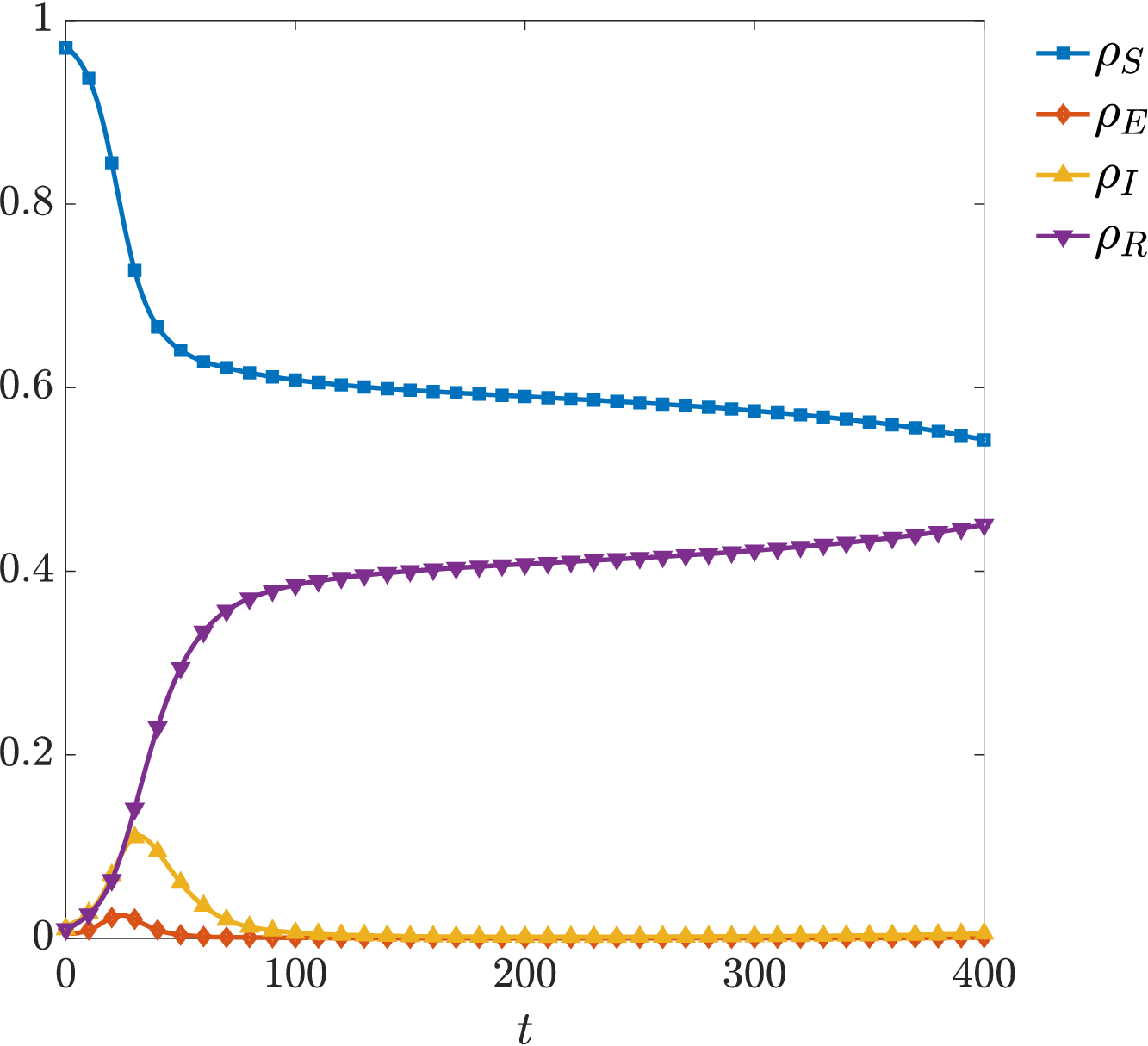}\quad\includegraphics[width=0.275\textwidth]{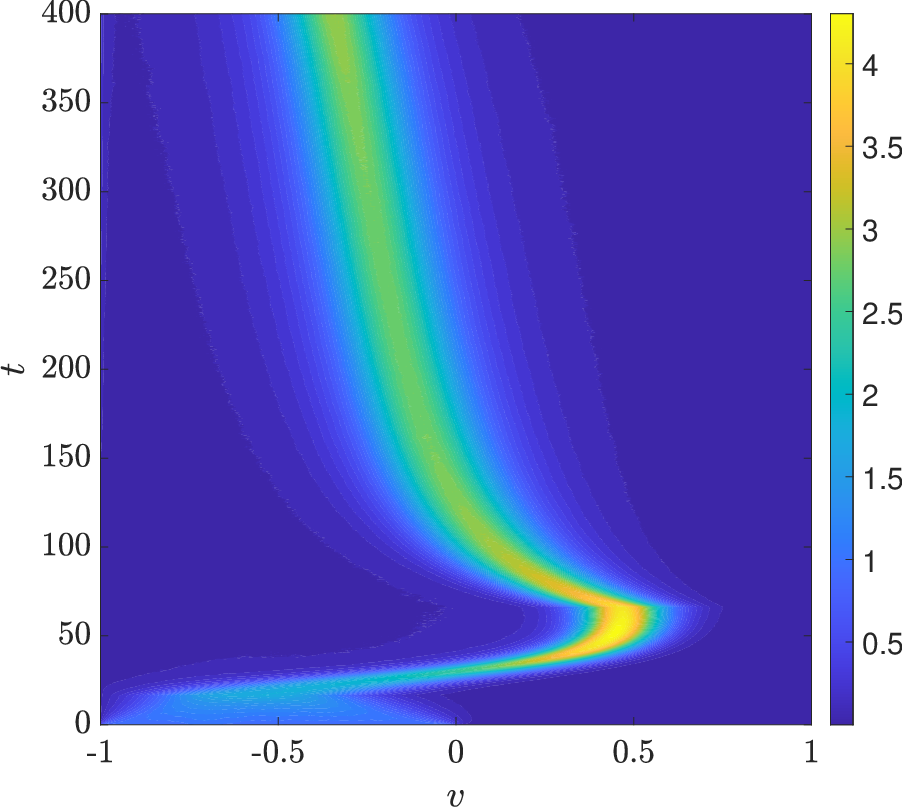}\quad\includegraphics[width=0.235\textwidth]{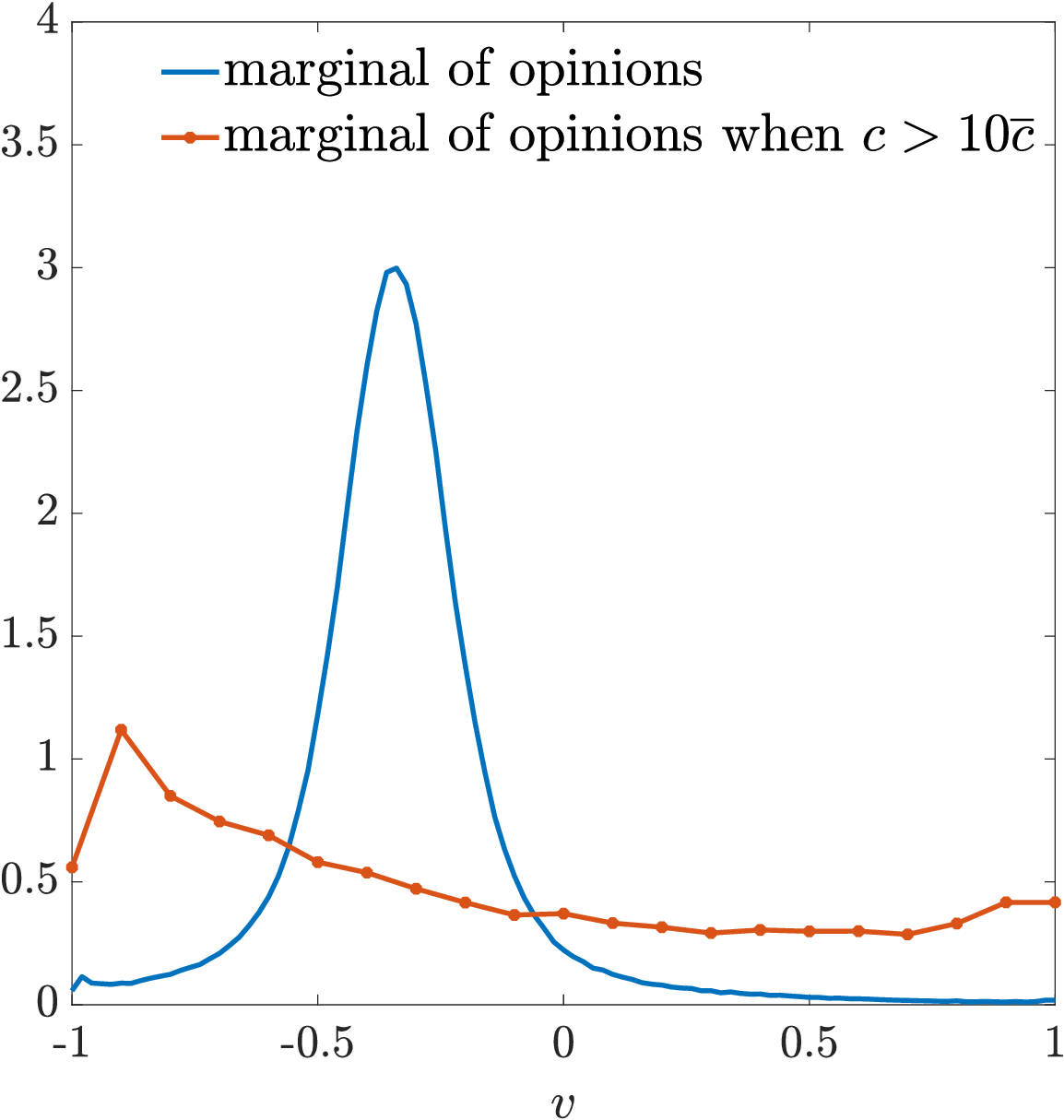}
		\includegraphics[width=0.27\textwidth]{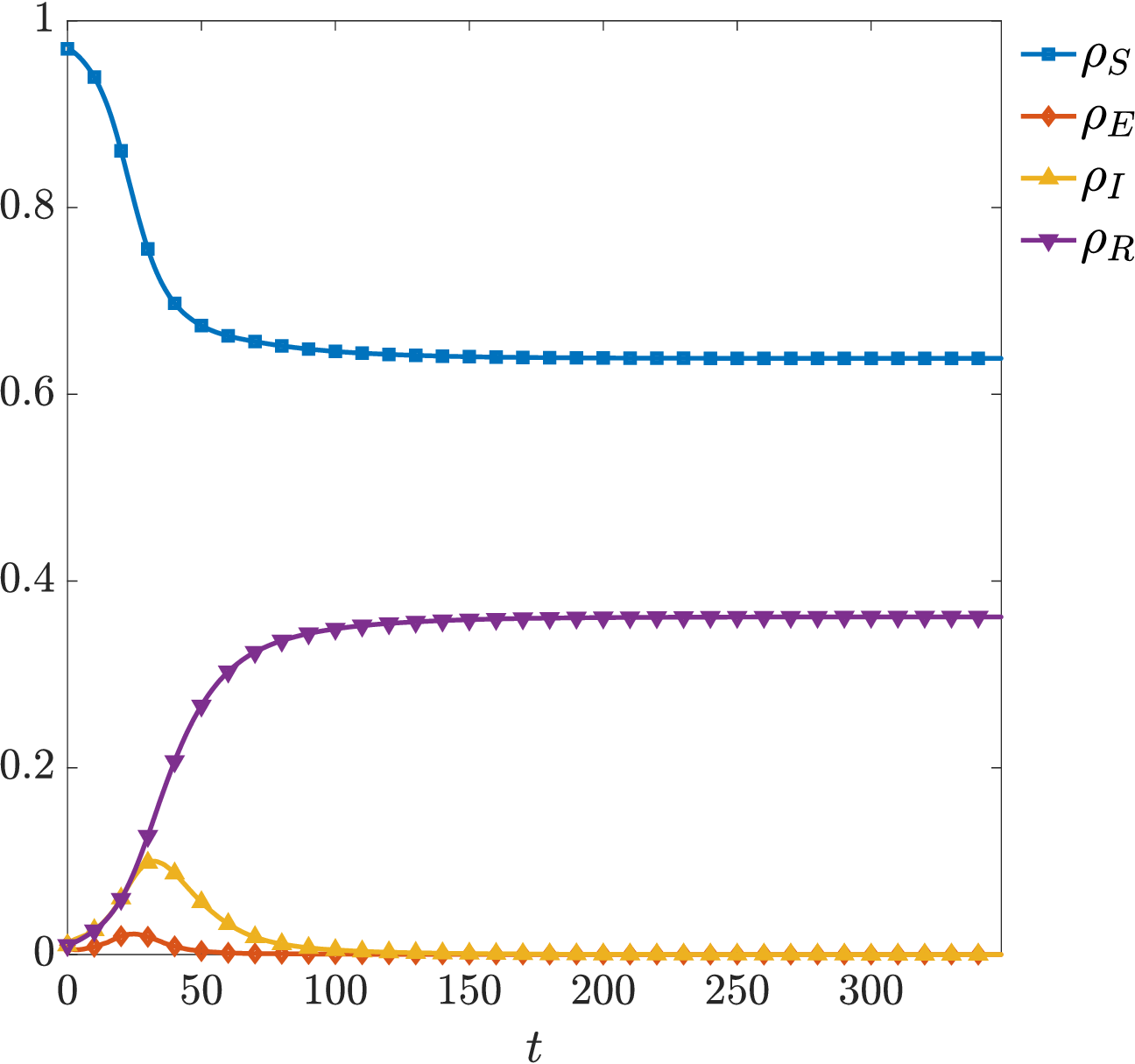}\quad\includegraphics[width=0.275\textwidth]{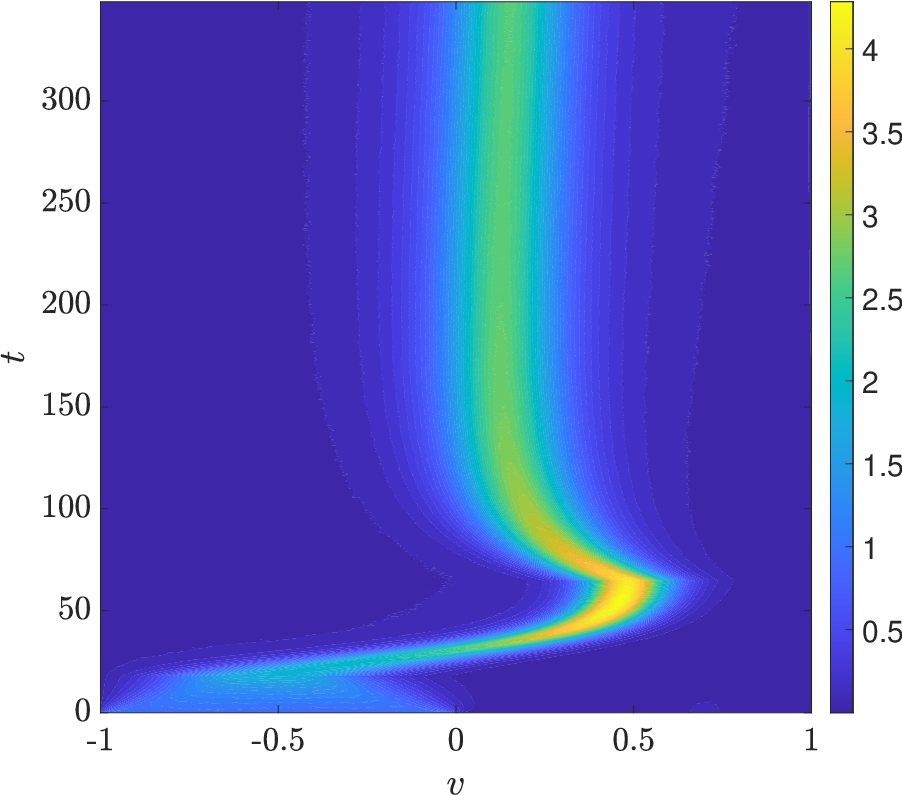}\quad\includegraphics[width=0.235\textwidth]{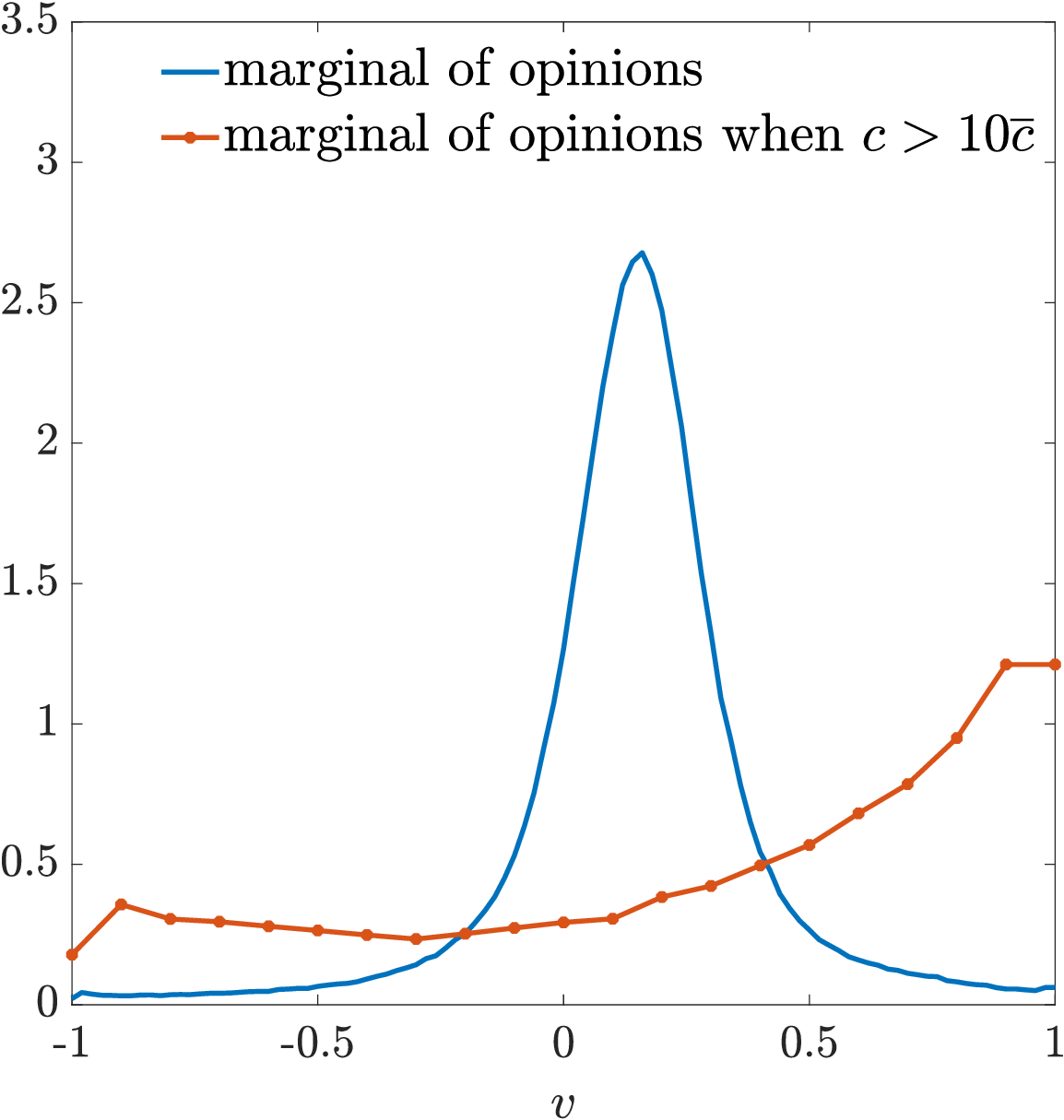}
		\caption{{\em Test 4b}.Time evolution of the masses $\rho_J$ (left), marginal density of opinions (center) and comparison between the normalized terminal marginal of opinions both for the entire population and for $c>10 \overline c$ (right),  for {\em Setting 1} (upper row) and {\em Setting 2} (bottom row). Here we set $\tau_1 = 10^{-2}$ and $\tau_2 = 10^{-3}$.}
		\label{fig:test3b_1}
	\end{figure}
	
	\begin{figure}[h!]
		\centering
		\includegraphics[width=0.4\textwidth]{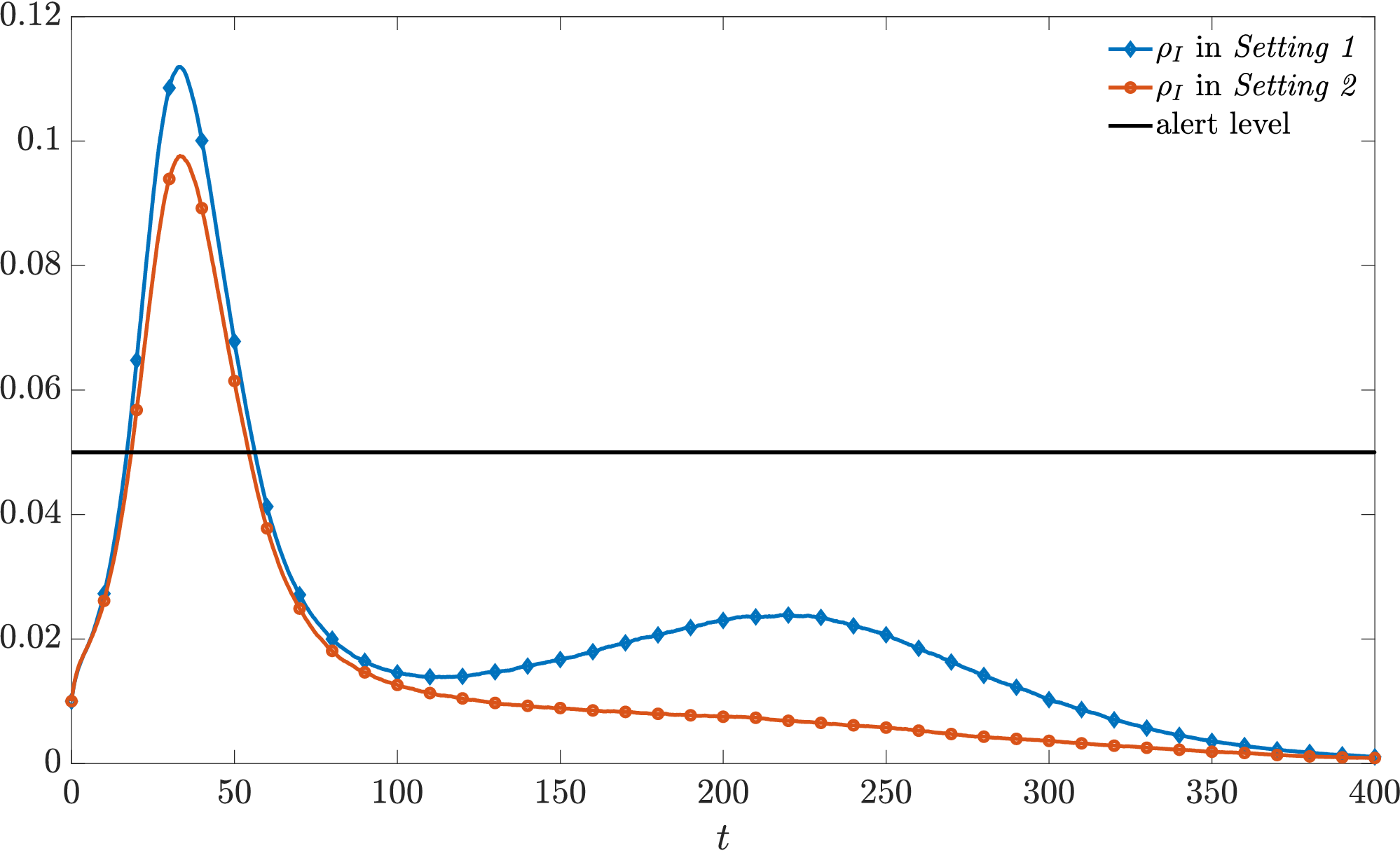} \quad
		\includegraphics[width=0.4\textwidth]{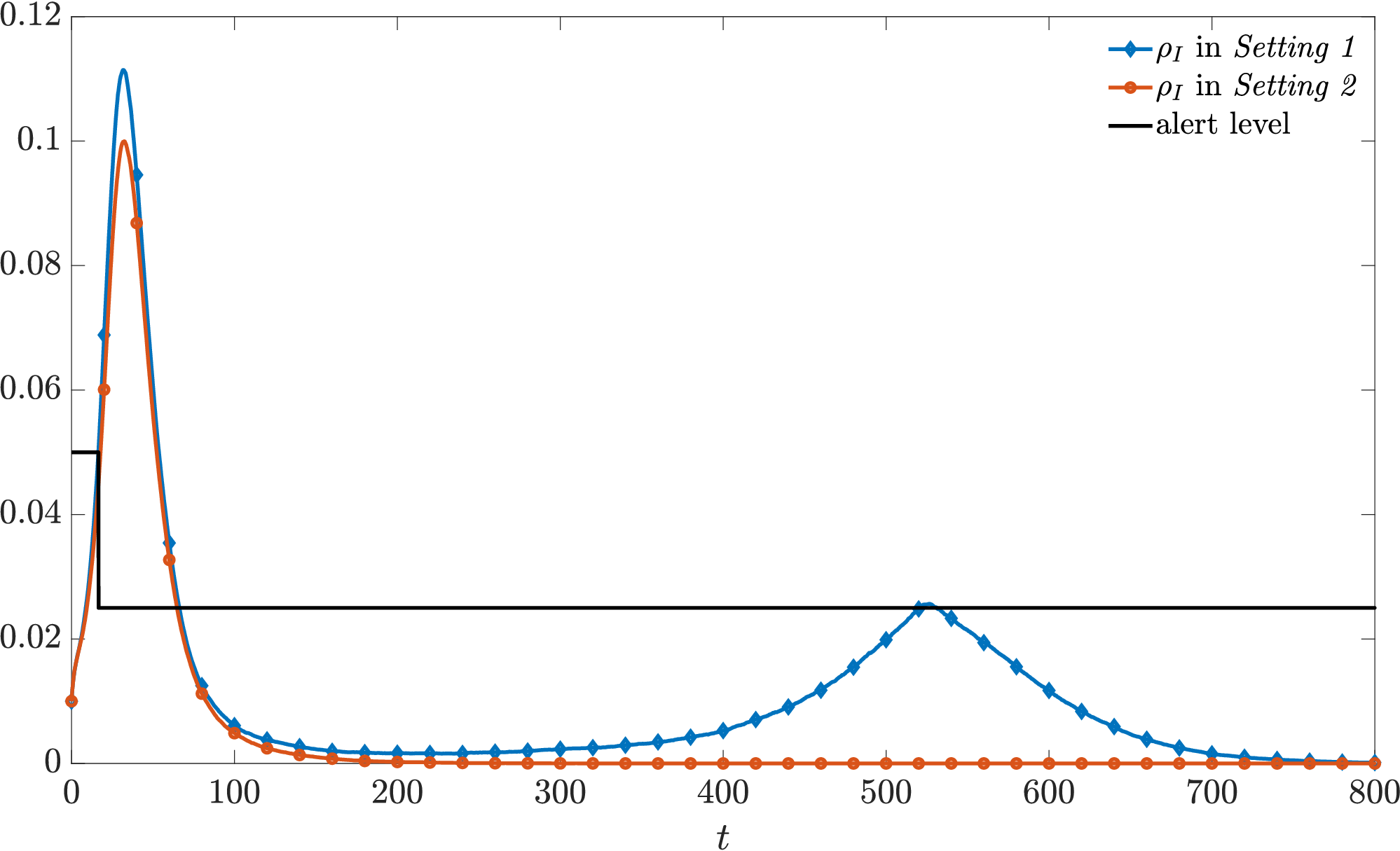}
		\caption{{\em Test 4a (left) and Test 4b (right)}. Time evolution of the infected population  $\rho_I(t)$ considering {\em Setting 1}, and {\em Setting 2}, respectively when influencers are against and in favor of protective measures. Here we set $\tau_1 = 10^{-2}$ and $\tau_2 = 10^{-3}$. The black solid line shows the value of the alert level. 
		}
		\label{fig:test3a_confronto}
	\end{figure}

	\section*{Conclusions}
	In this work, we have introduced a novel approach that integrates traditional epidemiological models with social behavior and opinion dynamics. This integration allows for a more comprehensive understanding of how diseases spread in the context of modern information dissemination in the presence of a static contact structure. More in details, we have extended a classical compartmental SEIR (Susceptible-Exposed-Infectious-Recovered) model to include opinions on protective measures influenced by social media and influential individuals (influencers). The new proposed model incorporates a kinetic approach to describe how opinions about protective measures form and evolve within a population. This includes the influence of social networks and key figures on public perception and behavior.
	
	In the second part, the numerical results demonstrate that social media opinions and influencers can significantly alter the trajectory of an epidemic. Positive influence can enhance the adoption of protective measures, while negative influence can lead to resistance and increased spread of the disease. This new model captures the occurrence of infection waves, which are influenced by the dynamics of social opinions and behaviors. This reflects real-world phenomena observed for example during the COVID-19 pandemic.
	
	This integrated model provides insights that can help policymakers and public health officials design more effective interventions. By understanding the role of social behavior, strategies can be tailored to improve public compliance with protective measures. In the future, we would like to explore additional factors, such as the role of physical contacts and  real social network structures extrapolated, to further enhance the model's accuracy and effectivity.

	\section*{Acknowledgments}
	This work has been written within the activities of the GNCS and GNFM groups of INdAM. G.A. and E.C. acknowledge the support of the Italian Ministry of University and Research (MUR) through the MUR-PRIN Project 2022 PNRR (No. P2022JC95T) “Data-driven discovery and control of multi-scale interacting artificial agent systems”, financed by the European Union - Next Generation EU. G.A.  acknowledges MUR-PRIN Project 2022 No.
	2022N9BM3N “Efficient numerical schemes and optimal control methods for time-dependent
	partial differential equations” financed by the European Union - Next Generation EU. G.D. and M.Z. acknowledge the support of MUR-PRIN 2020 project (No. 2020JLWP23) “Integrated Mathematical Approaches to Socio–Epidemiological Dynamics”. M.Z. acknowledges the support of the ICSC – Centro Nazionale di Ricerca in High Performance Computing, Big Data and Quantum Computing, funded by European Union - NextGeneration EU.

\bibliographystyle{abbrv}
\bibliography{biblio_SEIR_new}

\end{document}